\documentclass{amsart}
\usepackage{amsmath, amssymb, amsthm}
\usepackage[authoryear, round]{natbib}
\usepackage{dsfont} % indicator functions
\usepackage{soul} % highlighting etc
\usepackage{url}
\usepackage{xcolor}

\usepackage{enumitem}
\usepackage{hyperref}

\newtheorem{thm}{Theorem}
\newtheorem{example}[thm]{Example}
\newtheorem{lemma}[thm]{Lemma}
\newtheorem{remark}[thm]{Remark}
\newtheorem{proposition}[thm]{Proposition}
\newtheorem{corollary}[thm]{Corollary}
\newtheorem{assumption}[thm]{Assumption}

\usepackage[plain,noend]{algorithm2e}
\usepackage{algorithmic}

\newcommand{\R}{\mathbb{R}}
\newcommand{\eqdef}{:=}
\newcommand{\coloneq}{:=}
\newcommand{\OO}{\mathcal{O}}

\newcommand{\ind}[1]{\mathds{1}\{#1\}}
\newcommand{\Ltwo}{\mathcal{L}_2}

\newcommand{\E}{\mathbb{E}} % Expectation
\newcommand{\Exp}{\E} % Expectation (bis)
\newcommand{\pr}{\mathbb{P}}
\newcommand{\var}{\operatorname{var}}
\newcommand{\cov}{\operatorname{cov}}
\newcommand{\cvprob}{\stackrel{P}{\rightarrow}}
\DeclareMathOperator{\wnv}{\textrm{WNV}}
\DeclareMathOperator{\mse}{\textrm{MSE}}

\newcommand{\real}{\mathbb{R}}
\newcommand{\N}{\mathcal{N}}
\newcommand{\rmd}{\mathrm{d}}

\newcommand{\sumest}{\hat{f}} % sum estimator (generic)
\newcommand{\sumestS}{\sumest^{\mathrm{S}}} % sum est (simple)
\newcommand{\sumestC}{\sumest^{\mathrm{C}}} % sum est (cycle)

\newcommand{\Urk}{U_{r, k}}
\newcommand{\Urll}{U_{r, l}}
\newcommand{\URk}{U_{R,k}}
\newcommand{\URll}{U_{R,l}}
\newcommand{\UrkS}{\Urk^{\mathrm{S}}}
\newcommand{\UrlS}{\Urll^{\mathrm{S}}}
\newcommand{\UrkC}{\Urk^{\mathrm{C}}}
\newcommand{\UrlC}{\Urll^{\mathrm{C}}}
\newcommand{\URkS}{\URk^{\mathrm{S}}}

\newcommand{\Geom}{\mathrm{Geometric}}

\newcommand{\Geomp}{\mathrm{Geometric}(p)}
\newcommand{\nablasumest}{\widehat{\nabla_\theta f}}
\newcommand{\nablasumestS}{\nablasumest^{\mathrm{S}}}
\newcommand{\nablasumestC}{\nablasumest^{\mathrm{C}}}
\newcommand{\Wrk}{W_{r, k}}
\newcommand{\Wrl}{W_{r, l}}
\newcommand{\WRk}{W_{R,k}}
\newcommand{\WrkS}{\Wrk^{\mathrm{S}}}
\newcommand{\WrlS}{\Wrl^{\mathrm{S}}}
\newcommand{\WrkC}{\Wrk^{\mathrm{C}}}
\newcommand{\WrlC}{\Wrl^{\mathrm{C}}}
\newcommand{\ess}{\textrm{ESS}}

\usepackage[textwidth=1.8cm, textsize=scriptsize]{todonotes}
\usepackage{setspace} % only used for comments

\graphicspath{{Images_paper/}}
\title[Unbiased Monte Carlo estimation of smooth functions]{Towards a turnkey
approach to unbiased Monte Carlo estimation of smooth functions of expectations}

\author{Nicolas Chopin}
\address{ENSAE, Institut Polytechnique de Paris, France}
\email{nicolas.chopin@ensae.fr}

\author{Francesca R. Crucinio} 
\address{ESOMAS, Università di Torino, Italy}
\email{francescaromana.crucinio@unito.it}

\author{Sumeetpal S. Singh}
\address{NIASRA and School of Mathematics and Applied Statistics, University of Wollongong, Australia}
\email{sumeetpals@uow.edu.au}

\begin{document}

\maketitle

\begin{abstract}
    Given a smooth function $f$, we develop a general approach to turn Monte
    Carlo samples with expectation $m$ into an unbiased estimate of $f(m)$.
    Specifically, we develop estimators that are based on randomly truncating
    the Taylor series expansion of $f$ and estimating the coefficients of the
    truncated series. We derive their properties and propose a strategy to set
    their tuning parameters -- which depend on $m$ -- automatically, with a
    view to make the whole approach simple to use. We develop our methods for
    the specific functions $f(x)=\log x$ and $f(x)=1/x$, as they arise in
    several statistical applications such as maximum likelihood estimation of
    latent variable models and Bayesian inference for un-normalised models.
    Detailed numerical studies are performed for a range of applications to
    determine how competitive and reliable the proposed approach is.
\end{abstract}

\section{Introduction}

\subsection{General setting}

We consider the problem of unbiasedly estimating $f(m)$, for a smooth function
$f$, using independent and identically distributed (IID) random variables $X_1,
	X_2, \ldots$ with mean
$m\in\R$.
In particular, using the Taylor expansion of $f$ around a certain $x_0\neq 0$,
\begin{equation}\label{eq:taylor}
	f(m) = \sum_{k=0}^{\infty} \frac{f^{(k)}(x_0)}{k!} \left( m - x_0 \right)^k
	= \sum_{k=0}^{\infty} \gamma_k \left( \frac{m}{x_0} - 1 \right)^k,\quad
	\gamma_k \eqdef  \frac{f^{(k)}(x_0)}{k!} x_0^k,
\end{equation}
we will employ the so-called `sum' estimator
\citep{mcleish, glynn2014exact,rhee2015unbiased}
\begin{equation} \label{eq:def_sumest}
	\sumest \eqdef
	\sum_{k=0}^R \frac{\gamma_k U_{R,k}}{\pr(R\geq k)} =
	\sum_{k=0}^\infty \frac{\ind{R\geq k}}{\pr(R\geq k)}\gamma_k \URk,
\end{equation}
where $R$ is an integer-valued random variable, 
\[\URk=\sum_{r=k}^\infty \Urk \ind{R=r},\]
and $\Urk$ is an unbiased estimate of $(m/x_0 - 1)^k$, as clarified in the following
assumption.

\begin{assumption} \label{ass:condition_urk}
Let $R$ be a random variable taking values in $\{0, 1, \ldots\}$  
such that $\pr(R\geq k)>0$ for all $k$. Let $\{\Urk: k,r\in \mathbb{N}; k \leq r\}$ be a collection of random variables,
which are independent of $R$, and 
\begin{equation}\label{eq:condition_urk}
	\E\left[ \Urk \right] = \left( \frac{m}{x_0} - 1\right)^k.%, \quad 0 < k \leq r.
\end{equation}
Let $U_{0,0}=U_{1,0}= \cdots = 1.$ Furthermore, we assume that $a_k\eqdef
\sup_{r\geq k}\E[ | \Urk | ] < \infty$ for $k >0$, and
$\sum_{k=0}^{\infty}|\gamma_{k}|a_{k}<\infty$.
\end{assumption}

These conditions ensure that $\sumest$ is an unbiased estimator, 
see the following lemma and Appendix~\ref{app:integrability} for a proof. 
\begin{lemma}
Under Assumption~\ref{ass:condition_urk}, $\sumest$ is integrable and $\E\left[\sumest
\right]=f(m)$ in \eqref{eq:taylor}.
\label{lem:integrablefhat}
\end{lemma}

A candidate for $\Urk$  is the following construction of $\Urk$ using
$X_1,\ldots, X_k$ only,
\begin{equation}\Urk = \prod_{i=1}^k (X_i/x_0 - 1), \quad 0 < k \leq r.
	\label{eq:simple}
\end{equation} Thus, $\Urk $ unbiasedly estimates $(m/x_0 - 1)^k$ using the
first $k$ univariates only. In Section \ref{subsec:cycle}, we propose a
lower-variance estimator based on the whole sample $X_1,\ldots, X_r$. Including
the factor $x_0^k$ in the definition of $\gamma_k$ above will give simpler
expressions for functions of interests, such as $f(m)=\log m$ and $f(m)=1/m$. We
shall also optimise this expansion point and henceforth, assume $x_0\neq 0$.

\subsection{Motivation and contributions}\label{subsec:motiv}

We are particularly interested in the functions $f(x) = \log x$ and $f(x) =
1/x$. The former arises in applications such as maximum likelihood estimation
where unbiased estimates of the likelihood are available, for example using
importance sampling. Using the sum estimator framework, these can then be
converted into unbiased estimates of the log-likelihood and its gradient. The
function $f(x)=1/x$ arises in Bayesian and frequentist estimation of models
with un-normalised likelihoods, that is, when the likelihood involves an
intractable normalising constant.

Various unbiased estimators have already been proposed in the literature for
such functions; see Section \ref{sub:related} for a review. However, in our
experience, their performance tends to be very sensitive to `tuning' parameters
such as the expansion point $x_0$ and the distribution of $R$. This is
problematic for users, in particular in situations where one needs to routinely
compute unbiased estimates of $f(m)$ at many different $m$ values. This is the
case, for instance, in stochastic gradient methods, where at each iteration one
needs to estimate the gradient of a different term; see Section
\ref{sec:gradients} for details.

Thus, our motivation is to propose a general `turnkey' approach to unbiased
estimation which is both robust and easy to automate. Specifically, our contributions are as follows. We derive (Section~\ref{sec:general_prop}) some results on the general properties of the randomly truncated Taylor-series based estimator, and use them to inform the choice of the
truncation random variable $R$. We also discuss in detail the properties of the
standard unbiased estimate \eqref{eq:simple} of $(m/x_0 - 1)^k$ and propose an
alternative with potentially much smaller variance. In particular, our new estimate uses the
\emph{whole} collection of samples $X_1,\ldots, X_r$ to construct the unbiased
estimate $\Urk$ of $(m/x_0 - 1)^k$ for $0<k\leq r$. We show the overall variance
of the sum estimator will then decrease at a rate $\OO \left((\log \E[R] )/ \E[R]\right)$. In contrast, the
variance of the sum estimate with the standard choice \eqref{eq:simple} for the
coefficient estimate does not diminish to zero as $\E[R]$ increases. This difference can be significant for some applications.

The second part of our study (Section~\ref{sec:special_funcs}) focuses on the
specific functions $f(x)=\log x$ and $f(x)=1/x$. For these functions, we discuss
how to choose tuning parameters, $x_0$ and the distribution of $R$, to ensure
that the resulting estimator has finite -- and hopefully small -- variance.

Section~\ref{sec:numerical} presents numerical experiments that showcase the
proposed approach in applications such as maximum likelihood estimation in
latent variable models and Bayesian inference in un-normalised models.
Section~\ref{sec:conclusion} concludes with a discussion of the merits and scope
of the proposed approach. All the proofs appear in the Appendix.

\subsection{Related work\label{sub:related}}

Sum estimators and related debiasing techniques have gained traction since their
introduction in \cite{mcleish, glynn2014exact,rhee2015unbiased}; for example, see
\cite{jacob2015nonnegative},
\cite{lyne2015russian},
\cite{Agapiou2018} and
\cite{vihola2018unbiased}.

Some of these papers also consider single term estimators,
which rely on a single randomly selected term to unbiasedly estimate the
infinite series. In our particular context, single term estimators are less
appealing since a single term of a Taylor expansion of $f(m)$ cannot give a
good approximation. Therefore, we do not consider them further. We refer to,
for example, \cite{vihola2018unbiased} for a more thorough discussion on the
connection between single term estimators and sum estimators, and furthermore,
their connection to multi-level Monte Carlo \citep{GilesMLMC} in the context of
stochastic differential equations.

The idea of applying a debiasing technique to a Taylor expansion appeared before
in \cite{blanchet2015unbiasedtaylor}, but without much discussion on how to make
the corresponding estimators applicable in practical scenarios. Our work helps fill this gap.

Debiasing techniques specific to the log function have been developed in
\cite{Luo2020SUMO} and \cite{shi2021multilevel}. However, as is shown in our
numerical experiments, they generate estimates that may either have an infinite
variance, or they have a CPU cost with infinite variance.
In contrast, the CPU cost of our methods have a finite variance, and if properly
tuned, the variance of the estimates are also finite.

The function $f(x)=\exp(x)$, which also has  many important applications, is outside
of the scope of this paper for reasons we explain later. We refer the reader to \cite{papaspiliopoulos2011monte} for a thorough treatment of debiasing techniques for this function.

\subsection{Notations and standing assumptions}

The random variables $X_1,X_2, \ldots$ are independent and identically
distributed (IID) with expected value $m$ and variance $\sigma^2$.

The expansion point $x_0\in \R\setminus\{0\}$ in~\eqref{eq:taylor} is fixed initially, 
and thus so are the $\gamma_k$'s, but we discuss how to choose $x_0$, for a
particular $m$, later on in Section~\ref{sec:special_funcs}.

Our results will be articulated when $R$ has the geometric distribution.
We denote by $\Geom(p)$ the geometric distribution such that $\pr(R=k) =
	(1-p)^k p $, $k=0, 1, \ldots$. Note that $\E[R] = (1-p)/p$ and $\pr(R\geq k) =(1-p)^k$.

We say that an estimator has $\OO(R^k)$ complexity when its computation involves
$\OO(r^k)$ operations, conditional on $R=r$. The estimators considered in
this paper have all complexity at least $\OO(R)$, since they require generating $R$
random variables. 

The notation $\sumest\in\Ltwo$ means $\E[\sumest^2]<\infty$, which implies that $\sumest$ has finite variance and is integrable,
$\E[|\sumest|]<\infty$.

\section{General properties of Taylor-based sum estimators}
\label{sec:general_prop}

\subsection{Variance of the general sum estimator}

In this section, we consider the general estimator~\eqref{eq:def_sumest},
without considering a particular choice for the $U_{r,k}$'s, beyond
Assumption \ref{ass:condition_urk}. Using this assumption,
\begin{equation} \E\left[ \sumest \middle| R=r\right]
	= \sum_{k=0}^\infty \gamma_k\frac{\ind{r \geq k}}{\pr(R\geq k)} \E[\Urk | R=r]
	=\sum_{k=0}^\infty \gamma_k\frac{\ind{r \geq k}}{\pr(R\geq k)} \left(
	\frac{m}{x_0} -1 \right)^k.
\label{eq:exp_fhat_given_r}
\end{equation}
%and therefore $\sumest$ is an unbiased estimate of $f(m) = \sum_k \gamma_k (m/x_0 -1)^k$ 
See Appendix~\ref{app:exp_fhat_given_r} for the verification.

One may decompose its variance, assuming it is finite, through the law of total variance,
\begin{equation}\label{eq:decomp_var}
	\var[\sumest] = \var\left[ \E[\sumest|R ] \right]
	+ \E\left[ \var[\sumest | R ] \right]
\end{equation}
where the first term can be expressed as
\begin{multline}\label{eq:varE}
	\var\left[ \E[\sumest|R ] \right] =
	\sum_{k=0}^\infty \gamma_k^2  \left(  \frac{m}{x_0} - 1 \right)^{2k}
	\left(\frac{1}{\pr(R\geq k)} - 1 \right) \\
	+ 2 \sum_{k=0}^\infty\sum_{l=k+1}^\infty \gamma_k \gamma_l \left(\frac{m}{x_0} - 1\right)^{k+l}
	\left(\frac{1}{\pr(R\geq k)} - 1 \right).
\end{multline}
See Appendix~\ref{app:varE} for the verification.

The term $\var\left[ \E[\sumest|R ] \right]$ has a simple interpretation, namely, it is the variance of $\sumest$ in the ideal case where the $X_i$'s and thus the $U_{r,k}$'s have zero variance, that is,
$U_{r,k}=(m/x_0 - 1)^k$. Hence, this variance component is only due to the random truncation of the Taylor expansion.

To further analyse $\var\left[ \E[\sumest|R ] \right]$, we must make some assumptions on the $\gamma_k$'s.

\begin{assumption} \label{hyp:f_and_x0}
	For the function $f:\mathbb{R}\rightarrow\mathbb{R}$, argument $m$ and expansion point $x_0$, the series expansion in \eqref{eq:taylor} holds. Furthermore,
  (i) for some $c>0$, $|\gamma_k|\leq c$  for all $k=0, 1, \ldots$;
  (ii) $|\gamma_k|\exp(\alpha k)\to \infty$ for all $\alpha>0$;
  (iii) $\beta_0\eqdef|m/x_0 - 1| < 1$.
\end{assumption}
Note that the series~\eqref{eq:taylor} cannot converge if we impose (i) and
(ii) without (iii).

\begin{proposition}
\label{prop:var_exp}
Under Assumptions~\ref{ass:condition_urk} and~\ref{hyp:f_and_x0}, if
$R\sim\Geomp$ with $p\in(0, 1-\beta_0^2)$, then
\begin{equation} \label{eq:hyp_varE}
	\left(f(m)-\gamma_0\right)^2\left ( \frac{p}{1-p}  \right) \leq
	\var\left[ \E[\sumest | R] \right]
	% \var\left[ \tilde{f} \right] \leq
	\leq \frac{c^2\beta_0^2}{(1-\beta_0)^2} \times\frac{p}{1-p-\beta_0^2}.
\end{equation}
\end{proposition}
See Appendix~\ref{app:proof_prop_var_exp} for a proof. Assumption~\ref{ass:condition_urk} is needed to ensure integrability and thus that $\E[\sumest | R]$ is well defined. The condition $\beta_0<1$ of Assumption~\ref{hyp:f_and_x0} and corresponding restriction on $p$ ensures the variance is finite.  

The result above tells us that, under Assumption~\ref{hyp:f_and_x0}, one can take
$R\sim\Geom(p)$ and  $\E[\sumest | R]$  will have finite variance,
provided $p<1-\beta_0^2$.
In particular, the conditional variance behaves like $\OO\left( \{ \E[R]\}^{-1}
\right)$ since $\E[R]=1/p-1\rightarrow \infty$ as $p\rightarrow 0$. We thus
recover the standard Monte Carlo error rate, which is a variance that decreases
as the inverse of the sample size.

\begin{remark}
  We focus on the case $R\sim\Geomp$, since it is an easy-to-sample-from
  distribution; however, the upper bound established above holds for any $R$
  such that $\pr(R\geq k) \geq (1-p)^k$. Conversely, one can see by inspecting the
  proof that a distribution for $R$ with sub-geometric tails should typically
  lead to an infinite variance.
\end{remark}

Of course, what remains to be seen is to what extent we can recover this ideal
behaviour when we use the practical estimates, namely $U_{r,k}$'s, for $(m/x_0 -
1)^k$. This will be the topic of the next two sections.

\begin{remark}
	Both our functions of interest fulfil Assumption~\ref{hyp:f_and_x0}:
	$\gamma_k=(-1)^{k-1}/k$ for $f(x)=\log x$, and $\gamma_k=(-1)^{k-1}$ for $f(x)=1/x$. A function that does \emph{not} fulfil Assumption~\ref{hyp:f_and_x0} is
  $f(x)=\exp(x)$ as $\gamma_k$ decreases at a super-geometric rate. This makes
  it possible to consider for $R$ some distribution with tails that are lighter
  than geometric, for example, the Poisson distribution, which is why we do not
  considered it further, but again see \cite{papaspiliopoulos2011monte} for a
  thorough treatment of this function.
\end{remark}

Before moving on, we state the expression for the second term of
decomposition~\eqref{eq:decomp_var}, which we will study in the
following sections. See Appendix~\ref{app:derivation_variance} for the derivation.
\begin{multline}\label{eq:Evar}
	\E\left[ \var[\sumest|R ] \right] =
	\sum_{k=0}^\infty \frac{\gamma_k^2}{\pr(R\geq k)^2}
	\left\{ \sum_{r=k}^\infty \pr(R=r) \var(\Urk)  \right\}\\
	+ 2 \sum_{k=0}^{\infty}\sum_{l=k+1}^{\infty} \frac{\gamma_k
		\gamma_l}{\pr(R\geq k)\pr(R\geq l)}
	\left\{  \sum_{r=l}^{\infty}\pr(R=r)\cov\left( U_{r,k}, U_{r,l}  \right)\right\}.
\end{multline}
\subsection{Simple estimator}
\label{subsec:simple}

A standard choice \citep{papaspiliopoulos2011monte, blanchet2015unbiased} for
$U_{r,k}$ is $U_{r,k}=U_{r,k}^\mathrm{S}$, where
\begin{equation*}
	U_{r,k}^\mathrm{S}\eqdef \prod_{i=1}^k \left(\frac{X_i}{x_0} - 1\right)
\end{equation*}
which actually does not depend on $r$. We denote by $\sumestS$ the corresponding
sum estimator.

Computing $\sumestS$ requires generating $R$ random variables $X_i$,
but only the first $k$ are used to compute the $k-$th term.
It is easy to see that for any $0 \leq k\leq l$,
\begin{equation*}
	\cov\left(\UrkS, \UrlS\right)
	=\left(\frac{\sigma^2}{x_0^2}+\left(\frac{m}{x_0}-1\right)^2\right)^k
	\left(\frac{m}{x_0} - 1\right)^{l-k}
	-\left(\frac{m}{x_0}-1\right)^{k+l}.
	% \left( \frac{m}{x_0} -1 \right)^{l-k}
\end{equation*}

This leads to the following condition to ensure the finiteness of the variance
of $\sumestS$.

\begin{proposition} \label{prop:exp_var_convergence_simple}
    In addition to Assumption~\ref{hyp:f_and_x0}, 
    suppose $\beta<1$ where
	\begin{equation}\label{eq:hyp_simple}
		\beta^2 \eqdef \frac{\sigma^2}{x_0^2} + {\left( \frac{m}{x_0} -1 \right)^2}= \frac{\sigma^2}{x_0^2} + {\beta_0^2},
	\end{equation}
	and $R\sim\Geomp$ with $p\in(0, 1-\beta^2)$. Then, $\sumestS\in\Ltwo$ and
	\begin{equation}
		\E\left[\var(\sumestS | R)\right] \leq
		\frac{c^2(1+\beta)}{1-\beta}\times\frac{(1-p)}{1-p-\beta^2}. 
		\label{eq:prop_exp_var_convergence_simple}
	\end{equation}
\end{proposition}
See Appendix~\ref{app:proof_prop_exp_var_convergence_simple} for a proof.

Note that, given $m$ and $\sigma^2$, it is always possible to find $x_0$ such
that $\beta < 1$; in fact the minimiser of $\beta$ with respect to $x_0$
verifies this condition.

What is unsatisfactory about this result is that the right hand side does not
vanish as $p\rightarrow 0$. This is not because the bound is not tight, as shown
by the following example whose derivation is given in Appendix~\ref{app:counter}. Our numerical experiments suggest that this is true for other functions of
interest.

\begin{example}
\label{ex:simple_var}
Consider $f(x) =1/x$, and thus $\gamma_k = (-1)^{k-1}$, and assume that we want to estimate $f(m)$ for $m>0$. Then, we have
\begin{align*}
  \lim_{p\to0}\E\left[ \var[\sumestS|R ] \right] &
                                                   = \left(\frac{2}{m}-\frac{1}{x_0}\right)\left\{  \frac{1}{1-\beta^2}
	-\frac{1}{1-\beta_0^2}\right\}.
\end{align*}
%If $\sigma^2=0$, in which case $U_{r,k}=(m/x_0 - 1)^k$, we have $\beta^2 = \beta_0^2$ and $\lim_{p\to0}\E\left[ \var[\sumestS|R ] \right] = 0$. 
If $\sigma^2>0$, then $\beta^2>\beta_0^2$, and since we need to select $x_0>m/2$ to guarantee that the Taylor sum~\eqref{eq:taylor} converges, we have $\lim_{p\to0}\E\left[ \var[\sumestS|R ] \right]>0$, showing that the variance of the simple estimator does not converge to 0 in general.
\end{example}

\subsection{Cycling estimator}
\label{subsec:cycle}

One obvious defect of the simple estimator is that computing $\sumestS$
requires generating $R$ random variables, but $\URkS$ estimates $(m/x_0 - 1)^k$ using only the first
$k$ of these univariates. Utilising the the full sample set $\{X_1,\ldots,X_R\}$ to
estimate $(m/x_0 - 1)^k$ could give more precise estimates.

We propose the following estimate obtained by averaging over circular
shifts of the data, henceforth referred to as the `cycling' estimate,
\begin{equation}
\label{eq:urkc}
	\UrkC \eqdef \frac 1 r \left[ \prod_{i=1}^k \left( \frac{X_i}{x_0} - 1 \right)
		+ \prod_{i=2}^{k+1} \left( \frac{X_i}{x_0} - 1 \right)
		+ \ldots
		+ \left(\frac{X_r}{x_0} - 1\right)\prod_{i=1}^{k-1}\left( \frac{X_i}{x_0} - 1 \right)
		\right].
\end{equation}
Although it averages $r$ non-independent estimates of $(m/x_0 - 1)^k$, it should
have lower variance than the simple estimator $\UrkS$. We denote by $\sumestC$
the corresponding sum estimator.

The following proposition shows that the cycling estimate, indeed, typically has lower
variance than the simple estimate.

\begin{proposition}
	\label{prop:exp_var_convergence_cycle}
	Assume the same conditions as
	Proposition~\ref{prop:exp_var_convergence_simple}, namely,
	Assumption~\ref{hyp:f_and_x0}, condition~\eqref{eq:hyp_simple}, and
  $R\sim\Geomp$ with $p\in(0, 1-\beta^2)$. Then $\sumestC\in\Ltwo$ and
	\begin{align*}
		\Exp\left[\var(\sumestC|R)\right]\leq
		\frac{4c^2 p \log(1/p)}{(1-p-\beta^2)^2}\left[\beta^2+2\frac{\beta_0(1-p)}{(1-\beta_0)^2}\right].
	\end{align*}
\end{proposition}
See Appendix~\ref{app:proof_prop_exp_var_convergence_cycle} for a proof.

In other words, by taking $\E[R]\rightarrow +\infty$, the cycling estimate recovers, up to a log
factor, the standard Monte Carlo rate, i.e. the variance is $\OO\left((\log\{\E[R]\}) / \E[R]\right)$.
Thus, both terms of the variance in \eqref{eq:decomp_var} now decrease as the mean of $R$ increases.

A direct corollary, through Chebyshev's inequality, of the above result is that
the cycling estimator converges in probability.

\begin{corollary}
\label{cor:cycle}
	Under the same assumptions as Proposition~\ref{prop:exp_var_convergence_cycle},
	$\sumestC \cvprob f(m) $ as $p\rightarrow 0$.
\end{corollary}

\subsection{Minimum variance unbiased estimator}

Instead of the cycling estimator, one could consider the following U-statistic:
\begin{equation}
  \label{eq:mvue}
	\Urk^{\mathrm{MVUE}} \eqdef
	\frac{1}{\binom{r}{k}}
	\sum_{1\leq i_1<\dots<i_k\leq r} \left(\frac{X_{i_1}}{x_0}-1\right)\cdot\dots\cdot \left(\frac{X_{i_k}}{x_0}-1\right).
\end{equation}
which has minimimum variance among unbiased estimators of its expectation, see
\cite{Hoeffding1948a}. In spite of the combinatorial number of terms, it is
still possible to compute the estimator in $\OO(r^2)$ operations, using a
recursion, as pointed out by Sam Power in a personal communication, see
Appendix~\ref{app:mvue}. However, repeating our variance analysis for the MVUE
estimator appears challenging, indeed due to the combinatorial number of terms
in \eqref{eq:mvue}, and is left for further study. For the moment, we present a
numerical comparison in Appendix~\ref{app:mvue}, which happens to show cycling
and MVUE performing similarly.

\subsection{Relative merits of the simple and cycling estimators}
\label{sec:merits}
The cycling estimator makes it possible to estimate $f(m)$ as accurately as
required, and enjoys, up to a log factor, the same error rate as Monte Carlo
estimates. On the other hand, it is more expensive to compute than the simple
estimator. Although both estimators require generating $R$ random variables, computing the $U_{R,k}$'s ($0\leq k \leq R$) for the cycling estimator requires $\OO(R^2)$ operations.

That said, in many practical scenarios, the $X_i$'s can be expensive to generate. For instance, in our third numerical experiment we run a sequential Monte Carlo (SMC) sampler to get a
single $X_i$, which takes about 15 seconds. The CPU cost of the $\OO(R^2)$ operations for the cycling estimate is negligible compared to the CPU cost of generating the $X_i$'s,  unless $R$ is
extremely large. In this scenario, the cycling estimate is clearly preferable owing to the potentially large variance reductions achieved.

Another way to see the relative merits of both estimators is to look at their
robustness relative to tuning parameters. This point will be discussed at
length in Section \ref{sec:special_funcs}. However, if our figure of merit
is `variance times CPU cost', which is a common figure of merit in
Monte Carlo, then the cycling estimator $\sumestC$ is clearly preferable since this
figure of merit diverges at a log rate as $p\rightarrow 0$, whereas it diverges at a
linear rate for the simple estimator.

\subsection{Unbiased estimators of gradients}\label{sec:gradients}

In this section, we further develop the simple and cycling estimators to
unbiasedly estimate $\nabla_\theta f\left(m(\theta)\right)$, the gradient of
$f\left(m(\theta)\right)$ when $m=m(\theta)$ depends on a parameter $\theta$. A
typical application that utilises unbiased gradient estimates is maximum
likelihood estimation; see Section~\ref{sec:numerical} for an example.

Assuming $m(\theta)$ is differentiable, using~\eqref{eq:taylor},
\begin{equation}\label{eq:taylor_gradient_theta}
	\nabla_\theta f(m(\theta))
  = \sum_{k=1}^{\infty} k\gamma_k \left( \frac{m(\theta)}{x_0} - 1 \right)^{k-1}\frac{\nabla_\theta m(\theta)}{x_0},
\end{equation}
and the corresponding `sum' estimator, that is similar to \eqref{eq:def_sumest} in construction, is 
\begin{equation} \label{eq:def_sumest_gradient}
	\nablasumest \eqdef
	%\sum_{k=1}^R \frac{k \gamma_k }{x_0}\frac{W_{R,k}}{\pr(R\geq k)} =
	\sum_{k=1}^\infty \frac{\ind{R\geq k}}{\pr(R\geq k)} \WRk\frac{k\gamma_k}{x_0},
\end{equation} where $R$ is a random variable taking values in $\{1, \ldots\}$; $\WRk=\sum_{r=k}^\infty \Wrk \ind{R=r}$; $\Wrk$ is an unbiased estimate of $(m(\theta)/x_0 - 1)^{k-1}\nabla_\theta m(\theta)$ for $r \geq k$; and each and every $\Wrk$ is independent of $R$.

As was stated for $\sumest$ in Lemma \ref{lem:integrablefhat}, the gradient
estimate can be shown to be integrable and unbiased when the estimates $\Wrk$
satisfy similar conditions of Assumption~\ref{ass:condition_urk}, that is, with
univariates $\Urk$ therein replaced with $\Wrk$, and $\gamma_k$ replaced with
$k\gamma_k$.  See Appendix~\ref{app:integrablefhat} for the proof of
unbiasedness of $\sumest$ which can be easily adapted to the present gradient
case.

We can define a simple estimator for~\eqref{eq:taylor_gradient_theta},
denoted by $\nablasumestS$, by replacing $\Wrk$
in~\eqref{eq:def_sumest_gradient} with
\begin{equation}
\label{eq:W_simple}
	\WrkS\eqdef G_k \prod_{i=1}^{k-1} \left(\frac{X_i}{x_0} - 1\right)
\end{equation}
where $(X_1, G_1), (X_2, G_2), \dots,(X_r, G_r)$ is an IID collection of bivariates such that
$\E[X_i]=m(\theta)$ and $\E[G_i] = \nabla_\theta m(\theta)$.

Instead of $\WrkS$, the cycling estimator, denoted by $\nablasumestC$, uses $\WrkC$ which is an average over `circular shifts' of $\WrkS$,
\begin{equation*}
	\WrkC \eqdef \frac 1 r \Bigg[ G_k\prod_{i=1}^{k-1} \left( \frac{X_i}{x_0} - 1 \right)
		+ G_{k+1}\prod_{i=2}^{k} \left( \frac{X_i}{x_0} - 1 \right) 
		+ \ldots
		+ G_{k-1}\left(\frac{X_r}{x_0} - 1\right)\prod_{i=1}^{k-2}\left( \frac{X_i}{x_0} - 1 \right)
		\Bigg].
\end{equation*}

As was done in the previous section, we can analyse the variance of both
estimators using the law of total variance~\eqref{eq:decomp_var}. To ensure that
the second term of the variance decomposition, namely $\E\left[\var(\nablasumest |
  R)\right]$, is finite we further need to assume that the second moment of
$(X_1, G_1)$ is finite.

\begin{proposition}
  \label{prop:var_gradients}
At the point $\theta$, assume $(X_1, G_1)$ has finite second moments and let
$\sigma^2(\theta) = \var(X_1) $. Furthermore, let
Assumption~\ref{hyp:f_and_x0} hold at $m=m(\theta)$; assume
	\begin{equation*}
		\beta^2 \eqdef \frac{\sigma^2(\theta)}{x_0^2} + \left( \frac{m(\theta)}{x_0} -1 \right)^2 < 1;
	\end{equation*}
	and let $R\sim\Geomp$ with $p\in(0, 1-\beta^2)$. Then both $\nablasumestS$ and $\nablasumestC$ are in $\Ltwo$. In addition,
\begin{enumerate}
\item $\var\left[ \E[\nablasumest|R ]\right] =\OO(p)$ as $p\to 0$, for $\nablasumest \in \{\nablasumestS,\nablasumestC\}$.
\item As $p\to 0$, $\E\left[\var(\nablasumestS | R)\right]$ converges to a finite, possibly non-zero, value.
\item $\E\left[\var(\nablasumestC | R)\right] =\OO\left(p \log(1/p) \right)$ as $p\to 0$.
\end{enumerate}
\end{proposition}
The proof is in Appendix~\ref{app:var_gradients}.

The first statement implies the variance of the conditional expectation decreases as $\OO(1/\E[R])$ as $\E[R]\rightarrow\infty$. The second statement confirms the variance of the simple estimate does not diminish to zero as $\E[R]$ increases; this is verified for $f(x)=\log (x).$ In contrast, the final statement confirms the cycle method's variance diminishes at rate $\log(\E[R])/\E[R]$.

\section{Tuning parameters}\label{sec:special_funcs}

\subsection{Tuning $x_0$\label{subsec:x0}}
We now discuss how to set the tuning parameters $x_0$ and $p$ in practice. For
the sake of simplicity, we focus on the cycling estimator, $\sumestC$, and we
assume $m>0$, but adapting our discussion to either $\sumestS$ or $m<0$ should
be straightforward.

Propositions~\ref{prop:exp_var_convergence_simple} and
\ref{prop:exp_var_convergence_cycle} of Section~\ref{sec:general_prop} reveal
the key considerations for selecting $x_0$ and parameter $p$ of $R\sim\Geomp$,
appropriately. Firstly, to make sure that the variance of $\sumestC$ is finite,
$x_0$ must be chosen to ensure that $\beta^2<1$; see~\eqref{eq:hyp_simple}.
Secondly, once $x_0$ is chosen, then $p$ has to be found from the interval $p\in
(0,1-\beta^2)$. This order cannot be inverted in general because, if $p$ is too
close to one, there may not be an $x_0$ that guarantees $\beta^2<1-p$ as
$\beta^2$ has a unique minimiser with respect to $x_0$.

We now go through how to make these critical choices. From our experience, this discussion is largely neglected in the literature. We will further illustrate in Section~\ref{sec:numerical}, with the aid of numerical examples, the importance of these tuning considerations, and the robustness of our tuning procedure below.

Recall, we focus the discussion for  $m>0$. We see that $\beta^2<1 \iff x_0 > x_0^{\min} \eqdef m /2+\sigma^2/(2m)$. Rather than targeting $x_0^{\min}$, we target the $x_0$ that minimises $\beta^2$, namely
 \begin{align}
\label{eq:optimal_x0}
 x_0^\star \eqdef \arg\min_{x_0} \left\{ \frac{\sigma^2}{x_0^2} + \left(
	\frac{m}{x_0} -1 \right)^2 \right\} =
	\frac{m^2 + \sigma^2}{m}.
\end{align}
Being the minimiser, $x_0^\star$ automatically fulfils the condition for $\beta^2<1$. Incidentally, $x_0^{\star}=2 x_0^{\min}$. But there are also other reasons why $x_0^{\star}$ might be a good choice. The bound in Proposition~\ref{prop:exp_var_convergence_cycle} grows
with $\beta^2$ and even diverges as $\beta^2\rightarrow 1 - p$. So, it makes sense to
minimise this quantity with respect to $x_0$. Also, minimising $\beta^2$ permits a larger $p$, and this could save on the computational cost to generate the estimate of $f(m)$.

The catch is that we do not know the values of $m$ and $\sigma^2$. The only general solution to this issue is to use a pilot run. That is, use an initial batch of $n_0$ random variables $X_i$ to obtain estimates
$\widehat{m}$ and
$\widehat{\sigma^2}$, plug these estimates into \eqref{eq:optimal_x0} to obtain
$\widehat{x_0^\star}$ and use this estimate as the value for $x_0$.

Note however the following fundamental drawback of this approach. No matter how
large the pilot run is, the probability that $\widehat{x_0^\star} < x_0^{\min}$,
and equally $\beta^2 > 1$,
will be non-zero, and the variance of $\sumestC$ cannot be guaranteed to be finite.

To mitigate this issue, we need to find a way to make the probability $x_0 <
x_0^{\min}$ as small as possible for the chosen $x_0$. We propose the following
{\it bootstrap} \citep{EfronTibshirani_bootstrap} approach. Generate a pilot
sample of size $n_0$ and use it as follows: 
\begin{itemize} 
    \item[i.] Estimate $x_0^\star = (m^2 + \sigma^2)/m$ as 
        $\widehat{x_0^\star}=(\widehat{m}^2+\widehat{\sigma^2})/\widehat{m}$ where
  $\widehat{m}$ and $\widehat{\sigma^2}$ are the empirical mean and variance of
  the pilot sample. 

    \item[ii.] Find the upper limit of a lower one-sided bootstrap
        $(1-\alpha)$-confidence interval for $x_0^{\min}$, with $\alpha$ very
        small; note that $x_0^{\min}=x_0^\star/2$. 

\item[iii.] Set $x_0$ to the larger of these two values.

\end{itemize} 

  In this way, not only will the
$\pr(x_0>x_0^{\min})$ be close to one for the chosen $x_0$, the latter will
potentially also be close to $x_0^{\star}$ when the sample mean and variance
are accurate.

In practice, we recommend using the percentile method to construct the bootstrap
confidence interval, that is, the upper bound is the $(1-\alpha)$-quantile of
the bootstrap estimates of $x_0^{\min}$, as this method may work well even when
the distribution of the $X_i$'s is not symmetric; for example, see
\cite{EfronTibshirani_bootstrap} for more background on bootstrap confidence
intervals. In our experiments we set $\alpha = 0.01$. A second obvious
recommendation is to favour a large $n_0$, as the computational budget permits.

\subsection{Setting $p$ and the CPU budget}\label{subsec:budget}

Steps (i)-(iii) of Section \ref{subsec:x0} return an $x_{0}$ value such that $\pr\left(\widehat{x_{0}^{\min}}<x_{0}\right)\geq1-\alpha$,
where $\widehat{x_{0}^{\min}}$ is the plug-in estimate and $\pr(\cdot)$
is computed with the empirical cumulative distribution function given by the pilot samples $\left\{ X_{i}\right\} _{i=1}^{n_{0}}$.
Let $\widehat{\beta^{2}}$ be the corresponding plug-in estimate of
$\beta^{2}$: in \eqref{eq:hyp_simple}, use the found $x_{0}$ from Step (iii) and $(\widehat{m}, \widehat{\sigma^2})$ from Step (i). Then -- at this returned value for $x_{0}$ -- evidently,
we also have $\pr\left(\widehat{\beta^{2}}<1\right)\geq1-\alpha$.
This motivates choosing $p$ as follows:
\begin{itemize}
\item[iv.] Set $p=\min\left(1-\widehat{\beta^{2}},\frac{1}{n_{0}+1}\right)$.
\end{itemize}
The reason why we `align' the choice of $p$ with the size of the
pilot run $n_{0}$ is discussed next. Note that $\Geom(1/(n_0+1))$ has mean $n_0$. The resulting procedure is summarised in Algorithm~\ref{alg:cycling} for the cycling estimator. The simple estimator is obtained by replacing $\UrkC$ with $\UrkS$.

\begin{algorithm}
\begin{algorithmic}[1]
\STATE{\textit{Inputs}: number of pilot runs $n_0$, a mechanism to sample $X_i$ with mean $m$, confidence level $\alpha$, $(\gamma_k)_{k\geq 0}$. }
\STATE{\textit{Tuning}:}
\STATE{Sample $X_1, \dots X_{n_0}$ and compute empirical mean $\widehat{m}$ and variance $\widehat{\sigma^2}$.}
\STATE{Estimate $\widehat{x_0^\star}=(\widehat{m}^2+\widehat{\sigma^2})/\widehat{m}$.}
\STATE{Build a one-sided bootstrap $(1-\alpha)$-confidence interval $(-\infty, u]$, for $\sigma^2/(2m)+m/2$.}
\STATE{Set $x_0 = \max\{u, \widehat{x_0^\star}\}$ and $\widehat{\beta^2} = \widehat{\sigma}^2/x_0^2 + (\widehat{m}/x_0-1)^2$.}
\STATE{Set $p=\min\left(1-\widehat{\beta^{2}},\frac{1}{n_{0}+1}\right)$.}
\STATE{\textit{Estimator}:}
\STATE{Sample $R\sim \Geom(p)$ and, given $R= r$, sample $X_1, \dots, X_r$.}
\STATE{Compute $\UrkC$ as in~\eqref{eq:urkc}.}
\STATE{\textit{Output:} $\sumestC = \sum_{k=0}^r \frac{\gamma_k \UrkC}{(1-p)^k}$.}
 \end{algorithmic}
 \caption{Cycling Taylor-based sum estimator} \label{alg:cycling}
\end{algorithm}

Choosing $p$ amounts to choosing the number of variables $X_i$ that must be
generated, and ultimately the CPU budget of $\sumestC$ since, as we already
mentioned, the generation of the $X_i$'s often dominates the CPU cost of
$\sumestC$. 
For both the cycling and simple estimators, we have seen in the previous section that a large pilot run of size $n_0$ is
often required to choose $x_0$ well. It is thus illusory to try to use $\sumestC$
in a {\it low budget} regime.

Put it in another way, the figure of merit we should
now consider is the \emph{work-normalised variance} $\wnv:= (n_0 + \E[R]) \times \var[\sumest]$ \citep{glynn1992asymptotic}. We do not consider the cost of building the estimators since, as discussed in
Section~\ref{sec:merits}, we are mostly interested in scenarios in which the
$X_i$'s are expensive to generate, see Section~\ref{sec:erg} for an example.

Figure~\ref{fig:wnv} plots an upper bound on the work-normalised variance, obtained by replacing $\var[\sumest]$ with the bounds given
in Proposition~\ref{prop:var_exp} and~\ref{prop:exp_var_convergence_cycle} for the cycling estimator (Proposition~\ref{prop:exp_var_convergence_simple} for the simple estimator) in the case $x_0=x_0^\star$, as a function of $\E[R]$. We consider $n_0=10$ since this is the value used in the experiments in Section~\ref{sec:numerical}.
When $\E[R]$ gets
too small, both bounds diverge at rate $\OO\left( 1/(\E[R] - \beta^2) \right)$; when $\E[R]$ tends to infinity, the bound for the cycling estimator grows slowly, at rate $\OO(\log \E[R])$, while that for the simple estimator diverges at rate $\OO(\E[R])$.

In practice, this means that the cycling estimator is more robust with respect to the choice of $p$: for larger values of $\E[R]$ the cost increases, but while the cycling estimator provides results with diminishing variance, the simple one fails to do so and the corresponding $\wnv$ explodes at a faster rate. This makes cycling preferable to the simple estimator. 

% The cycling estimator with  $\E[R]$ chosen to be the same order as this pilot run size delivers a lower variance estimate with the same order of cost. 

\begin{figure}
\centering
\begin{tikzpicture}[every node/.append style={font=\normalsize}]
\node (img1) {\includegraphics[width = 0.3\textwidth]{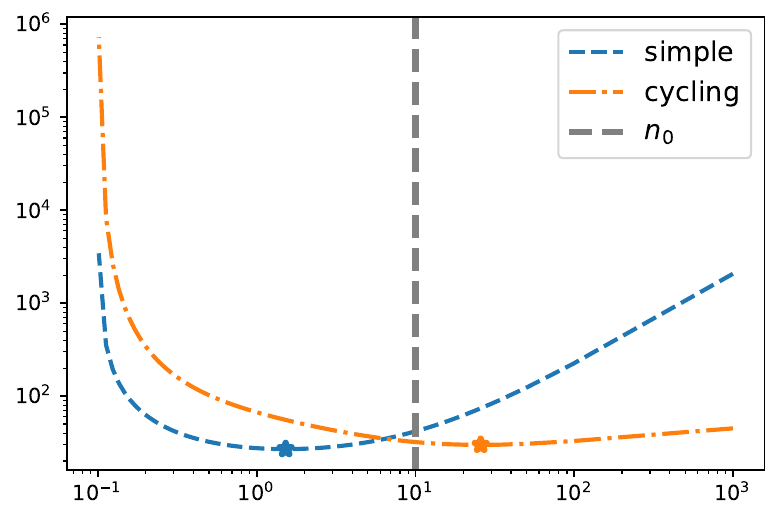}};
\node[above=of img1, node distance = 0, yshift = -1.2cm] {low};
\node[right=of img1, node distance = 0, xshift = -0.9cm] (img2) {\includegraphics[width = 0.3\textwidth]{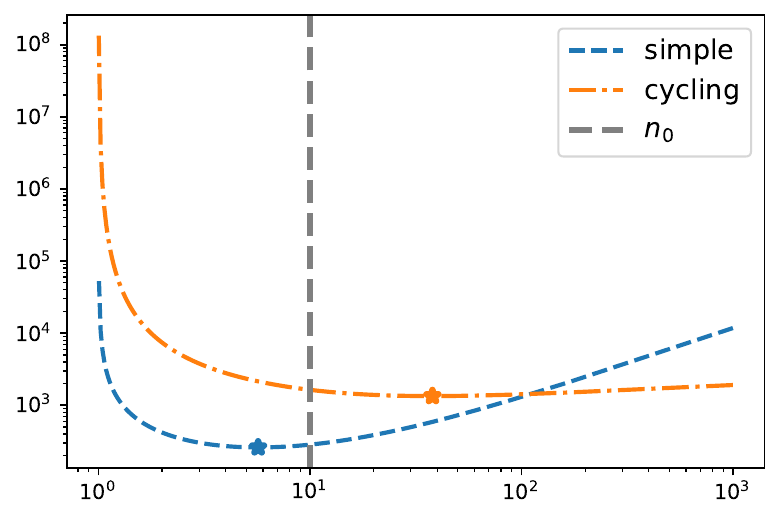}};
\node[right=of img2, node distance = 0, xshift = -0.9cm] (img3) {\includegraphics[width = 0.3\textwidth]{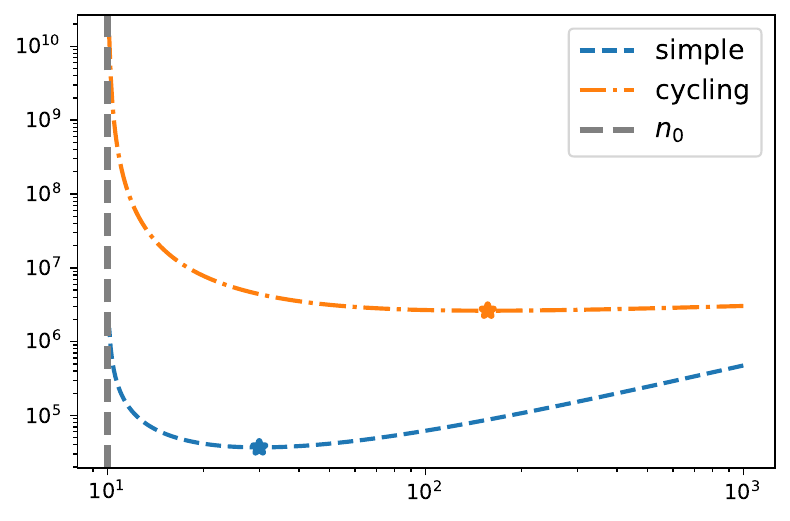}};
\node[above=of img2, node distance = 0, yshift = -1.2cm] {moderate};
\node[above=of img3, node distance = 0, yshift = -1.2cm] {high};
\node[left=of img1, node distance = 0, rotate = 90, anchor = center, yshift = -0.8cm] {$\wnv$};
\node[below=of img1, node distance = 0, yshift = 1.2cm] {$\E[R]$};
\node[below=of img2, node distance = 0, yshift = 1.2cm] {$\E[R]$};
\node[below=of img3, node distance = 0, yshift = 1.2cm] {$\E[R]$};
\end{tikzpicture}
\caption{Comparison of the upper-bound of the work-normalised variance in a low, moderate and high variance scenario for $n_0 = 10$ as in the experiments in Section~\ref{sec:numerical}. The stars denote the value of $\E[R]$ minimising the upper-bound of the WNV and the vertical line $\E[R]=n_0 = 10$.}
\label{fig:wnv}\end{figure}
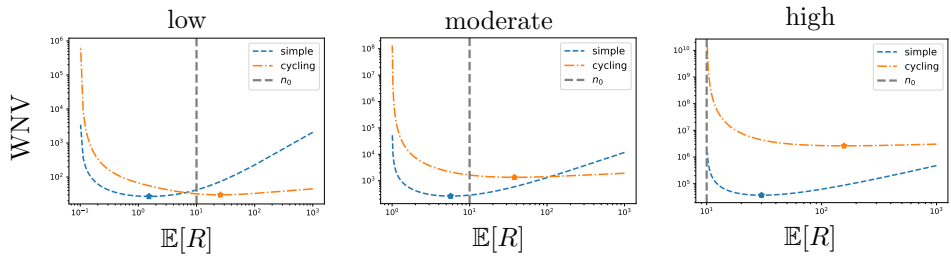

We tried various ways to choose $p$ in an adaptive manner, by trying to minimise the upper bound on the work-normalised variance; see Appendix~\ref{app:wnv} for more details.
However, we found in practice the following very basic strategy to work well. Set $\E[R]=n_0$, the size of the pilot run; see Algorithm \ref{alg:cycling}.  In this way, we fix the CPU budget beforehand to twice $n_0$ and given that $n_0 \gg1$, this is often enough to ensure that we are in the right part of the curves in Figure~\ref{fig:wnv}.
In the low and moderate variance cases, we also see that $\E[R]=n_0$ is not far from the value minimising the upper bound on $\wnv$ for the cycling and simple estimator.

\section{Numerical experiments}
\label{sec:numerical}

\subsection{Toy latent variable model}\label{subsec:toy_lvm}

We consider the toy latent variable model (LVM) of \cite{rainforth2018tighter}, also used in \cite{shi2021multilevel}. Conditioned on $\theta\in \real$, the observations $y_1, \dots, y_n$ are sampled independently as follows,
\begin{align*}
	z_i|\theta & \sim \mathcal{N}(\cdot;\theta\textsf{1}_d, \textsf{Id}_{d}), \qquad\qquad y_i|z_i \sim \mathcal{N}(\cdot;z_i, \textsf{Id}_{d}),
\end{align*}
where $z_i$ and $y_i\in \real^d$. 

We set $\theta = 1$ and sample a data set of size $n=10^3$.
The goal is to  unbiasedly estimate the log-likelihood $\sum_{i=1}^n \log p(y_i|\theta)$ for
$\theta=1$, a quantity which we are able to compute exactly for this model.

To do so, we set $f(m)=\log m$ and consider each of the $n$ data points separately. For
a given $y\in\{y_1,\ldots,y_n\}$, we apply our approach as follows:
conditional on $R=r$, we sample $Z_1, \dots, Z_r\sim p(z|\theta)$ and set $X_j =
	p(y|Z_j, \theta)$ for $j=1,\dots, r$, which is evidently an unbiased estimate of
$p(y |\theta)=\int p(y|z,\theta) p(z|\theta)dz$. In this particular case, we can exactly compute  the 
moments $m=\E[X_j]=p(y|\theta)$ and $\sigma^2=\var[X_j]$; the optimal expansion point $x_0^\star = m + \sigma^2/m$; and the
corresponding value for $\beta^2$. See Appendix~\ref{app:toy_lvm} for their
expressions.

We use this toy model to compare the performance of the simple and cycling
estimators of $\log p(y|\theta)$ with that of the multilevel Monte Carlo (MLMC)
estimator of \cite{shi2021multilevel}. They use MLMC to modify the SUMO
approach of \cite{Luo2020SUMO} to return an unbiased estimate of $\log
	p(y|\theta)$ {\it with} finite variance. Thus, we do not consider the SUMO
estimator here since its variance is infinite. Furthermore, \citet[Section
	7.1]{shi2021multilevel} show, empirically, that MLMC outperforms SUMO.

The estimator proposed by \cite{shi2021multilevel} employs 
MLMC \citep{blanchet2019unbiased} to debias the estimator of $\log
	p(y|\theta)$ given by Importance Weighted Autoencoders (IWAE;
\cite{burda2015importance}). The MLMC estimator is
\begin{align}
	\label{eq:mlmc_mll}
	\widehat{\log p(y|\theta)}^{\textrm{MLMC}}= I_j + \sum_{k=0}^\infty \frac{\Delta_k\ind{\widetilde{R}\geq k}}{\pr(\widetilde{R}\geq k)},
\end{align}
where $\widetilde{R}\sim\Geom(0.6)$,
$I_j$, $j\geq 0$, is the IWAE estimate with proposal $q(\cdot;y)$ and $2^j$ samples:
\begin{align*}
	I_j \coloneq 
     \log\frac{1}{2^j}\sum_{h=1}^{2^j} w(Z_h;y, \theta), 
     \quad \text{with } w(Z_h;y, \theta) \coloneq \frac{p(Z_h, y| \theta)}{q(Z_h;y)},
     \quad Z_1, \dots, Z_{2^j} \sim q(\cdot;y).
\end{align*}
and $\Delta_k:= I_{j+k+1}-(I_{j+k}^E+I_{j+k}^O)/2$; with $I_{j+k}^E$ (resp.
$I_{j+k}^O$) being the IWAE estimate that only uses the even-indexed (resp.
odd-indexed) samples. We consider the same proposal as in
\cite{shi2021multilevel}, i.e.
$\N(\cdot;(y+\theta^\star\textsf{1}_d)/2, \frac{2}{3}\textsf{Id}_d)$, where $\theta^\star$ denotes the maximum likelihood estimator for $\theta$, which is
shown to be optimal within a certain family in \cite{rainforth2018tighter}.

The expected CPU cost, i.e. the number of generated $X_j$, is $\E[R]=1/p-1$ for
our Taylor-based estimators, and $\E[2^{j+1+\widetilde{R}}]$ for the MLMC
approach of \cite{shi2021multilevel}. We adjust $p$ and $j$ so that both
expectations are equal to a certain pre-specified budget $C$.
We consider two values for $C$, namely $C=6$ and $C=96$.

Figure~\ref{fig:iwae_comparison_cost} shows 100 estimates of each type (simple,
cycling, MLMC) for $d=2$ (left column) and $d=5$ (right column) against their
computational cost, for an expected cost $C$ of 6 (top line) and 96 (bottom
line) samples per data point. The actual computational cost (reported on the
$y$-axis) is the total number of variables $Z_j$ generated in order to estimate
all the terms $\log p(y_i|\theta)$ for $i=1,\dots, n$.

This experiment shows that the cost of MLMC has large empirical variance. In
fact, the true variance of the cost is infinite, see Appendix~\ref{app:toy_lvm}
for a proof of this statement. In contrast, the cost of the simple and cycling
estimators remains close to its expectation.

% ; however, for large $d$ this cost might not be
% enough to guarantee that the upper bound on the variance obtaining by combining
% Proposition~\ref{prop:var_exp} and~\ref{prop:exp_var_convergence_cycle} is
% finite.

The simple and cycling estimators outperform MLMC in terms of
accuracy, with the simple estimator having larger variance than the cycling
estimator for higher computational cost. This is particularly evident in the first plot of the second row, in which we compare the estimates for $\E[R]=96$, in this case the cloud of estimates obtained with cycling is considerably more concentrated around the true value (the horizontal line). 
This is consistent with the
theoretical results established in Section~\ref{sec:general_prop} which show
that the variance of the simple estimator does not decay to 0 when the
computational cost increases.

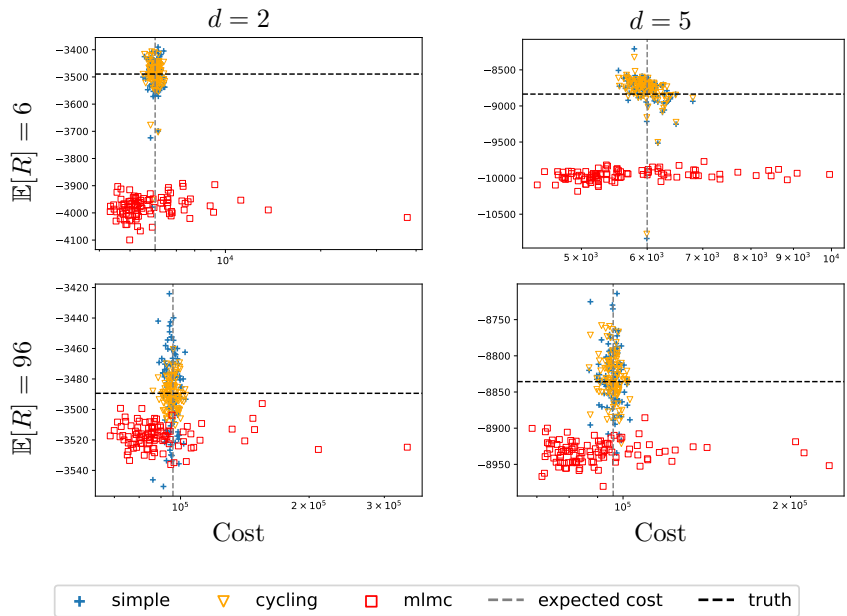
\begin{figure}
	\centering
	\begin{tikzpicture}[every node/.append style={font=\normalsize}]
		\node (img1) {\includegraphics[width = 0.4\textwidth]{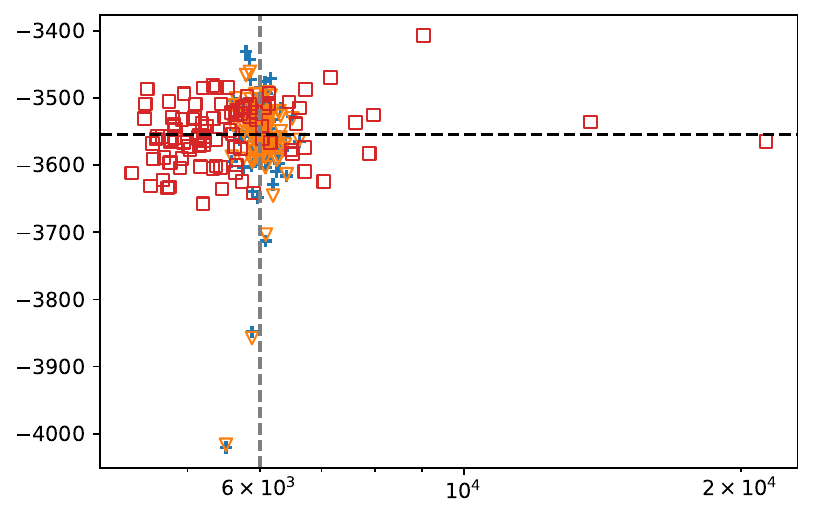}};
		\node[above=of img1, node distance = 0, yshift = -1.2cm] {$d=2$};
		\node[right=of img1, node distance = 0, xshift = -0.7cm] (img2) {\includegraphics[width = 0.4\textwidth]{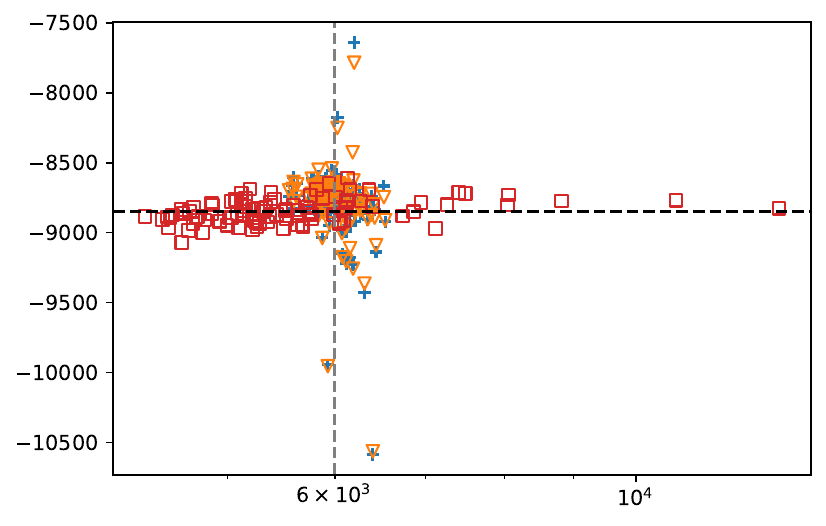}};
		\node[left=of img1, node distance = 0, rotate = 90, anchor = center, yshift = -0.8cm] {$\Exp[R] = 6$};
		\node[above=of img2, node distance = 0, yshift = -1.2cm] { $d=5$};
		\node[below=of img1, node distance = 0, yshift = 1.2cm] (img3) {\includegraphics[width = 0.4\textwidth]{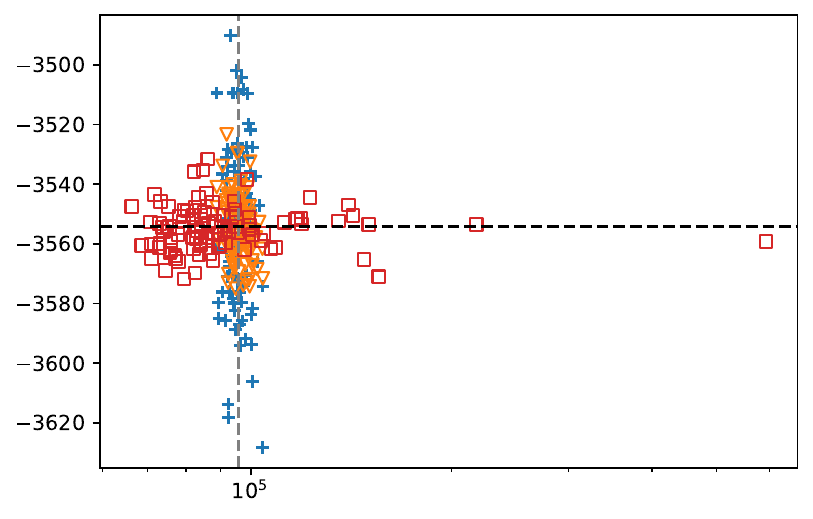}};
		\node[right=of img3, node distance = 0, xshift = -0.7cm] (img4) {\includegraphics[width = 0.4\textwidth]{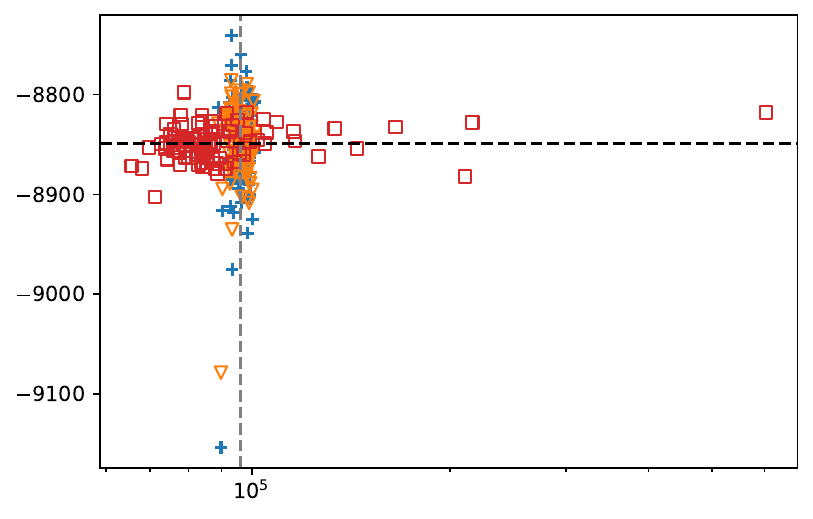}};
		\node[left=of img3, node distance = 0, rotate = 90, anchor = center, yshift = -0.8cm] {$\Exp[R] = 96$};
		\node[below=of img3, node distance = 0, yshift = 0.5cm, xshift = 2.5cm] (img9) {\includegraphics[width = 0.8\textwidth]{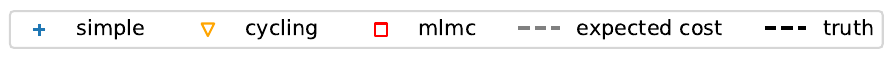}};
		\node[below=of img3, node distance = 0, yshift = 1.2cm] {Cost};
		\node[below=of img4, node distance = 0, yshift = 1.2cm] {Cost};
	\end{tikzpicture}
	\caption{Estimates of the marginal log-likelihood for a full data set of size
		$n=10^3$ plotted against computational cost (measured in number of samples
		$Z_j$ drawn) for an expected cost $C$ of 6 and 96 samples per data point.
		The first column shows the results
		for $d=2$ while the second one for $d=5$. The vertical dashed line denotes the
		expected cost while the horizontal one denotes the true value of
		$\sum_{i=1}^n\log  p(y_i|\theta)$.}
	\label{fig:iwae_comparison_cost}
\end{figure}

As a note of caution, our Taylor-based estimators may give very poor results
for a larger dimension $d$; see Appendix~\ref{app:toy_lvm} for extra results
when $d=20$. Simple calculations show that, in that case, one would need to
take $\E[R]\gg 10^8$ to ensure that $p<1-\beta^2$ and that the estimators have
finite variance. In this case, the inputs have very large variance, and
presumably the only way to fix this issue is to come up with a different
strategy to unbiasedly estimate $p(y|\theta)$, so that the corresponding
estimates have much lower variance.

\subsection{Independent component analysis}
\label{sec:ica}
We consider a LVM associated with probabilistic independent component analysis (ICA). The data $y_1, \dots, y_n$ are $d_y$-dimensional vectors obtained from 
\begin{equation}
y = \sum_{j=1}^{d_z} z_j a_j +\sigma \varepsilon,
\label{eq:ica}
\end{equation}
where $z_j$ are the latent random variables; $a_j$'s are parameter vectors of dimension $d_y$; $\varepsilon$ is a $d_y$-dimensional standard Gaussian vector and $\sigma>0$. Furthermore,  $z_j = b_j\zeta_j$ where $b_j\sim \textrm{Bernoulli}(\alpha)$ and $\zeta_j \sim \textrm{Logistic}(1/2)$ for some fixed $\alpha$. 
The aim is to estimate the parameter vector $\theta=(A, \sigma)$, with $A:=(a_1, \dots, a_{d_z} )$, via maximum marginal-likelihood estimation, that is, we want to find the $\theta$ that maximises $\sum_{i=1}^n \log p(y_i|\theta)$.

We repeat the first experiment in \cite{allassonniere2012stochastic}. We generate $n=100$ data points from \eqref{eq:ica} with the following model choices: $d_z = 2$, $\sigma = 0.5$, $\alpha=0.8$. Vectors $a_1$ and $a_2$ are the $16\times 16$ images in the first panel of Figure~\ref{fig:ica_true_theta}. Note that $d_y = 256$ and the parameter space has dimension $d_\theta = 2\times256+1$.

To estimate $\theta$, \cite{allassonniere2012stochastic} use a stochastic approximation expectation maximisation (SAEM) algorithm. We will instead use stochastic gradient descent \citep{robbins1951stochastic},  implemented with our unbiased estimators of $\nabla_\theta f(m(\theta)) = \nabla_\theta\log p(\theta|y)$. We call our two implementations  simple-SGD and cycling-SGD. 
We also implement SGD with the gradient unbiasedly estimated using the MLMC strategy of \cite{shi2021multilevel} (MLMC-SGD).
The bivariates $(X_i, G_i)$ (for $i=1,2,\ldots$) of Section~\ref{sec:gradients} are generated using importance sampling; and the same importance sampling proposal is used for MLMC. For MLMC, we adjust its parameters $p$ and $j$, defined in \eqref{eq:mlmc_mll} with $\tilde{R}\sim\Geom(p)$, so that $\Exp\left[2^{j+1+\widetilde{R}}\right]= n_0$. Similarly, $p$ for the Geometric random truncation of the simple and cycling gradient estimates is found by setting  $\Exp\left[R\right]= n_0$. See Appendix~\ref{app:ica} for all the  details.

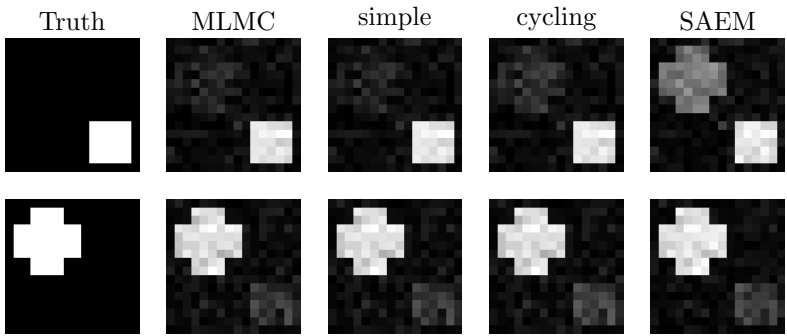
\begin{figure}
\centering
\begin{tikzpicture}[every node/.append style={font=\normalsize}]
\node (img1) {\includegraphics[width = 0.15\textwidth]{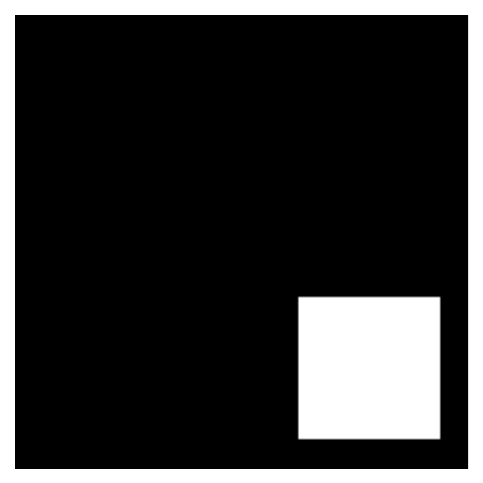}};
\node[below=of img1, node distance = 0, yshift = 1cm] (img2) {\includegraphics[width = 0.15\textwidth]{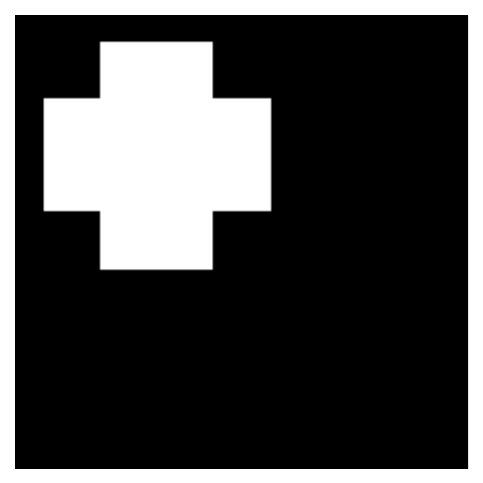}};
\node[right=of img1, node distance = 0, xshift = -1cm] (img3) {\includegraphics[width = 0.15\textwidth]{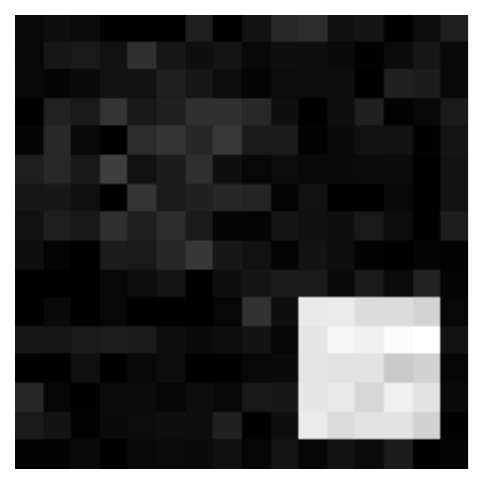}};
\node[below=of img3, node distance = 0, yshift = 1cm] (img4) {\includegraphics[width = 0.15\textwidth]{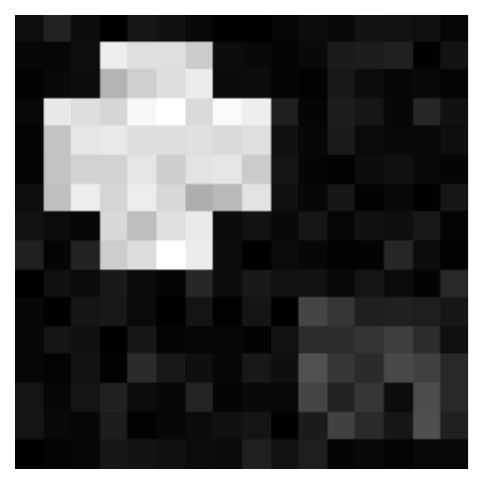}};
\node[right=of img3, node distance = 0, xshift = -1cm] (img5) {\includegraphics[width = 0.15\textwidth]{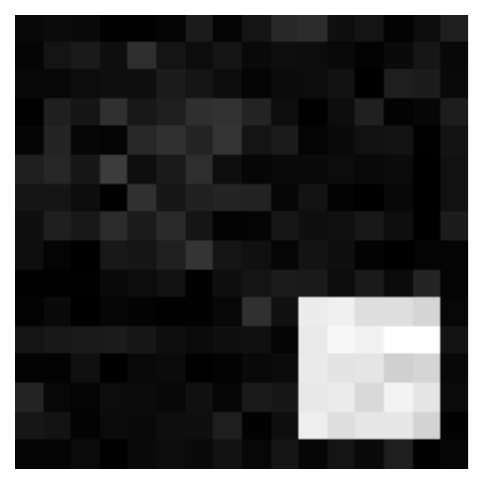}};
\node[below=of img5, node distance = 0, yshift = 1cm] (img2) {\includegraphics[width = 0.15\textwidth]{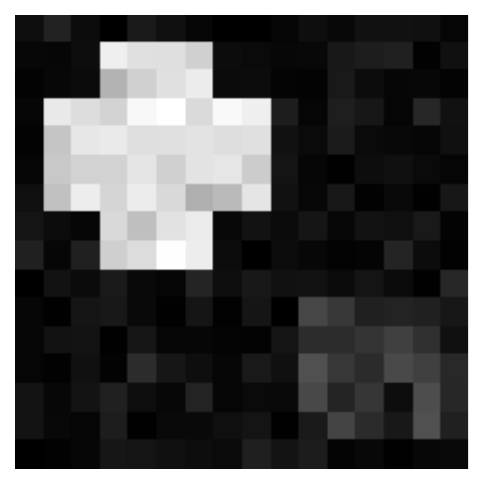}};
\node[right=of img5, node distance = 0, xshift = -1cm] (img7) {\includegraphics[width = 0.15\textwidth]{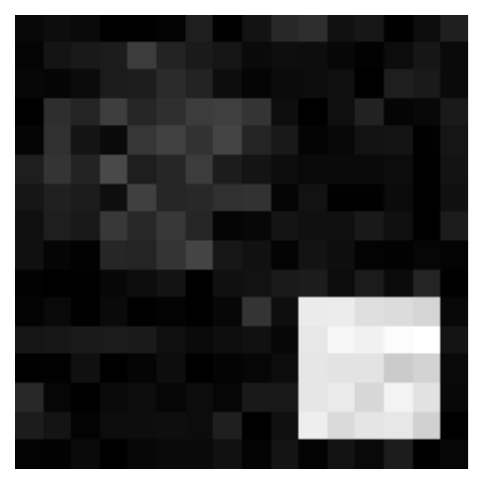}};
\node[below=of img7, node distance = 0, yshift = 1cm] (img2) {\includegraphics[width = 0.15\textwidth]{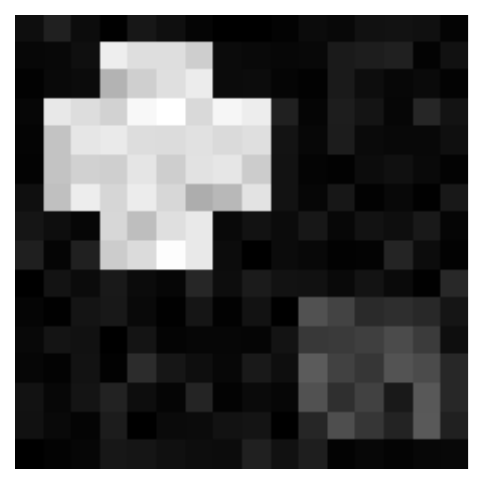}};
\node[right=of img7, node distance = 0, xshift = -1cm] (img9) {\includegraphics[width = 0.15\textwidth]{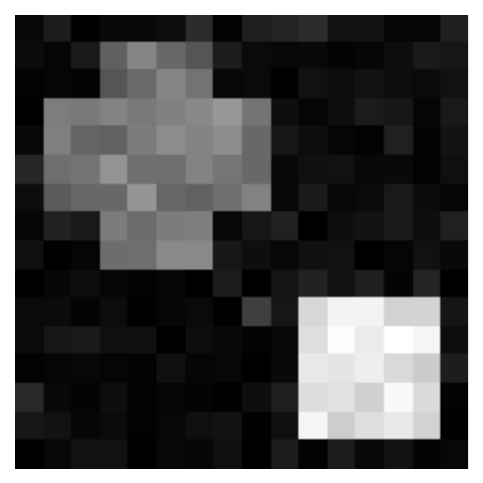}};
\node[below=of img9, node distance = 0, yshift = 1cm] (img10) {\includegraphics[width = 0.15\textwidth]{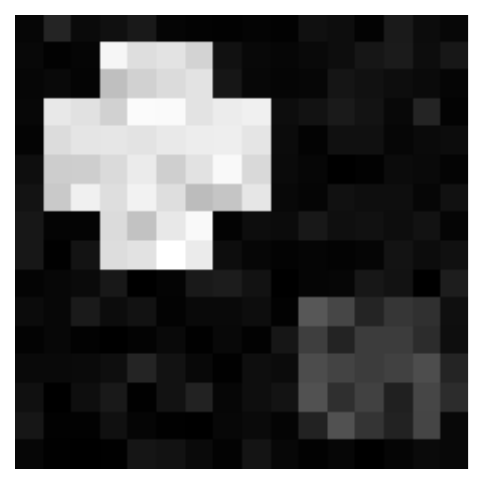}};
\node[above=of img1, node distance = 0, yshift = -1.2cm] {Truth};
\node[above=of img3, node distance = 0, yshift = -1.2cm]{MLMC};
\node[above=of img5, node distance = 0, yshift = -1.2cm] {simple};
\node[above=of img7, node distance = 0, yshift = -1.2cm] {cycling};
\node[above=of img9, node distance = 0, yshift = -1.2cm] {SAEM};
\end{tikzpicture}
\caption{Estimated decomposition images (parameter vector $A$ of $\theta$) with SAEM, MLMC-SGD and Taylor-based sum estimators for the censored logistic ICA model. The first column shows the true value of $a_1, a_2$.}
\label{fig:ica_true_theta}
\end{figure}

\begin{table}
\centering
\begin{tabular}{l|cccccc}
Method & $\mse(\sigma)$ & $\mse(A)$ & $\Exp\left[\textrm{Nsamples}\right]$ & $\var\left(\textrm{Nsamples}\right)$ & runtime (s) \\
\hline\noalign{\smallskip}
SAEM & $7.00\cdot 10^{-7}$ & 0.22 & -- & -- & 33 \\
MLMC & $2.95\cdot 10^{-6}$ & 0.29 & 12 & 19311 & 11\\
simple & $2.97\cdot 10^{-6}$ & 0.27  & 10 & 90 & 10 \\
cycling & $6.80\cdot 10^{-5}$ & 0.24  & 10 & 90 & 14 \\
\end{tabular}
\caption{Comparison of reconstruction accuracy and cost for the censored logistic ICA model. The results are averaged over 100 repetitions.}
\label{tab:sim_ica}
\end{table}

Table~\ref{tab:sim_ica} shows the mean square error ($\mse$) of the estimates
and their cost over 100 repetitions. Figure~\ref{fig:ica_true_theta} shows the
estimated images (parameter $A$ of $\theta$) from one run of each algorithm.
SAEM is more costly than the other methods since it scans the whole data set of
size $n=100$ at each iteration, while the other SGD methods use randomly
selected mini-batches of size 2. The overhead caused by cycling is small
compared to the cost of running SGD for 5000 steps. Although the final MSE is
similar for the four algorithms, on a visual inspection
(Figure~\ref{fig:ica_true_theta}), SGD seems to yield better $A$ estimates than
SAEM.

We also observe that there seems to be a small difference between the results
obtained with the simple and the cycling estimator. Both are
unbiased, however, while running our experiments we noticed that the higher
variance of the simple estimator requires smaller learning rates to
counterbalance the gradient estimate.

\subsection{Exponential random graphs}
\label{sec:erg}
Exponential random graph (ERG) models are one of the classical examples of
models with a doubly intractable posterior distributions. For this model, the data is the observed edge connectivity of an $n$-node undirected graph. In particular, one
observes $y = (y_{ij})_{i<j}$, where $y_{ij} = 1$ (resp. $0$) if nodes
$i$ and $j$ are connected (resp. are not).  The likelihood is, given a
parameter $\theta$ and canonical statistics $s(y)$, both in $\R^k$,
\begin{equation}
	\label{eq:erg_target}
	p(y|\theta) = \exp\left\{ \theta^T s(y) \right\}/\mathcal{Z}(\theta).
\end{equation}
The normalising constant $\mathcal{Z}(\theta)$ is a sum over $2^{\binom{n}{2}}$ terms,
and is therefore intractable.

We construct an unbiased estimate of $1/\mathcal{Z}(\theta)$, for any $\theta$, which may then be combined with the remaining tractable terms of \eqref{eq:erg_target} to perform Bayesian inference and
model choice.
We consider the same example as in \cite{caimo2011bayesian}: an $n=16$
node graph; the statistic is $s(y) = (\sum_{i<j} y_{ij} , \sum_{i<j<k} y_{ij}y_{jk})$, which counts
the number of edges and of two stars; $\theta\in \real^2$ is assigned the prior $N_2(0_2, 30 I_2)$. The data -- the observed connectivity -- is 
the Florentine family business graph \citep{padgett1993robust}. % For the Florentine family business graph we have $ s(y) = (15, 36)$.
% The normalising constant $\mathcal{Z}(\theta) = \sum_{y\in\mathcal{Y}}\exp(\theta s(y))$ requires summing over the set $\mathcal{Y}$, containing all $2^{\binom{n}{2}}$ possible undirected graphs with $n=16$ nodes, and is thus intractable.

For a given $\theta$, we will unbiasedly estimate  $1/\mathcal{Z}(\theta)$ using the simple and cycling estimator applied to the function $f(m)=1/m$. The unbiased estimates of $m=\mathcal{Z}(\theta)$, or inputs  $X_i$ used in the simple and cycling methods,  are obtained by running the waste-free variant
\citep{dau2022waste} of tempered SMC; see Appendix~\ref{app:erg} for more details on
this algorithm. In this example, the cost of generating the inputs with SMC is significantly larger than the cost of assembling the cycling estimate.

% \subsubsection{Variance Comparison (fixed $\theta$)}

We first fix $\theta=(-3.15, 0.58)$, a value in the high probability region of
the posterior, and compare the performance of the simple and the cycling
estimator.

\begin{table}
	\centering
	\begin{tabular}{l|cc|cc|cc|cc}
		             & \multicolumn{4}{c}{Simple}   & \multicolumn{4}{c}{Cycling}                                                                                                                \\
		\hline\noalign{\smallskip}
		             & \multicolumn{2}{c}{Moderate} & \multicolumn{2}{c}{Low}     & \multicolumn{2}{c}{Moderate} & \multicolumn{2}{c}{Low}                                                       \\
		\hline\noalign{\smallskip}
		             & $\wnv$                       & $\mathbb{P}(-)$             & $\wnv$                       & $\mathbb{P}(-)$         & $\wnv$ & $\mathbb{P}(-)$ & $\wnv$ & $\mathbb{P}(-)$ \\
		\hline\noalign{\smallskip}
		$\Exp[R]=1$  & 111                           & $0.026$                      & 106                           & $0.006$                 & 102     & $0.028$         & 77     & $0.002$         \\
		$\Exp[R]=10$ & 61                           & $0.027$                     & 83                           & $0.001$                   & 24     & $0.01$         & 26     & $0.0$
	\end{tabular}

	\caption{Work-normalised variance and proportion of negative
		estimates out of $10^3$ replicates of $\widehat{\mathcal{Z}(\theta)^{-1}}$ for
		the simple and the cycling estimator when the variance of the estimates of
		$\mathcal{Z}(\theta)$ is moderate and low.}
	\label{tab:ergm}
\end{table}

We consider two scenarios, one in which the unbiased estimates of
$\mathcal{Z}(\theta)$ have moderate variance and one in which the variance is
low. See Appendix~\ref{app:erg} for details on the moderate and low variance
regimes. Basically, the latter is obtained by taking 10 times more particles in
SMC, which leads to a (roughly) ten times smaller variance for  the inputs.

We select the value of $x_0$ using the strategy in Section~\ref{subsec:x0} with $n_0=10$.
Once $x_0$ has been selected, we identify the corresponding $\widehat{\beta^2}$
as discussed in Section~\ref{subsec:budget}, see Algorithm \ref{alg:cycling} for a summary. This gives us a range for $p\in(0,
	1-\widehat{\beta^2})$. In the moderate variance regime we obtain $p\in(0, 0.7)$
while for the low variance regime we have $p\in(0, 0.95)$. In both cases, it is
sufficient to set $\Exp[R]=1$ (i.e. $p=0.5$) to ensure that both the simple and
the cycling estimator have finite variance.

Since the cost of obtaining the estimates of $\mathcal{Z}(\theta)$ is
significantly larger than that of computing the estimates,
selecting $\Exp[R]=1$ might be wasteful as our theoretical results in
Section~\ref{sec:general_prop} guarantee that the cycling estimator has a lower
variance than the simple estimator for large values of $\Exp[R]$. Thus, we also
consider $\E[R]=10=n_0$ in our simulations as recommended
in Section~\ref{subsec:budget}.

Table~\ref{tab:ergm} compares the simple and cycling estimators in these four different
scenarios (two variance regimes, two values of $\E[R]$) according to two metrics
estimated on $10^3$ independent runs: work-normalised variance (relative
variance times CPU cost), and proportion of negative $1/\mathcal{Z}(\theta)$ estimates. The results show that
(a) the cycling estimator systematically outperforms the simple one; and (b)
increasing $\E[R]$ does lead to better performance for the cycling estimator
but not for the simpler one, as expected.
The CPU cost of cycling is less than 0.01 seconds higher than that of the simple estimator.

% \subsubsection{Posterior inference and model choice}

To recover the posterior distribution and the marginal likelihood, we use
importance sampling, with a Gaussian proposal $q$, from which we generate $n=1024$
vectors $\theta_i$ through randomised quasi-Monte Carlo. Using quasi-Monte Carlo ensures that the
variability of our importance sampling estimates is dominated by the randomness
of the estimates of $1/\mathcal{Z}(\theta)$. See the appendix for more details
on the proposal and on randomised quasi-Monte Carlo.

The importance sampling weights are 
\begin{align}
	\label{eq:erg_weights}
	w_i = \frac{p(\theta_i)\exp\left\{ \theta_i^T s(y) \right\}}{q(\theta_i)}
	\times \widehat{\frac{1}{\mathcal{Z}(\theta_i)}},\qquad \theta_i \sim q.
\end{align}
The average of these noisy weights gives us an unbiased
estimate of the marginal likelihood, $p(y)=\int p(\theta) p(y|\theta)d\theta$.
To the best of our knowledge, this is the first study that uses unbiased estimates of the
marginal likelihood of an ERG model.

For each $\theta_i$, we use exactly the same approach to estimate
$1/\mathcal{Z}(\theta_i)$ as the one which led to best performance in the
previous section, namely, $n_0=10$, $\E[R]=10$, and waste-free SMC run in the ``moderate
variance'' regime.

Figure~\ref{fig:ergm_posterior} shows the histograms obtained with the importance sampling
approximations of the posterior distribution $p(\theta|y)$; `simple' and `cycling' correspond to \eqref{eq:erg_weights} with $1/\mathcal{Z}(\theta)$ estimated with the simple and cycling methods respectively. The simple (resp. cycling) estimator returns
35 (resp. 10) negative weights out of 1024 samples.
The effective sample size ($\ess$), i.e.$\left(\sum_{i}
	w_i\right)^2 /\sum_{i} (w_i)^2$ is $6\%$ (resp. $31\%$) when using the simple
(resp. cycling) estimator. Using the cycling estimator here, approximately, yields a fivefold increase in computational 
efficiency  with little additional cost.

\begin{figure}
	\centering
	\begin{tikzpicture}[every node/.append style={font=\normalsize}]
		\node (img1) {\includegraphics[width = 0.4\textwidth]{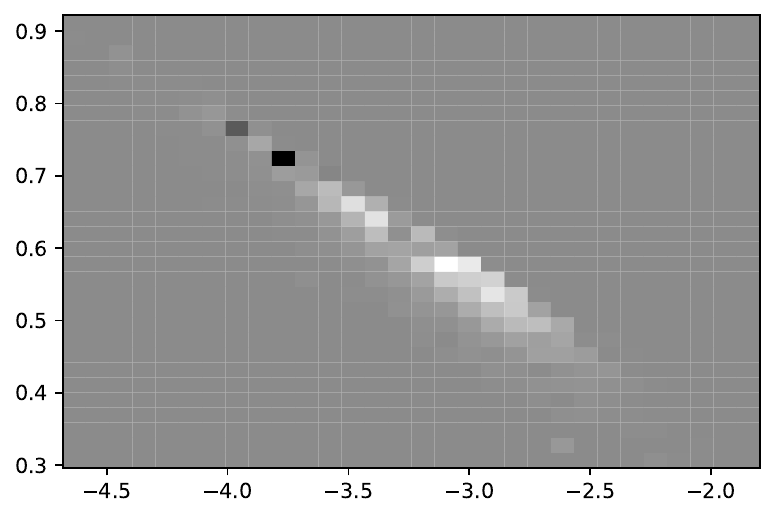}};
		\node[right=of img1, node distance = 0, xshift = -0.7cm] (img2) {\includegraphics[width = 0.4\textwidth]{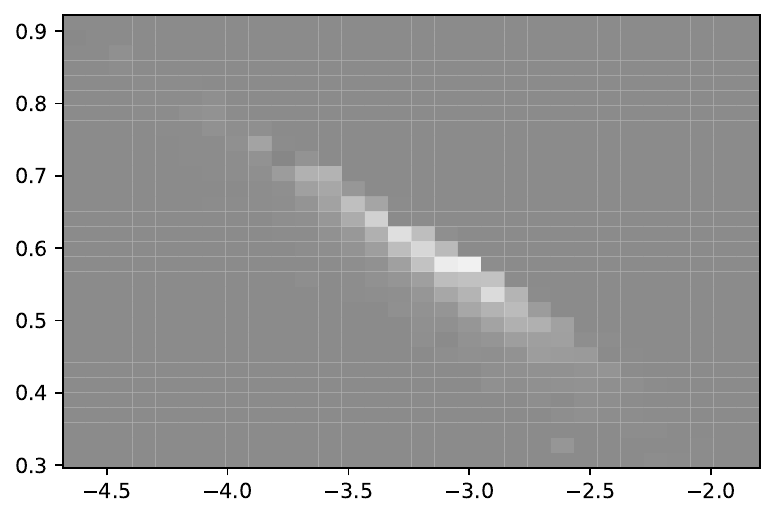}};
		\node[above=of img1, node distance = 0, yshift = -1.2cm] {Simple};
		\node[above=of img2, node distance = 0, yshift = -1.2cm] {Cycling};
		\node[below=of img1, node distance = 0, yshift = 1.2cm] {$\theta_1$};
		\node[below=of img2, node distance = 0, yshift = 1.2cm] {$\theta_1$};
		\node[left=of img1, node distance = 0, rotate = 90, anchor = center, yshift = -0.8cm] {$\theta_2$};
		\node[left=of img2, node distance = 0, rotate = 90, anchor = center, yshift = -0.8cm] {$\theta_2$};
	\end{tikzpicture}
	\caption{Bivariate weighted histograms approximating the posterior distributions
		obtained with the simple and the cycling estimator using $n=1024$ samples from proposal $q$.
	}
	\label{fig:ergm_posterior}
\end{figure}
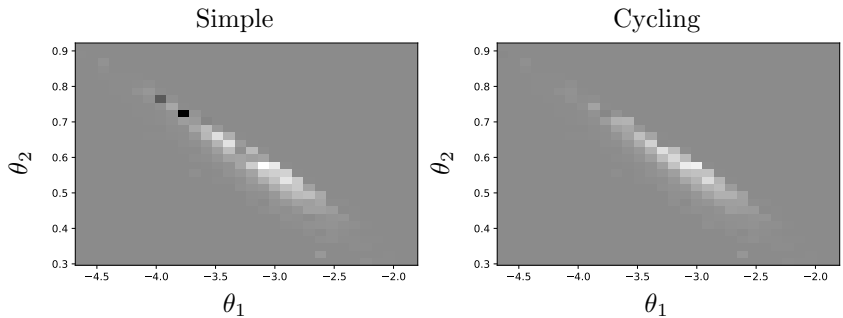

To assess the variability of our approach we repeat the estimation of
$1/\mathcal{Z}(\theta)$, for each $\theta_i$, one hundred times,
resulting in 100 different posterior approximations obtained from the same set
of $n=1024$ samples obtained from the proposal.
For each repetition, we compute the $\ess$, and the unbiased approximation of the model evidence $p(y)$
given by $n^{-1}\sum_{i=1}^n w(\theta_i)$
(Figure~\ref{fig:ergm_evidence}).
The $\ess$ obtained with the cycling estimator is generally higher than that obtained with the simple estimator, and is above 0.2 in 74 out of  the 100 repetitions. The average number of negative weights for a repetition is 44 for the simple estimator and 18 for the cycling estimator.

Low values of the $\ess$ correspond to outliers in the distribution of the
evidence, in particular, negative values of $p(y)$ are more
frequent with the simple estimator than with the cycling estimator. The
resulting variance for the estimates of $p(y)$ is 6 times larger for the simple
estimator than for the cycling estimator.
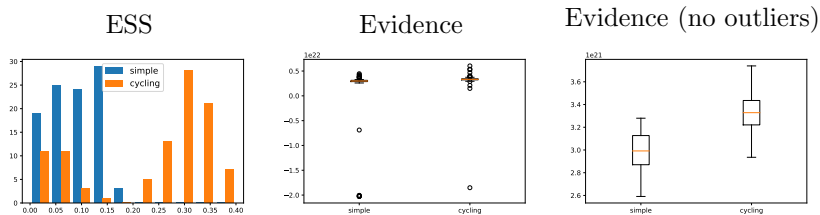
\begin{figure}
	\centering
	\begin{tikzpicture}[every node/.append style={font=\normalsize}]
		\node (img1) {\includegraphics[width = 0.3\textwidth]{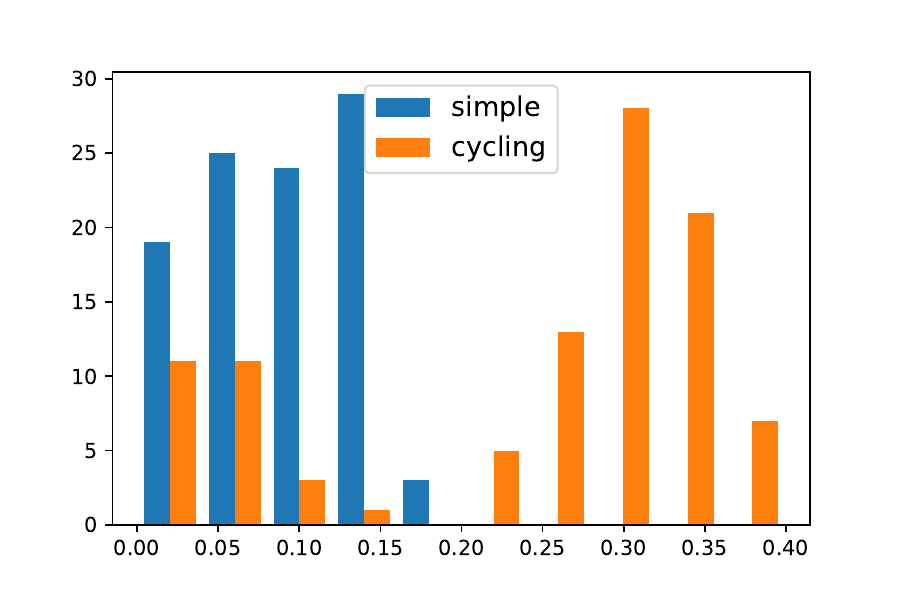}};
		\node[above=of img1, node distance = 0, yshift = -1.2cm] {$\ess$};
		\node[right=of img1, node distance = 0, xshift = -1.3cm] (img2) {\includegraphics[width = 0.3\textwidth]{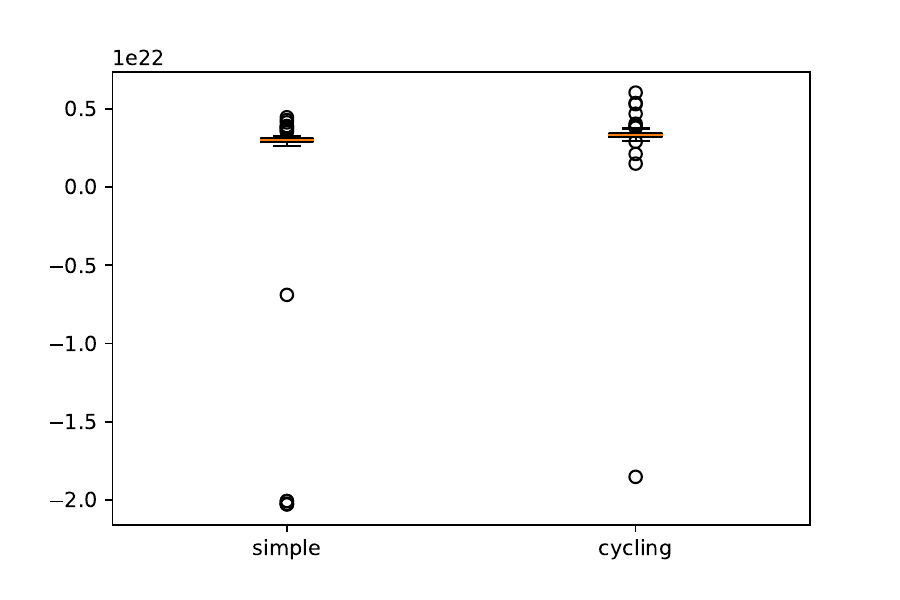}};
		\node[above=of img2, node distance = 0, yshift = -1.2cm] { Evidence};
\node[right=of img2, node distance = 0, xshift = -1.3cm] (img3) {\includegraphics[width = 0.3\textwidth]{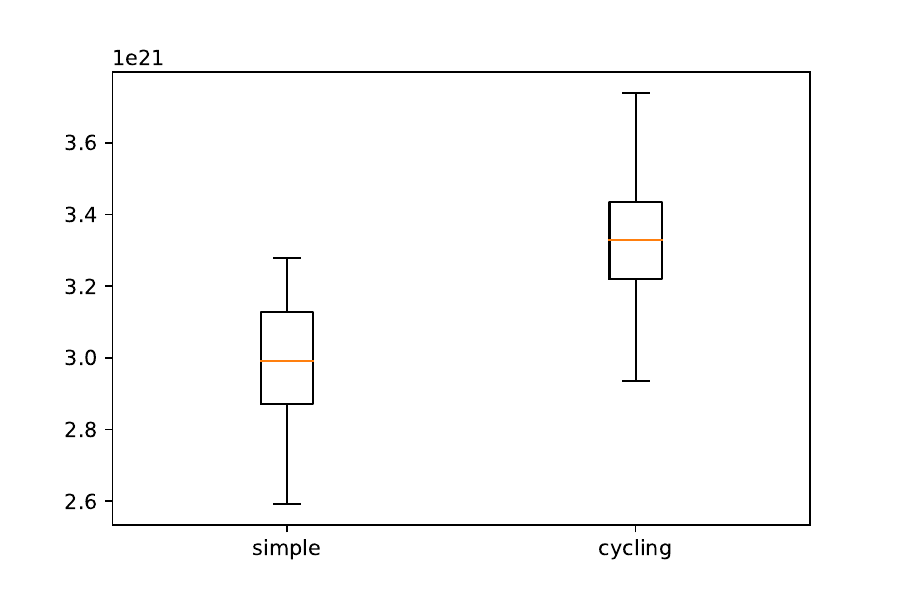}};
		\node[above=of img3, node distance = 0, yshift = -1.2cm] { Evidence (no outliers)};
	\end{tikzpicture}
	\caption{Distribution of ESS and evidence (marginal likelihood) for 100
		repetitions of the posterior approximation using the simple and the cycling
		estimator. }
	\label{fig:ergm_evidence}
\end{figure}

%%% Local Variables:
%%% mode: latex
%%% TeX-master: "paper.tex"
%%% End:

\section{Conclusion}
\label{sec:conclusion}

We develop a general strategy for unbiasedly estimating $f(m)$, for a smooth
function $f$, from  random variables $X_i$ with expectation $m$. Inspired by
the simple estimator (Section~\ref{subsec:simple}) -- which is widely used in
the literature -- we proposed a new cycling estimator
(Section~\ref{subsec:cycle}), and obtained conditions that guarantee finite
estimator and execution time (CPU cost) variance for both methods. This is an
improvement over previous approaches. In particular, the  multi-level Monte
Carlo method of \cite{Luo2020SUMO} returns estimates with infinite variance,
while the one of \cite{shi2021multilevel} has a running time with infinite
variance. The latter issue may be particularly problematic in a parallel
implementation since the completion time is the maximum of all threads being
run.

As pointed out by a referee, the fact that the execution cost is random may be 
a drawback in certain implementation scenarios, even if it has finite variance.
In such a case, one may prefer using a biased estimate whose execution
cost is deterministic. Such an approach has less theoretical support, but seems
to work well in certain practical applications, e.g., to approximate gradients
in variational auto-encoders \citep{burda2015importance}.

We show that the cycling estimator outperforms the simple one both
theoretically, see Proposition~\ref{prop:exp_var_convergence_cycle} and
Example~\ref{ex:simple_var} and empirically, see Section~\ref{sec:numerical},
with potentially only a little more computational cost. We also discussed the
sensitivity of these methods to their tuning parameters, an issue we found to
be largely overlooked in the literature. We gave precise guidelines for
selecting the expansion point $x_0$ and the distribution of the random variable
$R$ that determines the random truncation (Sections~\ref{subsec:x0}
and~\ref{subsec:budget}).

Our experiments show that, in order to build reliable unbiased estimates of
$f(m)$ or its gradient, one needs input variables $X_i$ with low or moderate
variance. In fact, the minimal $\beta^2$ in
Proposition~\ref{prop:exp_var_convergence_simple}, achieved for $x_0^\star$
in~\eqref{eq:optimal_x0}, is $(\beta^2)^\star = \sigma^2/(\sigma^2+m^2)$. Note
that $1-(\beta^2)^\star$ is indeed the inverse of relative variance of the
$X_i$'s. As
$\sigma^2/m^2\rightarrow \infty$, we need to select $p\in(0, 1-(\beta^2)^\star)
\to \{0\}$, hence requiring an infinite computational cost. In this case, we
recommend dedicating the computational effort towards building lower variance
estimates of $m$.

Even if the $X_i$'s have low variance, a poorly chosen $x_0$ can cause the
estimators of $f(m)$ to have high variance. While in low dimensional problems
one can get away with a poorly tuned $x_0$, we found that for high dimensional
applications, e.g., the ICA example in Section~\ref{sec:ica}, it is crucial
that the input variables $X_i$'s have low variance \emph{and} to fine tune
$x_0$,  as recommended in Section~\ref{subsec:x0}. Once $x_0$ has been
selected, we describe a simple strategy to select $p$ to guarantee that $\E[R]$
is sufficiently large to have finite and small variance. This strategy works
well for the experiments we considered, however, in some cases it might be
possible to build optimal distributions for particular $f$'s (see, 
e.g.~\cite{rhee2015unbiased, cui2021optimal, zheng2023optimal}).

As a by-product of our study, we provide a method to unbiasedly estimate the
model evidence in doubly intractable models. Previous attempts in the context
of ERG models have focused on pseudo-likelihood approaches in which the
intractable likelihood is replaced by a tractable approximation
\citep{bouranis2018bayesian}. {Also, using} a combination of path sampling and kernel
density estimation, which provides biased estimators and are only suitable in
low-dimensional settings \cite[Section 4]{caimo2013bayesian}.

\section*{Acknowledgments} 

Research of the first two authors is partly supported by the French Agence
Nationale de la Recherche (ANR) through grant EPISVACAGE.  SSS holds the Tibra
Foundation professorial chair and gratefully acknowledges research funding as
follows: This material is based upon work supported by the Air Force Office of
Scientific Research under award number FA2386-23-1-4100.  The authors would
like to thank Otmane Sakhi, Adrien Corenflos Sam Power and two referees for
helpful comments at various stages of this project.

\bibliographystyle{apalike}
\bibliography{unbiased_biblio}

\appendix
\section{Integrability and moment expressions}
\label{app:integrability}

\subsection{Proof of Lemma \ref{lem:integrablefhat}} \label{app:integrablefhat}

We will show
\begin{align}
\label{eq:fubini_sumest}
\E\left[\sum_{k=0}^\infty \frac{\ind{R\geq k}}{\pr(R\geq k)}\vert \gamma_k \vert\vert \URk\vert \right] <\infty,
\end{align}
which implies  the infinite series on the right-hand side of \eqref{eq:def_sumest} is, almost surely, absolutely summable. Thus, $\sumest$ exists almost surely and is integrable. To show it has the desired mean, as~\eqref{eq:fubini_sumest} holds, we can apply Fubini's Theorem for general functions \cite[Theorem 18.3]{billingsley1995measure} to  swap the expectation and the sum and obtain
\begin{align*}
\E[\sumest] = \sum_{k=0}^\infty \frac{\E\left[\ind{R\geq k} \URk\right]}{\pr(R\geq k)}\gamma_k = \sum_{k=0}^\infty\gamma_k\left(\frac{m}{x_0}-1\right)^k.
\end{align*}

To show  \eqref{eq:fubini_sumest}, we note that all the terms in the sum are positive. Fubini's Theorem for non-negative functions \cite[Theorem 18.3]{billingsley1995measure} allows the swap of the expectation and the sum,
\begin{align*}
  \E\left[\sum_{k=0}^\infty \frac{\ind{R\geq k}}{\pr(R\geq k)}\vert \gamma_k \vert \vert\URk\vert \right]
  &= \sum_{k=0}^\infty \frac{\E\left[\ind{R\geq k} \vert \URk\vert \right]}{\pr(R\geq k)}\vert \gamma_k\vert
  \leq \sum_{k=0}^\infty \vert \gamma_k\vert a_k <\infty,
\end{align*}
where we used
\begin{align*}
\E\left[\ind{R\geq k} \vert \URk\vert \right] 
%&=\sum_{r=0}^\infty\pr(R=r) \E\left[\ind{R\geq k} \vert \URk\vert |R=r\right]\\
&=\sum_{r=k}^\infty \E\left[\ind{R=r} \vert \Urk\vert \right]\\
&=\sum_{r=k}^\infty \pr(R=r) \E\left[ \vert \Urk\vert \right]\\
&\leq a_k\sum_{r=k}^\infty\pr(R=r).
\end{align*} Note that the independence of $R$ and $\Urk$ was invoked in the second line.

\subsection{Verification of \eqref{eq:exp_fhat_given_r}%Conditional expectation of $\sumest$
\label{app:exp_fhat_given_r}}
The conditional expectation exists, and is integrable, since $\sumest$ itself is integrable.
By definition, $\E\left[ \sumest \middle| R=r\right]=\E\left[ \sumest \ind{R = r}\right]/\pr(R=r)$, provided $\pr(R=r)>0$. It follows from \eqref{eq:fubini_sumest} that 
\begin{align}
\label{eq:fubini_sumest_conditional_2}
 \E\left[\sum_{k=0}^\infty \frac{\ind{R\geq k}}{\pr(R\geq k)}\vert \gamma_k \vert\vert \URk\vert
\ind{R = r}\right] <\infty.
\end{align}
Since~\eqref{eq:fubini_sumest_conditional_2} holds, we can apply Fubini's Theorem for general functions \cite[Theorem 18.3]{billingsley1995measure} and obtain
\[ \E\left[ \sumest \middle| R=r\right]
	= \sum_{k=0}^\infty \gamma_k\frac{\ind{r \geq k}}{\pr(R\geq k)} \E[\Urk | R=r]
	=\sum_{k=0}^\infty \gamma_k\frac{\ind{r \geq k}}{\pr(R\geq k)} \left(
	\frac{m}{x_0} -1 \right)^k.
\]
Note the denominator is non-zero, for $r\geq k$, since $\pr(R=r)>0$.

\subsection{Verification of \eqref{eq:varE}}
\label{app:varE}
We now show that $\E\left[ \sumest | R=r\right]$ in \eqref{eq:exp_fhat_given_r} admits the conditional variance expression~\eqref{eq:varE}.
We will do so under the following sufficient condition which will permit the interchange of the expectation operator and the infinite sum:
\begin{align}
\sum_{k=0}^\infty\sum_{l=k+1}^\infty |\gamma_k| |\gamma_l | \frac{1}{\pr(R\geq k)} \left|
	\frac{m}{x_0} -1 \right|^{k+l} < \infty.
\label{eq:verify_varE}
\end{align}
We only consider the expression for the second moment since $\E[\E[\sumest|R]]^{2}=f(m)^{2}$. 
Squaring  \eqref{eq:exp_fhat_given_r} -- noting that it is a finite series -- and replacing $r$ with the (integer valued) random variable $R$ gives
\begin{align}
\label{eq:conditional_expe_squared}
\E[ \sumest | R]^2 = & \sum_{k=0}^\infty \gamma_k^2 \frac{\ind{R \geq k}}{\pr(R\geq k)^2} \left(
	\frac{m}{x_0} -1 \right)^{2k} \\
	&+ 2\sum_{k=0}^\infty\sum_{l=k+1}^\infty \gamma_k\gamma_l\frac{\ind{R \geq l}}{\pr(R\geq k)
	\pr(R\geq l)}\left(\frac{m}{x_0} -1 \right)^{k+l}.
\end{align}
Since the terms in the first sum are all positive, Fubini's Theorem for non-negative functions \cite[Theorem 18.3]{billingsley1995measure} allows the interchange of expectation and the sum,
\begin{align*}
\E\left[\sum_{k=0}^\infty \gamma_k^2\frac{\ind{R \geq k}}{\pr(R\geq k)^2} \left(
	\frac{m}{x_0} -1 \right)^{2k}\right] &= \sum_{k=0}^\infty\E\left[ \gamma_k^2\frac{\ind{R \geq k}}{\pr(R\geq k)^2} \E[\URk|R]^{2}\right]\\
	&= \sum_{k=0}^\infty \gamma_k^2\frac{1}{\pr(R\geq k)} \left(
	\frac{m}{x_0} -1 \right)^{2k}.
\end{align*}
For the second sum, assumption \eqref{eq:verify_varE} also permits the interchange of expectation and the sum. Thus, 
\begin{align*}
&\E\left[\sum_{k=0}^\infty\sum_{l=k+1}^\infty  \gamma_k\gamma_l \frac{\ind{R \geq l}}{\pr(R\geq k)\pr(R\geq l)} 
\left(\frac{m}{x_0} -1 \right)^{k+l}\right]\\
&=\sum_{k=0}^\infty\sum_{l=k+1}^\infty\E\left[  \gamma_k\gamma_l \frac{\ind{R \geq l}}{\pr(R\geq k)\pr(R\geq l)} 
\left(\frac{m}{x_0} -1 \right)^{k+l}\right]
\end{align*}
This completes the verification of \eqref{eq:varE}.

Finally, condition \eqref{eq:verify_varE} holds when the distribution of $R$ satisfies certain conditions. For example, precisely under all the conditions of Proposition~\ref{prop:var_exp}.

\subsection{Derivation of~\eqref{eq:Evar}}
\label{app:derivation_variance}
We start by computing $\E[(\sumest)^2|R=r]$:
\begin{align*}
\E[(\sumest)^2|R=r]&= \E\left[\left(\sum_{k=0}^\infty \frac{\ind{R\geq k}}{\pr(R\geq k)}\gamma_k \URk\right)^2 |R=r\right]\\
&= \E\left[\sum_{k=0}^\infty \frac{\ind{R\geq k}}{\pr(R\geq k)^2}\gamma_k^2 \URk^2|R=r\right]\\
&+\E\left[2\sum_{k=0}^\infty\sum_{l=k+1}^\infty \frac{\ind{R\geq k}\ind{R\geq l}}{\pr(R\geq k)\pr(R\geq l)}\gamma_k\gamma_l \URk\URll |R=r\right]\\
&= \sum_{k=0}^\infty \frac{\ind{r\geq k}}{\pr(R\geq k)^2}\gamma_k^2 \Urk^2+2\sum_{k=0}^\infty\sum_{l=k+1}^\infty \frac{\ind{r\geq l}}{\pr(R\geq k)\pr(R\geq l)}\gamma_k\gamma_l \Urk\Urll.
\end{align*}
This and~\eqref{eq:conditional_expe_squared} give
\begin{align*}
\var[\sumest|R=r ] &=\sum_{k=0}^\infty \frac{\ind{r\geq k}}{\pr(R\geq k)^2}\gamma_k^2 \left(\Urk^2-\left(
	\frac{m}{x_0} -1 \right)^{2k}\right)\\
&+2\sum_{k=0}^\infty\sum_{l=k+1}^\infty \frac{\ind{r\geq l}}{\pr(R\geq k)\pr(R\geq l)}\gamma_k\gamma_l \left(\Urk\Urll-\left(
	\frac{m}{x_0} -1 \right)^{k+l}\right)\\
	&=\sum_{k=0}^\infty \frac{\ind{r\geq k}}{\pr(R\geq k)^2}\gamma_k^2 \var(\Urk) \\
&+2\sum_{k=0}^\infty\sum_{l=k+1}^\infty \frac{\ind{r\geq l}}{\pr(R\geq k)\pr(R\geq l)}\gamma_k\gamma_l \cov\left( U_{r,k}, U_{r,l}  \right).
\end{align*}
Taking expectations w.r.t. $R$ gives the result.
%%% Local Variables: 
%%% mode: latex
%%% TeX-master: "paper.tex"
%%% End: 

\section{Proofs of Section~\ref{sec:general_prop}}

\subsection{Proof of Proposition~\ref{prop:var_exp}}
\label{app:proof_prop_var_exp}

Let $\tilde{f} = \E[\sumest|R]$.
The lower bound can be found by the law of total variance: $\var[\tilde{f}] \geq \var
\left[ \E [\tilde{f} \vert \ind{R>0}]\right]$ and
$\E [\tilde{f} \vert \ind{R>0}]
= \gamma_0 + \frac{\ind{R>0}}{\pr(R>0)} (f(m)-\gamma_0)$.
In the case of $R\sim\Geomp$ we have
\begin{align*}
\var[\tilde{f}] &\geq (f(m)-\gamma_0)^2\left( \frac{1}{\pr(R> 0) }  -1 \right)\\
&=(f(m)-\gamma_0)^2\left( \frac{1}{1-p }  -1 \right).
\end{align*}

For the upper bound, under Assumption~\ref{hyp:f_and_x0},
%and using the fact 
%\[ 	\left(\frac{1}{\pr(R\geq k)} - 1 \right) \geq 0\]
the first sum \eqref{eq:varE} can be bounded as follows:
\begin{align*}
\sum_{k=0}^\infty \gamma_k^2  \beta_0^{2k}
	\left(\frac{1}{\pr(R\geq k)} - 1 \right) &\leq
c^2\sum_{k=0}^\infty\left[\left(\frac{\beta_0^2}{1-p}\right)^k -\beta_0^{2k}\right] \\
& = c^2\left(\frac{1-p}{1-p-\beta_0^2} -\frac{1}{1-\beta_0^2}\right)\\
&= \frac{c^2p\beta_0^2}{(1-\beta_0^2)(1-p-\beta_0^2)},
\end{align*}
where the convergence of the infinite sum is guaranteed by the fact that $\beta_0 <1$ and $1-p>\beta_0^2$.

The absolute value of the second term may be bounded in a similar manner:
\begin{align*}
2 \sum_{k=0}^\infty\sum_{l=k+1}^\infty \left| \gamma_k \gamma_l \right| \beta_0^{k+l}
	\left(\frac{1}{\pr(R\geq k)} - 1 \right) &\leq 2c^2\sum_{k=0}^\infty
	\left(\frac{1}{\pr(R\geq k)} - 1 \right)\sum_{l=k+1}^\infty \beta_0^{k+l}\\
	&=\frac{2c^2 \beta_0 }{1-\beta_0 }\sum_{k=0}^\infty \beta_0^{2k}
	\left(\frac{1}{\pr(R\geq k)} - 1 \right)\\
	&=\frac{2 \beta_0 }{1-\beta_0 }\times\frac{c^2p\beta_0^2}{(1-\beta_0^2)(1-p-\beta_0^2)}.
\end{align*}

Combining the two inequalities leads to:
\begin{align*}
	\var\left[ \E[\sumest|R ] \right]
  & \leq \left(1+\frac{2 \beta_0 }{1-\beta_0
    }\right)\frac{c^2p\beta_0^2}{(1-\beta_0^2)(1-p-\beta_0^2)} \\
	& \leq \frac{c^2\beta_0^2}{(1-\beta_0 )^2}\times \frac{p}{1-p-\beta_0^2}.
\end{align*}

\subsection{Proof of Proposition~\ref{prop:exp_var_convergence_simple}}
\label{app:proof_prop_exp_var_convergence_simple}
\subsubsection{Proof of $\hat{f}^{S}\in\mathcal{L}_{2}$}

For brevity, let $\hat{f}$ denote $\hat{f}^{S}$. By the monotone convergence
theorem, $\mathbb{E}[ \hat{f}^{2}] =\sum_{r=0}^{\infty}\mathbb{E}[ \hat{f}^{2}\ind{R = r}] .$
By definition of $\hat{f}$ (see \eqref{eq:def_sumest}), $\hat{f}\ind{R = r}=\sum_{k=0}^{r}\gamma_{k}U_{r,k}\ind{R = r}/\mathbb{P}(R\geq k)$.
Thus, 
\begin{align*}
\hat{f}^{2}\ind{R = r} & =\sum_{k=0}^{r}\frac{\gamma_{k}^{2}U_{r,k}^{2}}{\mathbb{P}(R\geq k)^{2}}\ind{R = r}+2\sum_{k=0}^{r-1}\sum_{l=k+1}^{r}\ind{R = r}\frac{\gamma_{k}U_{r,k}}{\mathbb{P}(R\geq k)}\frac{\gamma_{l}U_{r,l}}{\mathbb{P}(R\geq l)}\\
\mathbb{E}[\hat{f}^{2}\ind{R = r}]  & =\sum_{k=0}^{r}\frac{\gamma_{k}^{2}\mathbb{E}[U_{r,k}^{2}]}{\mathbb{P}(R\geq k)^{2}}\mathbb{E}[\ind{R = r}]+2\sum_{k=0}^{r-1}\sum_{l=k+1}^{r}\frac{\gamma_{k}\gamma_{l}\mathbb{E}[\ind{R = r}]}{\mathbb{P}(R\geq k)\mathbb{P}(R\geq l)}\mathbb{E}[U_{r,k}U_{r,l}].
\end{align*}
The independence of $R$ and $\Urk$ has been invoked in the final expression for $\mathbb{E}[ \hat{f}^{2}\ind{R = r}]$. For the simple estimate $\hat{f}^{S}$,
\begin{align*}
\Exp\left[\left( \Urk \right)^2\right]= \left(\frac{\sigma^2}{x_0^2}+\left(\frac{m}{x_0}-1\right)^2\right)^k = \beta^{2k}.
\end{align*}
Along with the Cauchy-Schwarz inequality, gives, $|\mathbb{E} [ U_{r,k}U_{r,l}]|\leq \mathbb{E}[U_{r,k}^2]^{1/2}\mathbb{E}[U_{r,l}^2]^{1/2}\leq\beta^{l+k}$. Thus, 
\begin{align*}
 \mathbb{E}[ \hat{f}^{2}\ind{R = r}]& \leq\mathbb{P}\left(R=r\right)\sum_{k=0}^{r}\frac{\gamma_{k}^{2}\beta^{2k}}{\mathbb{P}(R\geq k)^{2}}+2\mathbb{P}\left(R=r\right)\sum_{k=0}^{r-1}\sum_{l=k+1}^{r}\frac{\gamma_{k}\gamma_{l}}{\mathbb{P}(R\geq k)\mathbb{P}(R\geq l)}\beta^{k+l}\\
\sum_{r=1}^{t}\mathbb{E}[\hat{f}^{2}\ind{R = r}]  & \leq\sum_{r=1}^{t}\mathbb{P}\left(R=r\right)\sum_{k=0}^{r}\frac{\gamma_{k}^{2}\beta^{2k}}{\mathbb{P}(R\geq k)^{2}}\\
 & \qquad+\sum_{r=1}^{t}2\mathbb{P}\left(R=r\right)\sum_{k=0}^{r-1}\sum_{l=k+1}^{r}\frac{\gamma_{k}\gamma_{l}}{\mathbb{P}(R\geq k)\mathbb{P}(R\geq l)}\beta^{k+l}
\end{align*}
The first sum can be expressed as 
\begin{align*}
\sum_{r=0}^{t}\mathbb{P}\left(R=r\right)\sum_{k=0}^{r}\frac{\gamma_{k}^{2}\beta^{2k}}{\mathbb{P}(R\geq k)^{2}} & =\sum_{k=0}^{t}\frac{\gamma_{k}^{2}\beta^{2k}}{\mathbb{P}(R\geq k)^{2}}\sum_{r=k}^{t}\mathbb{P}\left(R=r\right)\\
 & \leq\sum_{k=0}^{t}\frac{\gamma_{k}^{2}\beta^{2k}}{\mathbb{P}(R\geq k)}\\
 & \leq\sum_{k=0}^{\infty}\frac{c^{2}\beta^{2k}}{(1-p)^{k}}\\
 & =c^{2}\frac{1-p}{1-p-\beta^{2}},
\end{align*} where the final line uses condition $1-p > \beta^2$.
Similarly, the double  sum in the right-hand side of $\sum_{r=1}^{t}\mathbb{E}[ \hat{f}^{2}\ind{R = r}] $ can be expressed, and then upper
bounded as follow: 
\begin{align*}
 & \sum_{r=1}^{t}2\mathbb{P}\left(R=r\right)\sum_{k=0}^{r-1}\sum_{l=k+1}^{r}\frac{\gamma_{k}\gamma_{l}}{\mathbb{P}(R\geq k)\mathbb{P}(R\geq l)}\beta^{k+l}\\
 & =\sum_{k=0}^{t-1}\sum_{l=k+1}^{t}\frac{\gamma_{k}\gamma_{l}}{\mathbb{P}(R\geq k)\mathbb{P}(R\geq l)}\beta^{k+l}\sum_{r=l}^{t}2\mathbb{P}\left(R=r\right)\\
 & \leq\sum_{k=0}^{t-1}\sum_{l=k+1}^{t}\frac{\gamma_{k}\gamma_{l}}{\mathbb{P}(R\geq k)}\beta^{k+l}2\\
 & \leq2c^{2}\sum_{k=0}^{t-1}\sum_{l=k+1}^{t}\frac{\beta^{k+l}}{\mathbb{P}(R\geq k)}\\
 & \leq2c^{2}\frac{\beta}{1-\beta}\frac{1-p}{1-p-\beta^{2}}.
\end{align*}
Thus, $\sum_{r=1}^{t}\mathbb{E}[ \hat{f}^{2}\ind{R = r}] $
is bounded uniformly in $t$ by 
\begin{equation}
c^{2}\frac{1-p}{1-p-\beta^{2}}+2c^{2}\frac{\beta}{1-\beta}\frac{1-p}{1-p-\beta^{2}},
\label{eq:proof_simple_est_in_L2}
\end{equation}
and so $\mathbb{E}[\hat{f}^{2}]<\infty$.

\subsubsection{Proof of \eqref{eq:prop_exp_var_convergence_simple}}
We may bound~\eqref{eq:Evar} by the second moment,
\[
\E\left[ \var[\sumest|R ] \right] = \E\left[ \E[\sumest^2|R ] \right] - \E\left[ \E[\sumest |R ]^2 \right] \leq  \E\left[ \sumest^2 \right].
\] Using \eqref{eq:proof_simple_est_in_L2},
\begin{align*}
\mathbb{E}[\hat{f}^{2}] &
    \leq c^2 \left(1+\frac{2\beta}{1-\beta}\right)\frac{1-p}{1-p-\beta^2}\\
	&=c^2\frac{1+\beta}{1-\beta}\frac{1-p}{1-p-\beta^2}.
\end{align*}

\subsection{Verification of Example \ref{ex:simple_var} }
\label{app:counter}
Recall that by the law of total probability we have $$\var[\sumestS] = \var\left[ \E[\sumestS|R ] \right]	+ \E\left[ \var[\sumestS | R ] \right].$$ If we let $p\to 0$, Proposition~\ref{prop:var_exp} guarantees that $\var\left[ \E[\sumestS|R ]\right]\to 0$ regardless of the choice of $f, m, x_0$. 
The behaviour of $\E\left[ \var[\sumestS | R ] \right]$ as $p$ tends to $0$ can be
obtained by invoking the dominated convergence theorem to interchange $\lim_{p\to 0}$ and the sums  
%where the integral is takenw.r.t the counting measure on sequence space $l^1$ 
(a result also known as Tannery's Theorem; \citealp{boas1965tannery}):
\begin{align*}
% \label{eq:exp_var_limit}
  \E\left[ \var[\sumestS|R ] \right]
  &= \sum_{k=0}^\infty \frac{\gamma_k^2}{\pr(R\geq k)^2}
	\left\{ \sum_{r=k}^\infty \pr(R=r) \var(\UrkS)  \right\}\\
	&\quad + 2 \sum_{k=0}^{\infty}\sum_{l=k+1}^{\infty} \frac{\gamma_k
		\gamma_l}{\pr(R\geq k)\pr(R\geq l)}
	\left\{  \sum_{r=l}^{\infty}\pr(R=r)\cov\left( \UrkS, \UrlS \right)\right\}\notag\\
	&= \sum_{k=0}^\infty \frac{\gamma_k^2}{\pr(R\geq k)}
	\left\{ \left(\frac{\sigma^2}{x_0^2}+\left(\frac{m}{x_0}-1\right)^2\right)^k
	-\left(\frac{m}{x_0}-1\right)^{2k}\right\}\notag\\
	& \quad + 2 \sum_{k=0}^{\infty}\sum_{l=k+1}^{\infty} \frac{\gamma_k
		\gamma_l}{\pr(R\geq k)}
	\left\{  \left(\frac{\sigma^2}{x_0^2}+\left(\frac{m}{x_0}-1\right)^2\right)^k
	\left(\frac{m}{x_0} - 1\right)^{l-k}
	-\left(\frac{m}{x_0}-1\right)^{k+l}\right\}\notag\\
 & \to \sum_{k=0}^\infty \gamma_k^2
	\left\{ \left(\frac{\sigma^2}{x_0^2}+\left(\frac{m}{x_0}-1\right)^2\right)^k
	-\left(\frac{m}{x_0}-1\right)^{2k}\right\}\notag\\
	& \quad + 2 \sum_{k=0}^{\infty}\sum_{l=k+1}^{\infty} \gamma_k
		\gamma_l
	\left\{  \left(\frac{\sigma^2}{x_0^2}+\left(\frac{m}{x_0}-1\right)^2\right)^k
	\left(\frac{m}{x_0} - 1\right)^{l-k}
	-\left(\frac{m}{x_0}-1\right)^{k+l}\right\}.\notag
\end{align*}

%To see that $\lim_{p\to0}\E\left[ \var[\sumestS|R ] \right]$ can take non-zero values we consider the following example:
%\begin{example}
%Consider $f(x) =1/x$ so that $\gamma_k = (-1)^{k-1}$ and assume that we want to estimate $f(m)$ for $m>0$. Then, we have
For $f(x) =1/x$, $\gamma_k = (-1)^{k-1}$, and substituting into the above expression gives
\begin{align*}
x_0\lim_{p\to0}\E\left[ \var[\sumestS|R ] \right] &= \sum_{k=0}^\infty
	\left\{ \beta^{2k}
	-\beta_0^{2k}\right\}\notag \\
	& \quad + 2 \sum_{k=0}^{\infty}(-1)^{k-1}\left\{  \beta^{2k}
	\left(\frac{m}{x_0} - 1\right)^{-k}
	-\left(\frac{m}{x_0}-1\right)^{k}\right\}\sum_{l=k+1}^{\infty} (-1)^{l-1}\left(\frac{m}{x_0}-1\right)^{l}\\
	& = \frac{1}{1-\beta^2}-\frac{1}{1-\beta_0^2}+2 \left(\frac{x_0}{m}-1\right)\sum_{k=0}^{\infty}\left\{\beta^{2k}
	-\left(\frac{m}{x_0}-1\right)^{2k}\right\}\\
	& = \frac{1}{1-\beta^2}-\frac{1}{1-\beta_0^2}+2 \left(\frac{x_0}{m}-1\right)\left\{  \frac{1}{1-\beta^2}
	-\frac{1}{1-\beta_0^2}\right\}\\
	& = \left(\frac{2x_0}{m}-1\right)\left\{  \frac{1}{1-\beta^2}
	-\frac{1}{1-\beta_0^2}\right\}.
\end{align*}
%If $\sigma^2=0$, i.e. $U_{r,k}=(m/x_0 - 1)^k$, then we have $\beta^2 = \beta_0^2$ and $\lim_{p\to0}\E\left[ \var[\sumestS|R ] \right] = 0$. On the other hand, 
%If $\sigma^2>0$, we have $\beta^2>\beta_0^2$, and since we need to select $x_0>m/2$ to guarantee that the Taylor sum~\eqref{eq:taylor} converges, we have $\lim_{p\to0}\E\left[ \var[\sumestS|R ] \right]>0$, showing that the variance of the simple estimator does not converge to 0 in general.
%\end{example}
\subsection{Proof of Proposition~\ref{prop:exp_var_convergence_cycle}}
\label{app:proof_prop_exp_var_convergence_cycle}

\subsubsection{Second moment of cycling variables}
We first show that the assumptions of Proposition~\ref{prop:exp_var_convergence_cycle} imply $\sumestC$ is integrable, and thus $\E[\sumestC|R]$ is well-defined, by appealing to Lemma \ref{lem:integrablefhat}. To use Lemma \ref{lem:integrablefhat}, we need to show Assumption~\ref{ass:condition_urk} holds for the $\UrkC$'s.

Let $\widetilde{X}_i:=(X_i/x_0-1)$, then (for $r\geq k$),
\begin{equation*}
  \E\left[|\UrkC|\right]
  \leq \E\left[\prod_{i=1}^k |\widetilde{X}_i|\right]
  = \left\{\E\left[|\widetilde{X}_i|\right] \right\}^k
  \leq \left\{\E\left[|\widetilde{X}_i|^2\right]^{1/2} \right\}^k
  \leq \beta^k,
\end{equation*}
where $\beta$, defined in \eqref{eq:hyp_simple}, is the square-root of the second moment of $\widetilde{X}_i$. Assumption~\ref{hyp:f_and_x0}-(i) assumes $ |\gamma_k|$ is uniformly bounded by some constant $c$ while condition  \eqref{eq:hyp_simple} assumes  $\beta<1$. Hence $\sum |\gamma_k| a_k$ (of Assumption~\ref{ass:condition_urk}) is bounded above by  $\sum c \beta_k < \infty$. Thus, Lemma \ref{lem:integrablefhat} applies to  $\sumest=\sumestC$.

Before proceeding with the proof of
Proposition~\ref{prop:exp_var_convergence_cycle}, we derive an expression for
the moments of $\UrkC $.
% In this section, we fix $R=r$ and derive the second moment of $\UrkC$.
It is convenient to define $\rho\eqdef 1 + \sigma^2/(m - x_0)^2$ and observe that
$\rho>1$.

The second moment of the cycling variables can be obtained by direct multiplication, distinguishing between three different cases:
\begin{lemma}
\label{lemma:covariance_Z}
For any $r \geq l\geq k$ we have
\begin{enumerate}%[label=(\alph*)]
\item \label{lemma:covariance_Z_a}If $r\geq l+k$
\begin{align*}
\E\left[\UrkC \UrlC\right] = \left(\frac{m}{x_0}-1\right)^{l+k}\frac{1}{r}\left[(r-l-k+1)+ 2\sum_{j=1}^{k-1}\rho^j+(l-k+1)\rho^k\right],
\end{align*}
with the convention that $\sum_{j=n_1}^{n_2} = 0$ if $n_1<n_2$.
\item \label{lemma:covariance_Z_b}If $r< l+k$
\begin{align*}
\E\left[\UrkC \UrlC\right] =\left(\frac{m}{x_0}-1\right)^{l+k}\frac{1}{r}\left[ 2\sum_{j=l+k-r+1}^{k-1}\rho^j+(l-k+1)\rho^k+(l+k-r+1)\rho^{l+k-r}\right],
\end{align*}
with the convention that $\sum_{j=n_1}^{n_2} = 0$ if $n_1<n_2$.
\item \label{lemma:covariance_Z_c} If $r = l$, the cycling variable $\UrlC$ coincides with $\UrlS$ and we have
\begin{align*}
\E\left[\UrkC \UrlC\right] = \left(\frac{\sigma^2}{x_0^2}+\left(\frac{m}{x_0}-1\right)^2\right)^k\left( \frac{m}{x_0} -1 \right)^{l-k} = \left(\frac{m}{x_0}-1\right)^{l+k}\rho^k.
\end{align*}
\end{enumerate}
\end{lemma}

\begin{proof}
We only prove case~\ref{lemma:covariance_Z_a}. The proof for the second and third case follows the same lines.

Recall that $\widetilde{X}_i=(X_i/x_0-1)$, let
$\widetilde{m}:=\Exp\left[X_i/x_0-1\right]$ and
$\widetilde{\sigma}^2:=\var\left(X_i/x_0-1\right) = \sigma^2/x_0^2$. Expanding
$\UrlC$ we find
\begin{align*}
  r \times \UrlC
  = & \prod_{i=1}^l\widetilde{X}_i + \dots + \prod_{i=k}^{l+k-1}\widetilde{X}_i \\
&  +\prod_{i=k+1}^{l+k}\widetilde{X}_i + \prod_{i=k+2}^{l+k+1}\widetilde{X}_i+\dots+\prod_{i=r-l+1}^{r}\widetilde{X}_i\\
& + \left(\prod_{i=r-l+2}^{r}\widetilde{X}_i\right)\widetilde{X}_1+\dots +\widetilde{X}_r\prod_{i=1}^{l-1}\widetilde{X}_i.
\end{align*}
The expansion has been organised in three separate rows for the following reason.  Note that each row is non-empty since $r\geq k+l.$ { (i)} We commence by calculating $\Exp\left[\UrlC\prod_{j=1}^k\widetilde{X}_j \right] $. The first and third row involve non-independent products with the term $\prod_{j=1}^k\widetilde{X}_j $, whereas for the middle row, the expected value is the product of the expectations. {(ii)} The previous calculation is sufficient since, by symmetry, $\Exp\left[ \UrkC\UrlC\right]=\Exp\left[\UrlC\prod_{j=1}^k\widetilde{X}_j \right]$.

\begin{align*}
  & r \times \UrlC \times \prod_{j=1}^k\widetilde{X}_j\\
  =& \prod_{j=1}^k\widetilde{X}_j^2\prod_{i=k+1}^l\widetilde{X}_i + \dots + \widetilde{X}_k^2\prod_{j=1}^{k-1}\widetilde{X}_j\prod_{i=k+1}^{l+k-1}\widetilde{X}_i\\
&+\prod_{j=1}^k\widetilde{X}_j\prod_{i=k+1}^{l+k}\widetilde{X}_i +\dots+\prod_{j=1}^k\widetilde{X}_j\prod_{i=r-l+1}^{r}\widetilde{X}_i \\
&+ \widetilde{X}_1^2\prod_{j=2}^k\widetilde{X}_j\prod_{i=r-l+2}^{r}\widetilde{X}_i+\dots +\prod_{j=1}^k\widetilde{X}_j^2\prod_{i=r-l+k+1}^{r}\widetilde{X}_i+\dots
+\left(\prod_{j=1}^k\widetilde{X}_j^2\right)\widetilde{X}_r \prod_{i=k+1}^{l-1}\widetilde{X}_i,
\end{align*}
and, taking expectations,
\begin{align*}
  \Exp\left[r \UrlC\prod_{j=1}^k\widetilde{X}_j \right]=
  & (\widetilde{\sigma}^2+\widetilde{m}^2)^k\widetilde{m}^{l-k} + \dots + (\widetilde{\sigma}^2+\widetilde{m}^2)\widetilde{m}^{l+k-2}\\
&+\widetilde{m}^{l+k}+\ldots+\widetilde{m}^{l+k}\\
&+ (\widetilde{\sigma}^2+\widetilde{m}^2)\widetilde{m}^{l+k-2}+\dots
+(\widetilde{\sigma}^2+\widetilde{m}^2)^k\widetilde{m}^{l-k}+\dots +(\widetilde{\sigma}^2+\widetilde{m}^2)^k\widetilde{m}^{l-k}.
\end{align*}
Counting how many times each term in the sum above appears we find that
\begin{align*}
\Exp\left[r \UrlC\prod_{j=1}^k\widetilde{X}_j \right] &=\widetilde{m}^{l+k}(r-(l+k)+1) + 2\sum_{j=1}^{k-1}(\widetilde{\sigma}^2+\widetilde{m}^2)^j\widetilde{m}^{l+k-2j}+(\widetilde{\sigma}^2+\widetilde{m}^2)^{k}\widetilde{m}^{l-k}(l-k+1)\\
& =\widetilde{m}^{l+k}\left[(r-l-k+1)+ 2\sum_{j=1}^{k-1}\left(\frac{\widetilde{\sigma}^2}{\widetilde{m}^2}+1\right)^j+(l-k+1)\left(\frac{\widetilde{\sigma}^2}{\widetilde{m}^2}+1\right)^k\right].
\end{align*}
We can then exploit the symmetry of $\UrkC$ to conclude
\begin{align*}
\Exp\left[ \UrkC\UrlC\right] &= \frac{\widetilde{m}^{l+k}}{r}\left[(r-l-k+1)+ 2\sum_{j=1}^{k-1}\left(\frac{\widetilde{\sigma}^2}{\widetilde{m}^2}+1\right)^j+(l-k+1)\left(\frac{\widetilde{\sigma}^2}{\widetilde{m}^2}+1\right)^k\right].
\end{align*}
Using the fact that $\widetilde{m}=(m/x_0-1)$ and $\widetilde{\sigma}^2=\sigma^2/x_0^2$ and recalling the definition of $\rho= 1 + \sigma^2/(m - x_0)^2$ we finally obtain
\begin{align*}
\Exp\left[ \UrkC\UrlC\right] = \left(\frac{m}{x_0}-1\right)^{l+k}\frac{1}{r}\left[(r-l-k+1)+ 2\sum_{j=1}^{k-1}\rho^j+(l-k+1)\rho^k\right].
\end{align*}
\end{proof}

\subsubsection{Bound on second moments of cycling variables}

To characterise the behaviour of the cycling estimator as the number of terms in the Taylor expansion increases, we establish the following bound on $ \cov\left[\UrkC\UrlC\right]$.
\begin{lemma}
\label{lemma:covariance_cycle_bound}
Let $r\geq l\geq k$. Then,
\begin{align*}
\cov\left[\UrkC\UrlC\right] < \left\lvert\frac{m}{x_0}-1\right\rvert^{l+k}
\times \begin{cases}
\rho^k(l+k)/r,& \text{for } r\geq l+k\\
(\rho^k+1)& \text{for } r< l+k\\
(\rho^k-1)& \text{for } r=l.
\end{cases}
\end{align*}
\end{lemma}
\begin{proof}
The proofs relies on bounding each one of the second moments obtained in Lemma~\ref{lemma:covariance_Z} separately.
Recall that $\cov\left[\UrkC\UrlC\right] = \Exp\left[ \UrkC\UrlC\right] - \Exp\left[ \UrkC\right]\Exp\left[ \UrlC\right]$.
Using the expression for $\Exp\left[ \UrkC\UrlC\right]$ in Lemma~
\ref{lemma:covariance_Z}--\ref{lemma:covariance_Z_a}, since $\rho>1$, we have
\begin{align*}
\cov\left[\UrkC\UrlC\right]&=\left(\frac{m}{x_0}-1\right)^{l+k}\frac{1}{r}\left[(r-l-k+1)+ 2\sum_{j=1}^{k-1}\rho^j+(l-k+1)\rho^k-r\right]\\
&<\left\lvert\frac{m}{x_0}-1\right\rvert^{l+k}\frac{1}{r}\left[1+ 2\sum_{j=1}^{k-1}\rho^j+(l-k+1)\rho^k\right]\\
&<\frac{\rho^k}{r}\left\lvert\frac{m}{x_0}-1\right\rvert^{l+k}\left[1+ 2\sum_{j=1}^{k-1}\rho^{j-k}+(l-k+1)\right]\\
&<\frac{\rho^k}{r}\left\lvert\frac{m}{x_0}-1\right\rvert^{l+k}\left[1+ 2(k-1)+(l-k+1)\right]\\
&=  \frac{\rho^k}{r}(l+k)\left\lvert\frac{m}{x_0}-1\right\rvert^{l+k}.
\end{align*}

Similarly, for case~\ref{lemma:covariance_Z_b} we find
\begin{align*}
\Exp\left[ \UrkC\UrlC\right] &=\frac{1}{r}\left(\frac{m}{x_0}-1\right)^{l+k}\rho^k\left[ 2\sum_{j=l+k-r+1}^{k-1}\rho^{j-k}+(l-k+1)+(l+k-r+1)\rho^{l-r}\right]\\
&<\frac{1}{r}\left\lvert\frac{m}{x_0}-1\right\rvert^{l+k}\rho^k\left[ 2(r-l-1)+(l-k+1)+(l+k-r+1)\right]\\
&=\left\lvert\frac{m}{x_0}-1\right\rvert^{l+k}\rho^k,
\end{align*}
where we used the fact that $\rho>1$. 
It follows that
\begin{align*}
\cov\left[\UrkC\UrlC\right] <\left\lvert\frac{m}{x_0}-1\right\rvert^{l+k}\left(\rho^k-\textrm{sgn}\left(\left(\frac{m}{x_0}-1\right)^{l+k}\right)\right) \leq \left\lvert\frac{m}{x_0}-1\right\rvert^{l+k}(\rho^k+1).
\end{align*}
%Observing that $r<l+k\leq 2l$ we find \begin{align*} \cov\left[\UrkC\UrlC\right] <\frac{4l}{r}\left\lvert\frac{m}{x_0}-1\right\rvert^{l+k}\rho^k. \end{align*}
Finally, case~\ref{lemma:covariance_Z_c} follows straightforwardly
\begin{align*}
\cov\left[\UrkC\UrlC\right] =  \left(\frac{m}{x_0}-1\right)^{l+k}(\rho^k - 1).% < 4\left\lvert\frac{m}{x_0}-1\right\rvert^{l+k}\rho^k ,
\end{align*} 
Collecting all the stated, for cases ~\ref{lemma:covariance_Z_a}, ~\ref{lemma:covariance_Z_b} and ~\ref{lemma:covariance_Z_c}, completes the proof.
\end{proof}

\subsubsection{Proof of Proposition~\ref{prop:exp_var_convergence_cycle}}

Before proceeding to the proof of Proposition~\ref{prop:exp_var_convergence_cycle} we state and prove the following simple result:
\begin{lemma}[Expectation of Reciprocal]
\label{lemma:reciprocal_geometric}
If $R\sim\Geom(p)$ then for all $k\geq 1$,
\[\E\left[\frac{1}{R}\middle| R\geq k\right]
  = \frac{1}{\pr(R\geq k)}\sum_{r=k}^{\infty}\frac{1}{r} \pr(R=r)
  \leq \frac{p\log (1/p) }{1-p}.
  \]
\end{lemma}
\begin{proof}
By definition we have
\begin{align*}
\frac{1}{\pr(R\geq k)}\sum_{r=k}^{\infty}\frac{1}{r} \pr(R=r)&= p\sum_{r=k}^{\infty}\frac{(1-p)^{r-k}}{r}\\
&= p\sum_{h=1}^{\infty}\frac{(1-p)^{h-1}}{h+k-1}\\
&\leq \frac{p}{1-p}\sum_{h=1}^{\infty}\frac{(1-p)^{h}}{h}\\
& = \frac{p}{1-p}\log (1/p).
\end{align*}
\end{proof}

We are now ready to prove Proposition~\ref{prop:exp_var_convergence_cycle}.
First, we observe that for $k=0$ we have $U_{r, 0}^\textrm{C} \equiv 1$, thus $ \var(U_{r, 0}^\textrm{C})=0$ and we have
\begin{multline*}
	\E\left[ \var[\sumestC|R ] \right] =
	\sum_{k=1}^\infty \frac{\gamma_k^2}{\pr(R\geq k)^2}
	\left\{ \sum_{r=k}^\infty \pr(R=r) \var(\UrkC)  \right\}\\
	+ 2 \sum_{k=0}^{\infty}\sum_{l=k+1}^{\infty} \frac{\gamma_k
		\gamma_l}{\pr(R\geq k)\pr(R\geq l)}
	\left\{  \sum_{r=l}^{\infty}\pr(R=r)\cov\left( \UrkC, \UrlC  \right)\right\}.
\end{multline*}
We use Lemma \ref{lemma:covariance_cycle_bound} to bound the terms $\cov\left( \UrkC, \UrlC  \right)$. For the case $r<l+k$, we use the simplified bound $\rho^k+1 \leq 2 \rho^k$ and then  $\rho^k< \rho^k(l+k)/r \leq \rho^k(2l)/r$. When $r=l$, clearly $\rho^k-1<  \rho^k(2l)/r $. 

Combining the above bounds from Lemma \ref{lemma:covariance_cycle_bound}  with Lemma~\ref{lemma:reciprocal_geometric}, gives
\begin{align*}
  \E\left[ \var[\sumestC|R ] \right] \leq
  & \sum_{k=1}^\infty \frac{\gamma_k^2}{\pr(R\geq k)^2}
	\left\{ 2k \left\lvert\frac{m}{x_0}-1\right\rvert^{2k}\rho^k\sum_{r=k}^{\infty}\frac{1}{r} \pr(R=r)\right\}\\
	&+ 2 \sum_{k=0}^{\infty}\sum_{l=k+1}^{\infty} \frac{|\gamma_k \gamma_l|}{\pr(R\geq k)\pr(R\geq l)}
	\left\{ 4l \left\lvert\frac{m}{x_0}-1\right\rvert^{l+k}\rho^k\sum_{r=l}^{\infty}\frac{1}{r} \pr(R=r)\right\}\\
	\leq &   \frac{p\log (1/p) }{1-p}
	\sum_{k=1}^\infty \frac{\gamma_k^2}{\pr(R\geq k)}
	2k \left\lvert\frac{m}{x_0}-1\right\rvert^{2k}\rho^k\\
	&+2\frac{p\log (1/p) }{1-p} \sum_{k=0}^{\infty}\sum_{l=k+1}^{\infty}
    \frac{| \gamma_k \gamma_l |}{\pr(R\geq k)}4l \left\lvert\frac{m}{x_0}-1\right\rvert^{l+k}\rho^k.
\end{align*}

Using Assumption~\ref{hyp:f_and_x0} and recalling that
\begin{align*}
\beta^{2k} &= \left(\frac{\sigma^2}{x_0^2}+\left(\frac{m}{x_0}-1\right)^2\right)^k = \left(\frac{m}{x_0}-1\right)^{2k}\rho^k,
\end{align*}
that $\beta_0 = \vert m/x_0-1\vert<1$,
we have
\begin{align*}
  \E\left[ \var[\sumestC|R ] \right]
  &\leq     2c^2\frac{p\log(1/ p) }{1-p}
    \left\{ \sum_{k=1}^\infty k\frac{\beta^{2k}}{(1-p)^k}
    + 4\sum_{k=0}^{\infty} \frac{\beta_0^{k}\rho^k}{(1-p)^k}\sum_{l=k+1}^{\infty} l \beta_0^{l}\right\}\\
	& = \frac{2c^2p\log (1/p)}{1-p} \left\{   \frac{\beta^2(1-p)}{(1-p-\beta^2)^2}
    +4 \frac{\beta_0}{(1-\beta_0)^2}\sum_{k=0}^{\infty}\frac{\beta_0^{2k}\rho^k }{(1-p)^k}(k(1-\beta_0)+1) \right\}\\
	& \leq \frac{2c^2p\log (1/p)}{1-p} \left\{   \frac{\beta^2(1-p)}{(1-p-\beta^2)^2}
    +4 \frac{\beta_0}{(1-\beta_0)^2}\sum_{k=0}^{\infty}\frac{\beta^{2k} }{(1-p)^k}(k+1) \right\}\\
	& = 2c^2\frac{p\log (1/p)}{(1-p-\beta^2)^2} \left(\beta^2+4\frac{\beta_0(1-p)}{(1-\beta_0)^2}\right),
\end{align*}
where we used that $\sum_{k=1}^{\infty} k g^k = g/(1-g)^2$ and $\sum_{k=0}^{\infty} (k+1) g^k = 1/(1-g)^2$ for $g\in(0, 1)$.

%%% Local Variables: 
%%% mode: latex
%%% TeX-master: "paper.tex"
%%% End: 

%\input{appendix_gradient}
%\subsection{Unbiasedness of simple and cycling estimators} \label{app:gradient_unbiased}
%By definition we have that $\Wrk$ is independent of $R$ and 
%\begin{align*}
%	\E[\Wrk]=\left(\frac{m(\theta)}{x_0}-1\right)^{k-1}\nabla_\theta m(\theta).
%\end{align*} Then, using the above
%\begin{align*}
%	\E\left[ \nablasumest \middle| R=r\right] & =
%	\sum_{k=1}^\infty \frac{\ind{r \geq k}}{\pr(R\geq k)} \E[\Wrk | R=r]\frac{k\gamma_k}{x_0}
%	\\
%	                                          & =\sum_{k=1}^\infty \frac{\ind{r \geq k}}{\pr(R\geq k)} \left(
%	\frac{m(\theta)}{x_0} -1 \right)^{k-1}\nabla_\theta m(\theta)\frac{k\gamma_k}{x_0},
%\end{align*} from which follows that $\E[\nablasumest] = \sum_{k}  k\gamma_k(m(\theta)/x_0 -1)^{k-1}\nabla_\theta m(\theta)/x_0$.

\subsection{Proof of Proposition~\ref{prop:var_gradients}}
\label{app:var_gradients}
The proof of the first assertion, that $\nablasumestS$ and $\nablasumestC$ have finite second moments, follows closely the steps of Proposition \ref{prop:exp_var_convergence_simple}'s proof (that $\sumestS$ has a finite second moment.) This part is thus omitted.

For the remaining three statements on the variance, the following
variance decomposition of $\nablasumest$ will be used,
\begin{equation*}%\label{eq:decomp_var_gradient}
	\var[\nablasumest] = \var\left[ \E[\nablasumest|R ] \right]
	+ \E\left[ \var[\nablasumest | R ] \right],
\end{equation*}
where the first term can be expressed as
\begin{multline}\label{eq:varE_gradient_2}
	\var\left[ \E[\nablasumest|R ] \right] =
	\sum_{k=1}^\infty k^2 \gamma_k^2  \left(  \frac{m}{x_0} - 1 \right)^{2k-2}\frac{\nabla_\theta m(\theta)^2}{x_0^2}
	\left(\frac{1}{\pr(R\geq k)} - 1 \right) \\
	+ 2 \sum_{k=1}^\infty\sum_{l=k+1}^\infty kl\gamma_k \gamma_l \left(\frac{m}{x_0} - 1\right)^{k+l-2}\frac{\nabla_\theta m(\theta)^2}{x_0^2}
	\left(\frac{1}{\pr(R\geq k)} - 1 \right).
\end{multline}
As for the second term, it is
\begin{multline}\label{eq:Evar_gradient_2}
	\E\left[ \var[\nablasumest|R ] \right] =
	\frac{1}{x_0^2}\sum_{k=1}^\infty k^2\frac{\gamma_k^2}{\pr(R\geq k)^2}
	\left\{ \sum_{r=k}^\infty \pr(R=r) \var(\Wrk)  \right\}\\
	+ \frac{2}{x_0^2} \sum_{k=1}^{\infty}\sum_{l=k+1}^{\infty} kl\frac{\gamma_k
		\gamma_l}{\pr(R\geq k)\pr(R\geq l)}
	\left\{  \sum_{r=l}^{\infty}\pr(R=r)\cov\left( W_{r,k}, W_{r,l}  \right)\right\}
\end{multline}
where the expression of $\Wrk$ depends on the chosen estimator, either simple or cycling, for $(m(\theta)/x_0 -1)^{k-1}\nabla_\theta m(\theta)$.

Statements {\it (1)}--{\it(3)} of Proposition~\ref{prop:var_gradients} are established as follows:
\begin{itemize}
\item [(1)] The stated result for $\var\left[ \E[\nablasumest|R ] \right]$  is established by Proposition \ref{prop:var_exp_gradient}. Note that $ \E[\nablasumest|R ]$ has the same expression for both the simple 
and cycling case.

\item [(2)] Example \ref{ex:Evar_grad_logx} shows $\lim_{p\rightarrow0} \E\left[ \var[\nablasumestS|R ] \right]  >0$ for $f(\cdot)=\log(\cdot)$. 
%Section \ref{app:simple_gradient_est} collects 
The necessary supporting results to enable this calculation are presented just before this example.
\item [(3)] Proposition \ref{prop:exp_var_gradient_cycle} verifies the final assertion for the cycle estimate. 
\end{itemize}
\begin{proposition}
	\label{prop:var_exp_gradient}
	Let Assumption~\ref{hyp:f_and_x0} hold with $m$ therein replaced by $m(\theta)$. If there exists $p\in(0, 1-\beta_0^2)$ such that
	$\pr(R\geq k) \geq (1-p)^{k-1}$, then a bound for \eqref{eq:varE_gradient_2} is 
	\begin{multline*}
		\var\left[ \E[\nablasumest|R ] \right] \leq c^2 \frac{\nabla_\theta m(\theta)^2}{x_0^2}\frac{1}{(1-\beta_0)^2}\\
		\times \left(\frac{(\beta_0^2+\beta_0^4)(1-\beta_0)^2+4\beta_0}{(1-\beta_0^2/(1-p))^3}-\frac{(\beta_0^2/(1-p)+\beta_0^4/(1-p)^2)(1-\beta_0)^2+4\beta_0}{(1-\beta_0^2)^3}\right).
	\end{multline*}
	
\end{proposition}
It is can be shown that the right-hand side is $\OO(p)$ as $p\rightarrow 0$. The proof of this result is very similar  to that of Proposition~\ref{prop:var_exp} and is therefore omitted.

\subsubsection{Simple gradient estimator}\label{app:simple_gradient_est}

Using the definition of $\WrkS$ in~\eqref{eq:W_simple} it is easy to see that for any $0 \leq k\leq l$,
\begin{align}
	  \var\left(\WrkS\right)
	  &=\left(\frac{\sigma^2(\theta)}{x_0^2}+\left(\frac{m(\theta)}{x_0}-1\right)^2\right)^{k-1}s^2(\theta)
	-\left(\frac{m(\theta)}{x_0}-1\right)^{2k-2}\nabla_\theta m(\theta)^2, \label{eq:varW_simple}\\
	  \cov\left(\WrkS, \WrlS\right)  
	 &=\left(\frac{\sigma^2(\theta)}{x_0^2}+\left(\frac{m(\theta)}{x_0}-1\right)^2\right)^{k-1}
	\left(\frac{m(\theta)}{x_0} - 1\right)^{l-k-1}\nabla_\theta m(\theta)t(\theta)\label{eq:covW_simple}\\
	&\qquad -\left(\frac{m(\theta)}{x_0}-1\right)^{k+l-2}\nabla_\theta m(\theta)^2, \nonumber
\end{align}
where we denoted $s^2(\theta):=\Exp\left[G_k^2\right]$ and $t(\theta):=\Exp\left[(X_k/x_0-1)G_k\right]$.

%Proposition~\ref{prop:var_exp_gradient} guarantees that $\var\left[ \E[\nablasumestS|R ]\right]\to 0$ as $p\to 0$.
To obtain the limit $\lim_{p\rightarrow 0}\E\left[\var(\nablasumestS | R)\right]$, substitute \eqref{eq:varW_simple}  and \eqref{eq:covW_simple} into \eqref{eq:Evar_gradient_2} and use Tannery's Theorem; \citealp{boas1965tannery} to interchange the limit and the sum (as previously shown for the estimator of $f(m)$):
\begin{align*}
	x_0^2\E\left[ \var[\nablasumestS|R ] \right] & \to \sum_{k=1}^\infty k^2\gamma_k^2
	\left\{ \left(\frac{\sigma^2(\theta)}{x_0^2}+\left(\frac{m(\theta)}{x_0}-1\right)^2\right)^{k-1}s^2(\theta)
	-\left(\frac{m(\theta)}{x_0}-1\right)^{2k-2}\nabla_\theta m(\theta)^2\right\}                          \\
	                                             & + 2 \sum_{k=1}^{\infty}\sum_{l=k+1}^{\infty} kl\gamma_k
	\gamma_l\left(\frac{m(\theta)}{x_0} - 1\right)^{l-k-1}\nabla_\theta m(\theta)                          \\
	                                             & \times
	\left\{  \left(\frac{\sigma^2(\theta)}{x_0^2}+\left(\frac{m(\theta)}{x_0}-1\right)^2\right)^{k-1}
	t(\theta)
	-\left(\frac{m(\theta)}{x_0}-1\right)^{2k-1}\nabla_\theta m(\theta)\right\}
\end{align*}

The following example shows $\lim_{p\to0}\E\left[ \var[\nablasumestS|R ] \right]$ can be non-zero.
\begin{example}\label{ex:Evar_grad_logx}
	Consider $f(x) =\log x$ so that $\gamma_k = (-1)^{k-1}/k$ and assume that we want to estimate $f(m)$ for $m>0$. Then, we have
	\begin{align*}
		x_0^2\lim_{p\to0}\E\left[ \var[\nablasumestS|R ] \right] & = \sum_{k=1}^\infty
		\left\{ \beta^{2(k-1)}s^2(\theta)
		-\beta_0^{2(k-1)}\nabla_\theta m(\theta)^2\right\}                                                                                                 \\
		                                                         & + 2\nabla_\theta m(\theta) \sum_{k=1}^{\infty}\left\{  \beta^{2(k-1)}
		t(\theta)
		-\beta_0^{2(k-1)}\beta_0\nabla_\theta m(\theta)\right\}                                                                                            \\
		                                                         & \times\sum_{l=k+1}^{\infty} (-1)^{k+l-2}\left(\frac{m(\theta)}{x_0} - 1\right)^{l-k-1}  \\
		                                                         & = \frac{s^2(\theta)}{1-\beta^2}-\frac{\nabla_\theta m(\theta)^2}{1-\beta_0^2}           \\
		                                                         & - \frac{2\nabla_\theta m(\theta)}{1+\beta_0} \sum_{k=1}^{\infty}\left\{  \beta^{2(k-1)}
		t(\theta)
		-\beta_0^{2(k-1)}\beta_0\nabla_\theta m(\theta)\right\}                                                                                            \\
		                                                         & = \frac{s^2(\theta)}{1-\beta^2}-\frac{\nabla_\theta m(\theta)^2}{1-\beta_0^2}           \\
		                                                         & - \frac{2\nabla_\theta m(\theta)}{1+\beta_0} \left( \frac{
				t(\theta)}{1-\beta^2}
		-\frac{\beta_0\nabla_\theta m(\theta)}{1-\beta_0^2}\right).
	\end{align*}
	Since $s^2(\theta)=\Exp[G_k^2] \geq\Exp[G_k]^2=\nabla_\theta m(\theta)^2$ we can lower bound the first term with
	\begin{align*}
		\frac{s^2(\theta)}{1-\beta^2}-\frac{\nabla_\theta m(\theta)^2}{1-\beta_0^2} \geq\nabla_\theta m(\theta)^2 \left( \frac{1}{1-\beta^2}
		-\frac{1}{1-\beta_0^2}\right).
	\end{align*}
	Then, assume that we can sample $G_k$ and $X_k$ independently, so that $t(\theta) = \beta_0\nabla_\theta m(\theta)$. It follows that the second term becomes
	\begin{align*}
		- \frac{2\nabla_\theta m(\theta)^2\beta_0}{1+\beta_0} \left( \frac{1}{1-\beta^2}
		-\frac{1}{1-\beta_0^2}\right) =  2\nabla_\theta m(\theta)^2\left(\frac{x_0}{m}-1\right) \left( \frac{1}{1-\beta^2}
		-\frac{1}{1-\beta_0^2}\right).
	\end{align*}
	Hence, we can lower bound
	\begin{align*}
		x_0^2\lim_{p\to0}\E\left[ \var[\nablasumestS|R ] \right] \geq \nabla_\theta m(\theta)^2 \left( \frac{1}{1-\beta^2}
		-\frac{1}{1-\beta_0^2}\right)\left(\frac{2x_0}{m}-1\right).
	\end{align*}
	If $\sigma^2>0$, we have $\beta^2>\beta_0^2$ and, since we need to select $x_0>m/2$ to guarantee that the Taylor sum~\eqref{eq:taylor_gradient_theta} converges, we have $\lim_{p\to0}\E\left[ \var[\nablasumestS|R ] \right]>0$, showing that the variance of the simple estimator does not converge to 0 in general.
\end{example}

\subsubsection{Cycling gradient estimator}
For the cycling estimator, under Assumption~\ref{hyp:f_and_x0} (with $m$ replaced by $m(\theta)$ therein), $\E\left[\var(\nablasumestC | R)\right]$ satisfies a result similar to that in Proposition~\ref{prop:exp_var_convergence_cycle}. The proof of this result can established by following the steps outlined to prove  Proposition~\ref{prop:exp_var_convergence_cycle} (that gave the corresponding bound for $\E\left[\var(\sumestC | R)\right]$.) In particular, substitute into \eqref{eq:Evar_gradient_2} the bounds on the second moments of $\WrkC$ given in Lemma \ref{lemma:covariance_cycle_bound_gradient}, and then also use Lemma~\ref{lemma:reciprocal_geometric}. 
\begin{proposition}
\label{prop:exp_var_gradient_cycle}
	Let Assumption~\ref{hyp:f_and_x0} hold with  $m$ therein replaced by $m(\theta)$.  Assume further that
	\begin{equation*}
			\beta^2 \eqdef \frac{\sigma^2(\theta)}{x_0^2} + \left( \frac{m(\theta)}{x_0} -1 \right)^2 < 1,
	\end{equation*}
	and $R\sim\Geomp$ with $p\in(0, 1-\beta^2)$. Then,
	\begin{align*}
		&\E\left[\var(\nablasumestC | R)\right] \leq  
		 \frac{8c^2s(\theta)^{2}p\log (1/p)}{x_0^2(1-p)}\\
		 &\qquad\times\left\{  \frac{\beta^2(1-p)(\beta^2+2(1-p))}{\beta_0(1-p-\beta^2)^4}
    +\frac{\beta_0^2+2}{(1-\beta_0)^3} \frac{\beta^2(\beta^4+4\beta^2(1-p)+(1-p)^2)}{(1-p-\beta^2)^4}\right\},
	\end{align*}
\end{proposition}
\begin{proof} 
Recall that from the law of total variance decomposition we have
\begin{multline}
	\E\left[ \var[\nablasumestC|R ] \right] =
	\frac{1}{x_0^2}\sum_{k=1}^\infty k^2\frac{\gamma_k^2}{\pr(R\geq k)^2}
	\left\{ \sum_{r=k}^\infty \pr(R=r) \var(\WrkC)  \right\}\\
	+ \frac{2}{x_0^2} \sum_{k=1}^{\infty}\sum_{l=k+1}^{\infty} kl\frac{\gamma_k
		\gamma_l}{\pr(R\geq k)\pr(R\geq l)}
	\left\{  \sum_{r=l}^{\infty}\pr(R=r)\cov\left( W_{r,k}^{\textrm{C}}, W_{r,l}^{\textrm{C}}  \right)\right\}.
\end{multline}
We use Lemma \ref{lemma:covariance_cycle_bound_gradient} to bound the terms $\cov\left( \WrkC, \WrlC  \right)$, and Lemma~\ref{lemma:reciprocal_geometric} to bound the expectation w.r.t. $R$. Since $\beta_0^2<\beta^2<1$, $\max(\beta_{0},1)=1$. 
\begin{align*}
 & \E\left[ \var[\nablasumestC|R ] \right] \\
  & \leq \frac{2s(\theta)^{2}}{x_0^2}\sum_{k=1}^\infty (2k-2)k^2\frac{\gamma_k^2}{\pr(R\geq k)^2}
	\rho(\theta)^{k-1}\beta_{0}^{2k-3}\left\{ \sum_{r=k}^\infty \pr(R=r) \frac{1}{r}  \right\}\\
	& + \frac{4s(\theta)^{2}}{x_0^2} \sum_{k=1}^{\infty}\sum_{l=k+1}^{\infty} (2l-2)kl\frac{\gamma_k
		\gamma_l}{\pr(R\geq k)\pr(R\geq l)}
	\rho(\theta)^{k-1}\beta_{0}^{k+l-3}\left\{  \sum_{r=l}^{\infty}\pr(R=r)\frac{1}{r} \right\}\\
	 & \leq \frac{p\log (1/p) }{1-p}\frac{2s(\theta)^{2}}{x_0^2}\sum_{k=1}^\infty (2k-2)k^2\frac{\gamma_k^2}{\pr(R\geq k)}
	\rho(\theta)^{k-1}\beta_{0}^{2k-3}\\
	& + \frac{p\log (1/p) }{1-p}\frac{4s(\theta)^{2}}{x_0^2} \sum_{k=1}^{\infty}\sum_{l=k+1}^{\infty} (2l-2)kl\frac{\gamma_k
		\gamma_l}{\pr(R\geq k)}
	\rho(\theta)^{k-1}\beta_{0}^{k+l-3}.
\end{align*}

Using Assumption~\ref{hyp:f_and_x0} and recalling that
\begin{align*}
\beta^{2k} &= \left(\frac{\sigma^2(\theta)}{x_0^2}+\left(\frac{m(\theta)}{x_0}-1\right)^2\right)^k = \beta_0^{2k}\rho(\theta)^k,
\end{align*}
we have
\begin{align*}
& \E\left[ \var[\nablasumestC|R ] \right]\\
  &\leq  \frac{p\log (1/p) }{1-p}\frac{2c^2s(\theta)^{2}}{x_0^2}\sum_{k=1}^\infty \frac{(2k-2)k^2}{(1-p)^k}
	\rho(\theta)^{k-1}\beta_{0}^{2k-3}\\
	&+ \frac{p\log (1/p) }{1-p}\frac{4c^2s(\theta)^{2}}{x_0^2} \sum_{k=1}^{\infty}\sum_{l=k+1}^{\infty} \frac{(2l-2)kl}{(1-p)^k}
	\rho(\theta)^{k-1}\beta_{0}^{k+l-3}   \\
	&=\frac{p\log (1/p) }{1-p}\frac{2c^2s(\theta)^{2}}{\beta_0 x_0^2}\sum_{k=1}^\infty (2k-2)k^2\frac{\beta^{2(k-1)}}{(1-p)^k}\\
	&+ \frac{p\log (1/p) }{1-p}\frac{4c^2s(\theta)^{2}}{x_0^2} \sum_{k=1}^{\infty}\frac{k}{(1-p)^k}\beta_{0}^{k-3}\rho(\theta)^{k-1}\sum_{l=k+1}^{\infty} (2l-2)l
	\beta_{0}^{l}   \\
	& = \frac{2c^2s(\theta)^{2}p\log (1/p)}{\beta_0 x_0^2(1-p)}   \frac{4\beta^2(1-p)(\beta^2+2(1-p))}{(1-p-\beta^2)^4}\\
    &+ \frac{4c^2s(\theta)^{2}p\log (1/p) }{x_0^2(1-p)}\sum_{k=1}^{\infty}\frac{k}{(1-p)^k}\beta_{0}^{2k-2}\rho(\theta)^{k-1}\frac{2(\beta_0^2k^2-2k^2\beta_0+k^2-k\beta_0^2+k+2\beta_0)}{(1-\beta_0)^3} \\
    &\leq \frac{2c^2s(\theta)^{2}p\log (1/p)}{\beta_0 x_0^2(1-p)}   \frac{4\beta^2(1-p)(\beta^2+2(1-p))}{(1-p-\beta^2)^4}\\
    &+ \frac{8c^2s(\theta)^{2}p\log (1/p) }{x_0^2(1-p)}\frac{\beta_0^2+2}{(1-\beta_0)^3}\sum_{k=1}^{\infty}k^3\frac{\beta^{2(k-1)}}{(1-p)^k} \\
	& \leq \frac{8c^2s(\theta)^{2}p\log (1/p)}{x_0^2(1-p)}\\
	&\qquad\times\left\{  \frac{\beta^2(1-p)(\beta^2+2(1-p))}{\beta_0(1-p-\beta^2)^4}
    +\frac{\beta_0^2+2}{(1-\beta_0)^3} \frac{\beta^2(\beta^4+4\beta^2(1-p)+(1-p)^2)}{(1-p-\beta^2)^4}\right\},
\end{align*}
where we used that $\sum_{k=0}^{\infty} 2k(k+1)^2 g^k = 4g(g+2)/(1-g)^4$ and $\sum_{k=0}^{\infty}k^3 g^k = g(g^2+4g+1)/(1-g)^4$.
\end{proof}

\subsubsection{Second moments of cycling variables}

The second moments of the cycling variables $\WrkC$ can be obtained by a simple counting argument identical to that used in the proof of Lemma~\ref{lemma:covariance_Z} as shown in the Proposition below.
\begin{lemma}
\label{lem:Exp_prod_of_cycle_gradient} Let $s(\theta)^{2}=\mathbb{E}[G_{1}^{2}]$,
$t(\theta)=\mathbb{E}\left((X_{1}/x_{0}-1)G_{1}\right)$ and $$\rho(\theta)=\frac{\mathbb{E}\left((X_{1}/x_{0}-1)^{2}\right)}{\left(\mathbb{E}(X_{1}/x_{0}-1)\right)^{2}}.$$
For any $r\geq l\geq k$,
\begin{alignat}{1}
\mathbb{E}\left(W_{r,k}^{\mathrm{C}}W_{r,l}^{\mathrm{C}}\right)\nonumber \\
 & =\frac{1}{r}\sum_{i=1}^{1}\rho(\theta)^{k-i}\left(\frac{m(\theta)}{x_{0}}-1\right)^{k+l-3}\left(\ind{l>k}t(\theta)\nabla m(\theta)+\ind{l=k}s(\theta)^{2}\left(\frac{m(\theta)}{x_{0}}-1\right)\right)\label{eq:cycgrad_term1a}\\
 & +\frac{1}{r}\sum_{i=2,i<k}^{r-l+1}\rho(\theta)^{k-i}\left(\frac{m(\theta)}{x_{0}}-1\right)^{k+l-3}t(\theta)\nabla m(\theta)\label{eq:cycgrad_term1b}\\
 & +\frac{1}{r}\sum_{i=1,i\geq k}^{r-l+1}\left(\frac{m(\theta)}{x_{0}}-1\right)^{k+l-2}\left(\nabla m(\theta)\right)^{2}\label{eq:cycgrad_term1c}\\
 & +\frac{1}{r}\sum_{i=1}^{k-1}\rho(\theta)^{i-1}\left(\frac{m(\theta)}{x_{0}}-1\right)^{k+l-3}t(\theta)\nabla m(\theta)\label{eq:cycgrad_term2a}\\
 & +\frac{1}{r}\sum_{i=k}^{k}\ind{l>k}\rho(\theta)^{k-1}s(\theta)^{2}\left(\frac{m(\theta)}{x_{0}}-1\right)^{l+k-2}\label{eq:cycgrad_term2b}\\
 & +\frac{1}{r}\sum_{i=k+1}^{l-1}\rho(\theta)^{k-1}\left(\frac{m(\theta)}{x_{0}}-1\right)^{l+k-3}t(\theta)\nabla m(\theta),\label{eq:cycgrad_term2c}
\end{alignat}
where empty sums, $\sum_{i=i_{1}}^{i_{2}}$ with $i_{2}<i_{1}$, are
by convention zero. 
\end{lemma}
\begin{proof}
Equivalently, we can also write $\rho(\theta)=1+\sigma(\theta)^{2}/(m(\theta)-x_{0})^{2}$.
The sums with single terms, namely (\ref{eq:cycgrad_term1a}) and
(\ref{eq:cycgrad_term2b}), single the unique cases that involve $s(\theta)^{2}$,
as explained more below.
For integers $i,j$, with $i<j$, let $i:j$ denote the sequence $i,i+1,\ldots,j$.
For $r\geq l$, let
\begin{equation}
W_{r,l}^{\mathrm{C}}=\frac{1}{r}\sum_{i=1}^{r-l+1}V(i:i+l-1)+\frac{1}{r}\sum_{i=1}^{l-1}\bar{V}(r-(l-i)+1:r)V(1:i)\label{eq:WC_via_index}
\end{equation}
where
\[
V(i:j)=\prod_{k=i}^{j-1}\left(\frac{X_{k}}{x_{0}}-1\right)\times G_{j},\qquad j>i
\]
and for $i=j$,
\[
V(j:j)=G_{j}.
\]
Finally, 
\[
\bar{V}(i:j)=\prod_{k=i}^{j}\left(\frac{X_{k}}{x_{0}}-1\right),\qquad j\geq i.
\]

For $r\geq l\geq k$, by symmetry of the circular shift that defines
$W_{r,l}^{\mathrm{C}}$, we have
\begin{equation}
\mathbb{E}\left(W_{r,k}^{\mathrm{C}}W_{r,l}^{\mathrm{C}}\right)=\mathbb{E}\left(V(1:k)\times W_{r,l}^{\mathrm{C}}\right),\label{eq:symmetryW}
\end{equation}
where, using (\ref{eq:WC_via_index}), we have
\begin{equation}
\frac{1}{r}\sum_{i=1}^{r-l+1}\mathbb{E}\left(V(1:k)V(i:i+l-1)\right)+\frac{1}{r}\sum_{i=1}^{l-1}\mathbb{E}\left(V(1:k)\bar{V}(r-(l-i)+1:r)V(1:i)\right),\label{eq:symmetryW_2}
\end{equation}
Evaluating the products in (\ref{eq:symmetryW_2}) due to the first
sum: For $i>k$
\begin{alignat}{1}
\mathbb{E}\left(V(1:k)V(i:i+l-1)\right) & =\mathbb{E}\left(V(1:k)\right)\mathbb{E}\left(V(i:i+l-1)\right)\nonumber \\
 & =\left(\frac{m}{x_{0}}-1\right)^{k-1}\nabla m\times\left(\frac{m}{x_{0}}-1\right)^{l-1}\nabla m\nonumber \\
 & =\left(\frac{m}{x_{0}}-1\right)^{k+l-2}\left(\nabla m\right)^{2},\label{eq:term_1c}
\end{alignat}
where the $\theta$ dependence has been omitted for brevity. This
accounts for (\ref{eq:cycgrad_term1c}). For $1<i\leq k$,
\begin{alignat*}{1}
 & \mathbb{E}\left(V(1:k)V(i:i+l-1)\right)\\
 & =\mathbb{E}\left(V(1:k)\bar{V}(i:k)V((k+1):(i+l-1)\right)\\
 & =\mathbb{E}\left(V(1:k)\left(\frac{X_{i}}{x_{0}}-1\right)\cdots\left(\frac{X_{k}}{x_{0}}-1\right)\right)\mathbb{E}\left(V((k+1):(i+l-1)\right)\\
 & =\mathbb{E}\left(V(1:k)\left(\frac{X_{i}}{x_{0}}-1\right)\cdots\left(\frac{X_{k}}{x_{0}}-1\right)\right)\left(\frac{m}{x_{0}}-1\right)^{i+l-k-2}\nabla m
 \end{alignat*}
 and
\begin{alignat*}{1} 
 & \mathbb{E}\left(V(1:k)\left(\frac{X_{i}}{x_{0}}-1\right)\cdots\left(\frac{X_{k}}{x_{0}}-1\right)\right)\\
 & =\mathbb{E}\left(\left(\frac{X_{1}}{x_{0}}-1\right)\cdots\left(\frac{X_{i-1}}{x_{0}}-1\right)\left(\frac{X_{i}}{x_{0}}-1\right)^{2}\cdots\left(\frac{X_{k-1}}{x_{0}}-1\right)^{2}G_{k}\left(\frac{X_{k}}{x_{0}}-1\right)\right)\\
 & =\left(\frac{m}{x_{0}}-1\right)^{i-1}\left(\frac{\sigma^{2}}{x_{0}^{2}}+\left(\frac{m}{x_{0}}-1\right)^{2}\right)^{k-i}t.
\end{alignat*}
Thus 
\begin{alignat}{1}
 & \mathbb{E}\left(V(1:k)V(i:i+l-1)\right),\qquad1<i\leq k,\nonumber \\
 & =\left(\frac{\sigma^{2}}{x_{0}^{2}}+\left(\frac{m}{x_{0}}-1\right)^{2}\right)^{k-i}\left(\frac{m}{x_{0}}-1\right)^{l-k+2i-3}\times t\times\nabla m\nonumber \\
 & =\left(\frac{\sigma^{2}}{x_{0}^{2}}+\left(\frac{m}{x_{0}}-1\right)^{2}\right)^{k-i}\left(\frac{m}{x_{0}}-1\right)^{-2(k-i)}\left(\frac{m}{x_{0}}-1\right)^{l+k-3}t\nabla m\nonumber \\
 & =\rho^{k-i}\left(\frac{m}{x_{0}}-1\right)^{k+l-3}t\nabla m,\label{eq:term_1b}
\end{alignat}
which accounts for (\ref{eq:cycgrad_term1b}).

For $i=1$, if $k=l,$
\begin{alignat}{1}
\mathbb{E}\left(V(1:k)V(1:l)\right) & =\mathbb{E}\left(V(1:k)^{2}\right)\nonumber \\
 & =\left(\frac{\sigma^{2}}{x_{0}^{2}}+\left(\frac{m}{x_{0}}-1\right)^{2}\right)^{k-1}s^{2}\nonumber \\
 & =\rho^{k-1}\left(\frac{m}{x_{0}}-1\right)^{2(k-1)}s^{2},\label{eq:term_1a-1}
\end{alignat}
which accounts for (\ref{eq:cycgrad_term1a}) when $\ind{l=k}=1$.
If $l>k$,
\begin{alignat}{1}
\mathbb{E}\left(V(1:k)V(1:l)\right) & =\mathbb{E}\left(V(1:k)\bar{V}(1:k)V((k+1):l)\right).\nonumber \\
 & =\rho^{k-i}\left(\frac{m}{x_{0}}-1\right)^{k+l-3}t\nabla m.\label{eq:term_1a-2}
\end{alignat}
Note this expected value can be studied as in the case $1<i\leq k$
above to yield the same answer as in (\ref{eq:term_1b}). Thus accounting
for (\ref{eq:cycgrad_term1a}) when $\ind{l>k}=1$.

Evaluating the products in (\ref{eq:symmetryW_2}) due to the second
sum of $W_{r,l}^{\mathrm{C}}$ in (\ref{eq:WC_via_index}): For $1\leq i<k$,
\[
V(1:k)\bar{V}(r-(l-i)+1:r)V(1:i)=\bar{V}(1:i)V(i+1:k)\bar{V}(r-(l-i)+1:r)V(1:i)
\]
and thus,
\begin{alignat}{1}
 & \mathbb{E}\left(\bar{V}(1:i)V(i+1:k)\bar{V}(r-(l-i)+1:r)V(1:i)\right)\nonumber \\
 & =\mathbb{E}\left(\bar{V}(1:i)V(1:i)\right)\mathbb{E}\left(V(i+1:k)\right)\mathbb{E}\left(\bar{V}(r-(l-i)+1:r)\right)\nonumber \\
 & =\left(\frac{\sigma^{2}}{x_{0}^{2}}+\left(\frac{m}{x_{0}}-1\right)^{2}\right)^{i-1}t\left(\frac{m}{x_{0}}-1\right)^{k-i-1}\nabla m\left(\frac{m}{x_{0}}-1\right)^{l-i}\nonumber \\
 & =\left(\frac{\sigma^{2}}{x_{0}^{2}}+\left(\frac{m}{x_{0}}-1\right)^{2}\right)^{i-1}\left(\frac{m}{x_{0}}-1\right)^{k+l-2i-1}t\nabla m\nonumber \\
 & =\rho^{i-1}\left(\frac{m}{x_{0}}-1\right)^{k+l-3}t\nabla m.\label{eq:term_2a}
\end{alignat}
This accounts for (\ref{eq:cycgrad_term2a}). 

For $i=k$,
\begin{alignat}{1}
 & \mathbb{E}\left(V(1:k)\bar{V}(r-(l-k)+1:r)V(1:k)\right)\nonumber \\
 & =\mathbb{E}\left(V(1:k)^{2}\right)\mathbb{E}\left(\bar{V}(r-(l-k)+1:r)\right)\nonumber \\
 & =\left(\frac{\sigma^{2}}{x_{0}^{2}}+\left(\frac{m}{x_{0}}-1\right)^{2}\right)^{k-1}s^{2}\left(\frac{m}{x_{0}}-1\right)^{l-k}\nonumber \\
 & =\rho^{k-1}s^{2}\left(\frac{m}{x_{0}}-1\right)^{l+k-2}.\label{eq:term_2b}
\end{alignat}
This accounts for (\ref{eq:cycgrad_term2b}). 

For $i>k$,
\begin{alignat}{1}
 & \mathbb{E}\left(V(1:k)\bar{V}(r-(l-i)+1:r)V(1:i)\right)\nonumber \\
 & =\mathbb{E}\left(V(1:k)\bar{V}(1:k)\right)\mathbb{E}\left(V(k+1:i)\right)\mathbb{E}\left(\bar{V}(r-(l-i)+1:r)\right)\nonumber \\
 & =\left(\frac{\sigma^{2}}{x_{0}^{2}}+\left(\frac{m}{x_{0}}-1\right)^{2}\right)^{k-1}t\left(\frac{m}{x_{0}}-1\right)^{i-k-1}\nabla m\left(\frac{m}{x_{0}}-1\right)^{l-i}\nonumber \\
 & =\rho^{k-1}\left(\frac{m}{x_{0}}-1\right)^{l+k-3}t\nabla m.\label{eq:term_2c}
\end{alignat}
This accounts for (\ref{eq:cycgrad_term2c}). 

The constituent terms of all the separate sums of $\mathbb{E}\left(W_{r,k}^{\mathrm{C}}W_{r,l}^{\mathrm{C}}\right)$
stated in the Lemma are given by (\ref{eq:term_1a-1}), (\ref{eq:term_1a-2}),
(\ref{eq:term_1b}), (\ref{eq:term_1c}), (\ref{eq:term_2a}), (\ref{eq:term_2b})
and (\ref{eq:term_2c}).
\end{proof}
The second moments can be bounded following a similar argument to that in Lemma~\ref{lemma:covariance_cycle_bound}. We obtain a general bound for the cycling estimate which does not require any assumption besides Assumption~\ref{ass:condition_urk} and finiteness of the second moments; however, under Assumption~\ref{hyp:f_and_x0}, $\beta_0<1$ and the bound below could be simplified.

\begin{lemma}
\label{lemma:covariance_cycle_bound_gradient} For any $r\geq l\geq$k, $\cov\left[\WrkC\WrlC\right] $ satisfies
\begin{alignat*}{1}
\cov\left[\WrkC\WrlC\right] &=\mathbb{E}\left(W_{r,k}^{\mathrm{C}}W_{r,l}^{\mathrm{C}}\right)-\mathbb{E}\left(W_{r,k}^{\mathrm{C}}\right)\mathbb{E}\left(W_{r,l}^{\mathrm{C}}\right) \\
& \leq\frac{2}{r}\rho(\theta)^{k-1}\beta_{0}^{k+l-3}\max(\beta_{0},1)s(\theta)^{2}(2l-2).
\end{alignat*}
\end{lemma}
\begin{proof}
We begin by bounding all the terms in the expression for $\mathbb{E}\left(W_{r,k}^{\mathrm{C}}W_{r,l}^{\mathrm{C}}\right)$
given in Lemma \ref{lem:Exp_prod_of_cycle_gradient}. Using the Cauchy-Schwarz
inequality, $\left|t(\theta)\right|\leq s(\theta)\beta_{0}<s(\theta)$.
By Jensen's inequality, $\nabla m(\theta)^{2}\leq s(\theta)^{2}$
and thus $\left|t(\theta)\nabla m(\theta)\right|<s(\theta)^{2}$.
\begin{alignat*}{1}
\mathrm{term}\;\eqref{eq:cycgrad_term1a} & \leq\frac{1}{r}\rho(\theta)^{k-1}\beta_{0}^{k+l-3}s(\theta)^{2}\max(\beta_{0},1),\\
\mathrm{term}\;\eqref{eq:cycgrad_term1b} & \leq\frac{1}{r}\rho(\theta)^{k-2}\beta_{0}^{k+l-3}s(\theta)^{2}(k-2),\\
\mathrm{term}\;\eqref{eq:cycgrad_term2a} & \leq\frac{1}{r}\rho(\theta)^{k-2}\beta_{0}^{k+l-3}s(\theta)^{2}(k-1),\\
\mathrm{term}\;\eqref{eq:cycgrad_term2b} & \leq\frac{1}{r}\rho(\theta)^{k-1}\beta_{0}^{k+l-3}s(\theta)^{2}\max(\beta_{0},1),\\
\mathrm{term}\;\eqref{eq:cycgrad_term2c} & \leq\frac{1}{r}\rho(\theta)^{k-1}\beta_{0}^{k+l-3}s(\theta)^{2}(l-k-1).
\end{alignat*}
Note that 
\begin{alignat*}{1}
-\mathbb{E}\left(W_{r,k}^{\mathrm{C}}\right)\mathbb{E}\left(W_{r,l}^{\mathrm{C}}\right) & =-\left(\frac{m(\theta)}{x_{0}}-1\right)^{k+l-2}\left(\nabla m(\theta)\right)^{2}\\
 & \leq\beta_{0}^{k+l-3}\max(\beta_{0},1)s(\theta)^{2}.
\end{alignat*}
Now subtract $\mathbb{E}\left(W_{r,k}^{\mathrm{C}}\right)\mathbb{E}\left(W_{r,l}^{\mathrm{C}}\right)$
from $\mathbb{E}\left(W_{r,k}^{\mathrm{C}}W_{r,l}^{\mathrm{C}}\right)$.
\[
\mathrm{term\;}(\ref{eq:cycgrad_term1c})-\frac{1}{r}\sum_{i=1,i\geq k}^{r-l+1}\mathbb{E}\left(W_{r,k}^{\mathrm{C}}\right)\mathbb{E}\left(W_{r,l}^{\mathrm{C}}\right)=0.
\]
All the remaining (non-zero) terms of the subtraction of $\mathbb{E}\left(W_{r,k}^{\mathrm{C}}\right)\mathbb{E}\left(W_{r,l}^{\mathrm{C}}\right)$
from $\mathbb{E}\left(W_{r,k}^{\mathrm{C}}W_{r,l}^{\mathrm{C}}\right)$
may be bounded by summing the bounds above for (\ref{eq:cycgrad_term1a}),
(\ref{eq:cycgrad_term1b}), (\ref{eq:cycgrad_term2a}), (\ref{eq:cycgrad_term2b}),
(\ref{eq:cycgrad_term2c}) and $-\mathbb{E}\left(W_{r,k}^{\mathrm{C}}\right)\mathbb{E}\left(W_{r,l}^{\mathrm{C}}\right)(l+k-2)/r$ to obtain
\begin{align*}
\mathbb{E}\left(W_{r,k}^{\mathrm{C}}W_{r,l}^{\mathrm{C}}\right)-\mathbb{E}\left(W_{r,k}^{\mathrm{C}}\right)\mathbb{E}\left(W_{r,l}^{\mathrm{C}}\right) & \leq\frac{2}{r}\rho(\theta)^{k-1}\beta_{0}^{k+l-3}\max(\beta_{0},1)s(\theta)^{2}(l+k-2).
\end{align*}
Recalling that $l\geq k$, we have the result.
\end{proof}

%\begin{lemma}
%\label{lemma:covariance_cycle_bound_gradient}
%Let  $\rho(\theta)\eqdef 1 + \sigma(\theta)^2/(m(\theta) - x_0)^2$. For all $l\geq k$,
%\begin{align*}
%\var(\WrkC) &< \frac{2k}{r}\left(\frac{m(\theta)}{x_0}-1\right)^{2k-2}\nabla_\theta m(\theta)^2\rho(\theta)^{k-1}D_1(\theta),\\
%\cov\left[\WrkC\WrlC\right] &<\frac{2l}{r}\left\lvert\frac{m(\theta)}{x_0}-1\right\rvert^{l+k-2}\nabla_\theta m(\theta)^2\rho(\theta)^{k-1}D_2(\theta),
%\end{align*}
%where $D_1(\theta), D_2(\theta)$ are constants that depend on $\nabla_\theta m(\theta), s^2(\theta), \sigma^2(\theta), t(\theta)$ and $x_0$ but not on $k, l, r$.
%\end{lemma}

%%% Local Variables: 
%%% mode: latex
%%% TeX-master: "paper.tex"
%%% End: 

\section{Adaptive tuning of $p$}
\label{app:wnv}

\begin{figure}
\centering
\begin{tikzpicture}[every node/.append style={font=\normalsize}]
\node (img1) {\includegraphics[width = 0.3\textwidth]{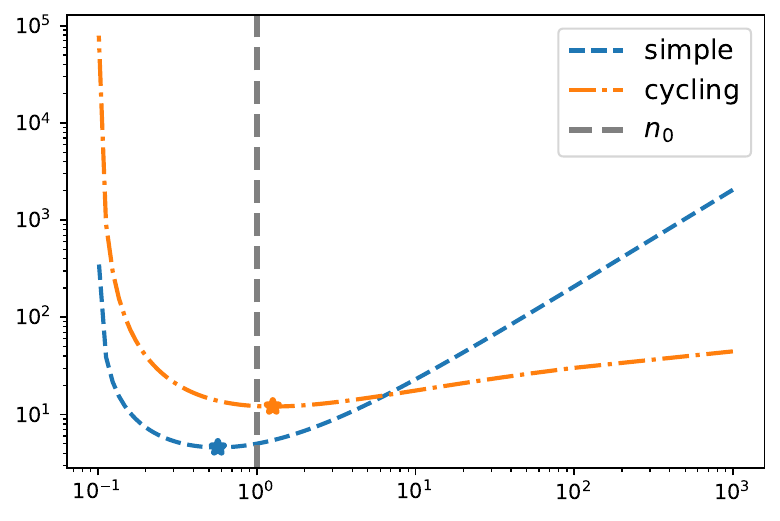}};
\node[above=of img1, node distance = 0, yshift = -1.2cm] {low};
\node[right=of img1, node distance = 0, xshift = -0.7cm] (img2) {\includegraphics[width = 0.3\textwidth]{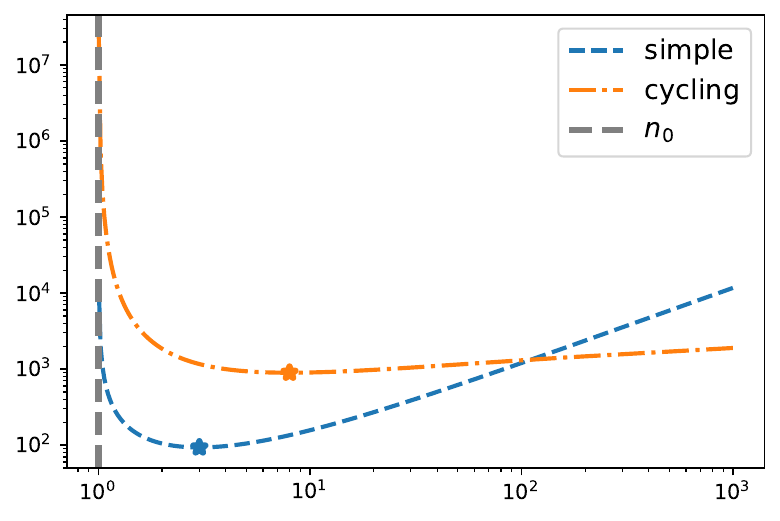}};
\node[right=of img2, node distance = 0, xshift = -0.7cm] (img3) {\includegraphics[width = 0.3\textwidth]{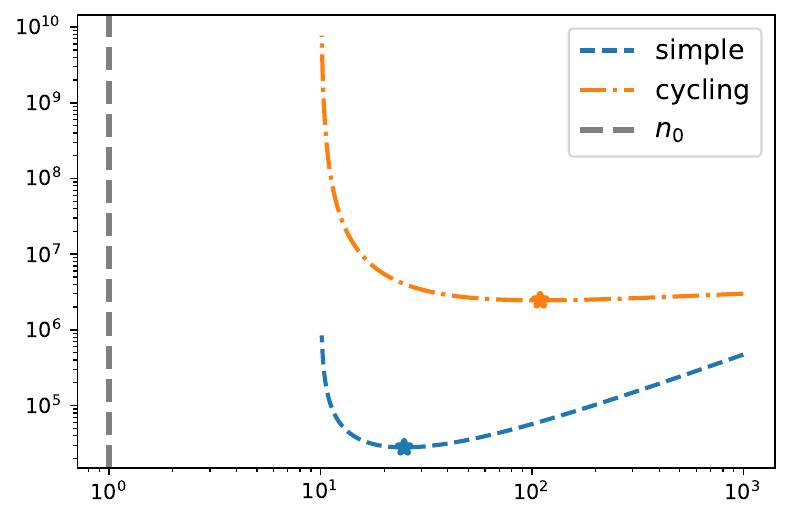}};
\node[above=of img2, node distance = 0, yshift = -1.2cm] {moderate};
\node[above=of img3, node distance = 0, yshift = -1.2cm] {high};
\node[left=of img1, node distance = 0, rotate = 90, anchor = center, yshift = -0.8cm] {$n_0=1$};
\node[below=of img1, node distance = 0, yshift = 1cm] (img4) {\includegraphics[width = 0.3\textwidth]{wnv_min_beta_point1.pdf}};
\node[right=of img4, node distance = 0, xshift = -0.7cm] (img5) {\includegraphics[width = 0.3\textwidth]{wnv_min_beta_one.pdf}};
\node[right=of img5, node distance = 0, xshift = -0.7cm] (img6) {\includegraphics[width = 0.3\textwidth]{wnv_min_beta_ten.pdf}};
\node[left=of img4, node distance = 0, rotate = 90, anchor = center, yshift = -0.8cm] {$n_0=10$};
\node[below=of img4, node distance = 0, yshift = 1cm] (img7) {\includegraphics[width = 0.3\textwidth]{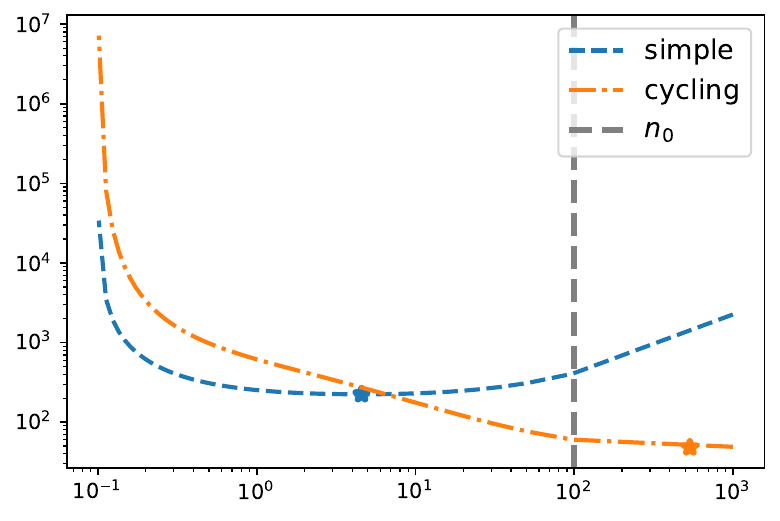}};
\node[right=of img7, node distance = 0, xshift = -0.7cm] (img8) {\includegraphics[width = 0.3\textwidth]{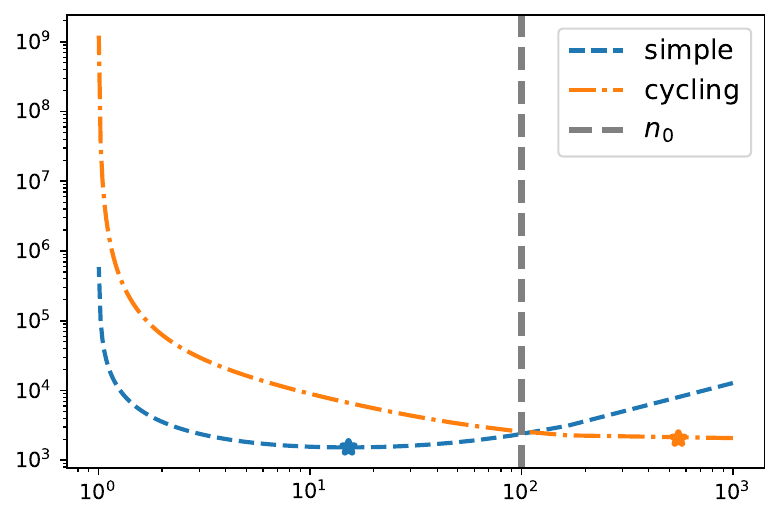}};
\node[right=of img8, node distance = 0, xshift = -0.7cm] (img9) {\includegraphics[width = 0.3\textwidth]{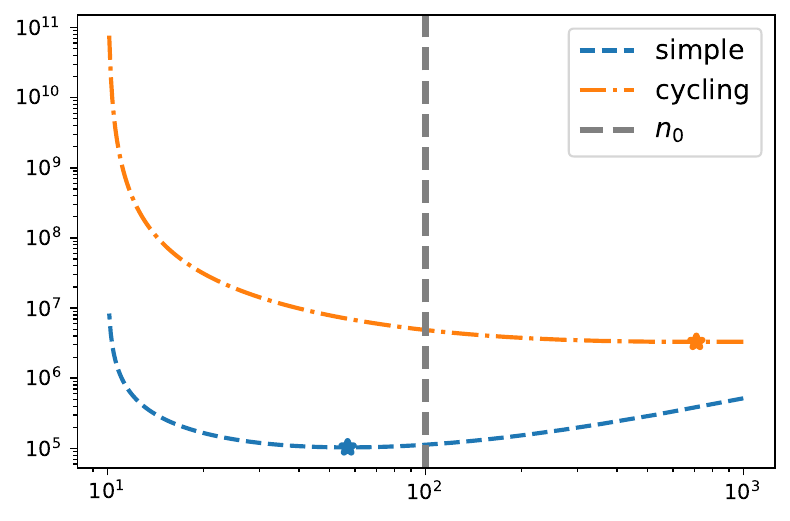}};
\node[left=of img7, node distance = 0, rotate = 90, anchor = center, yshift = -0.8cm] {$n_0=100$};
\node[below=of img7, node distance = 0, yshift = 1.2cm] {$\E[R]$};
\node[below=of img8, node distance = 0, yshift = 1.2cm] {$\E[R]$};
\node[below=of img9, node distance = 0, yshift = 1.2cm] {$\E[R]$};
\end{tikzpicture}
\caption{Comparison of work-normalised variance in a low, moderate and high variance scenario for several values of $n_0$. The stars denote the value of $\E[R]$ minimising the upper-bound of the WNV and the vertical line $\E[R]=n_0$.}
\label{fig:wnv_app}
\end{figure}

As discussed in Section~\ref{subsec:budget}, the choice of $p$ is crucial in obtaining estimators which do not exhibit high variance. However, the results in Section~\ref{sec:general_prop} only give the range in which $p$ should be to guarantee finite variance, $(0, 1-\beta^2)$.

Ideally, one would like to select $p$ that minimises the work-normalised variance $\wnv = (n_0 + \E[R]) \times \var[\sumest]$. However, in general, $\var[\sumest]$ is unknown. One could replace $\var[\sumest]$ with the upper-bounds established in Section~\ref{sec:general_prop} and minimise instead an upper-bound on $\wnv$ \citep{blanchet2015unbiased}.
The upper bounds for the simple and cycling estimators are
\begin{align*}
&(n_0 + \E[R]) \times \var[\sumestS] \\
&\qquad\leq c^2 (n_0 + (1-p)/p)\left[\frac{\beta_0^2}{(1-\beta_0)^2} \times\frac{p}{1-p-\beta_0^2}+\frac{(1+\beta)}{1-\beta}\times\frac{(1-p)}{1-p-\beta^2}\right]\\
&(n_0 + \E[R]) \times \var[\sumestC] \\
&\qquad\leq c^2 (n_0 + (1-p)/p)\left[\frac{\beta_0^2}{(1-\beta_0)^2} \times\frac{p}{1-p-\beta_0^2}+\frac{4 p \log(1/p)}{(1-p-\beta^2)^2}\left[\beta^2+\frac{2\beta_0}{1-\beta_0^2}\right]\right]
\end{align*}
and are obtained by combining Proposition~\ref{prop:var_exp} and~\ref{prop:exp_var_convergence_simple} and Proposition~\ref{prop:var_exp} and~\ref{prop:exp_var_convergence_cycle}, respectively.
In the interval $(0, 1-\beta^2)$, both bounds are convex functions of $p$ whose unique minimum can be found using standard optimisation routines.

We compare the optimal value of $p$ obtained by minimising the bounds above in three scenarios: we set the mean of the $X_i$'s to be $m=1$ and consider $\sigma^2=0.1, 1, 10$ which correspond to a low, moderate and high variance scenario.
We consider three values of $n_0$ for the pilot run.
Figure~\ref{fig:wnv_app} shows the results for the simple and the cycling estimators. As expected, higher input variance requires higher $\E[R]$ to control the variance of the estimators. Generally, the optimal $p$ is larger for the cycling estimator, this is a consequence of the slower rate of divergence shown in Proposition~\ref{prop:exp_var_convergence_cycle}.

We can see that, in the low and moderate variance case,  our strategy of setting $\E[R]=n_0$ corresponds to selecting $p$ close to the optimal value for the cycling estimator.

\section{MVUE}
\label{app:mvue}

Let $Y_i\eqdef X_i/x_0 -1$, then~\eqref{eq:mvue} may be rewritten as $U_{r,k}^{\mathrm{MVUE}} =
S_{r,k} / \binom{r}{k}$, with
\[
	S_{r, k} \eqdef
	\sum_{1\leq i_1<\dots<i_k\leq r} Y_{i_1}\cdot\dots\cdot Y_{i_k},
\]
for $r\geq k \geq 1$, which may be computed recursively as follows:
\[
S_{r, k} = S_{r-1, k} + Y_r S_{r-1, k-1}
\]
with initialisations $S_{r, k} \eqdef 0$ for $r<k$, and $S_{r,0} \eqdef 1$ for $r \geq 0$.
For a fixed $r$, it is thus possible to compute $U_{r,k}^{\mathrm{MVUE}}$ for
$k=1,\ldots, r$ at a $\OO(r^2)$ cost. This is the same computational complexity as
the cycling method.

Figure~\ref{fig:boxplots_toy_mvue} compares the variability of the simple,
cycling and MVUE estimators in the toy model example of
Section~\ref{subsec:toy_lvm}, when $d=2$, $n=1$ and $p=10^{-2}$. One sees that
cycling and MVUE seems to lead to the same level of performance.

\begin{figure}
  \centering
  \includegraphics[scale=0.5]{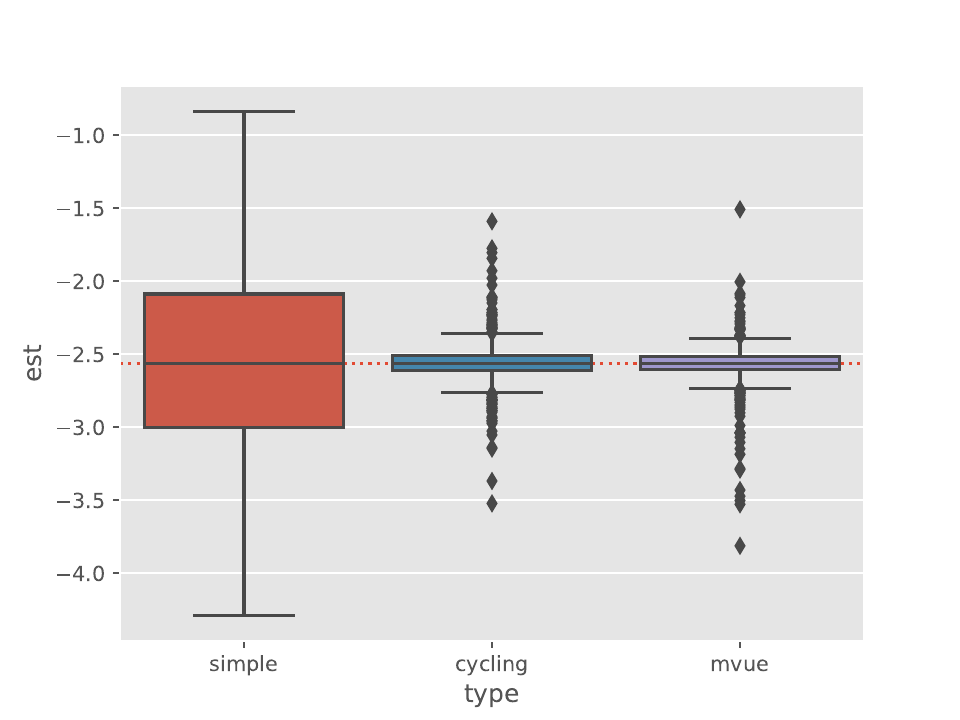}
  \caption{Box-plots of $10^3$ independent realisations of the simple, cycling,
    and MVUE estimators for the toy LVM model when $d=2$, $n=1$ and $p=10^{-2}$.
    The true value is indicated by the red dashed horizontal line.
  \label{fig:boxplots_toy_mvue}}
\end{figure}

%%% Local Variables:
%%% mode: latex
%%% TeX-master: "paper.tex"
%%% End:
\section{Additional details for the numerical experiments}

\subsection{Toy LVM: MLMC estimators}
\label{app:toy_lvm}

For a given $y\in\{y_1,\ldots,y_n\}$ and a given $\theta$, the exact expressions
for the moments of $X_j=p(y|Z_j,\theta)$ and $\beta^2$ are
\begin{align*}
  m & = p(y |\theta)= \N(y;\theta \textsf{1}_d, 2\textsf{Id}_d),\\
  \sigma^2 & = (4\pi)^{-d/2}\N(y;\theta \textsf{1}_d, \frac{3}{2} \textsf{Id}_d)
             - \N(y;\theta \textsf{1}_d, 2 \textsf{Id}_d)^2\\
\beta^2 & = 1-\left(\frac{4}{3}\right)^{-d/2}\exp\left(-\frac{\Vert
          y-\theta\textsf{1}_d\Vert^2}{6}\right).
          %\label{eq:toylvm-beta}
\end{align*}
The maximum likelihood estimator is given by $\theta^\star = (dn)^{-1}\sum_{i=1}^n \textsf{1}_d^T y_i$.

Under the assumption that $\widetilde{R}\sim \Geomp$, the MLMC estimators have finite variance and finite expected cost whenever $0.5<p<3/4$ \citep[Theorem 2 and 3]{shi2021multilevel}. However the second moment of the cost is infinite. To see this consider the second moment of the cost
\begin{align*}
\sum_{k=0}^\infty  2^{2(j+1+k)}\pr(\widetilde{R}=k) = \sum_{k=0}^\infty  2^{2(j+1+k)}p(1-p)^{k} = \frac{4^{j+1}p}{1-4(1-p)},\quad\textrm{provided } p>3/4.
\end{align*}
Hence, it is not possible to guarantee that MLMC estimators have finite expectation and cost with finite variance at the same time.

\subsection{Toy LVM: High dimensions}

\begin{figure}
\centering
\begin{tikzpicture}[every node/.append style={font=\normalsize}]
\node (img1) {\includegraphics[width = 0.75\textwidth]{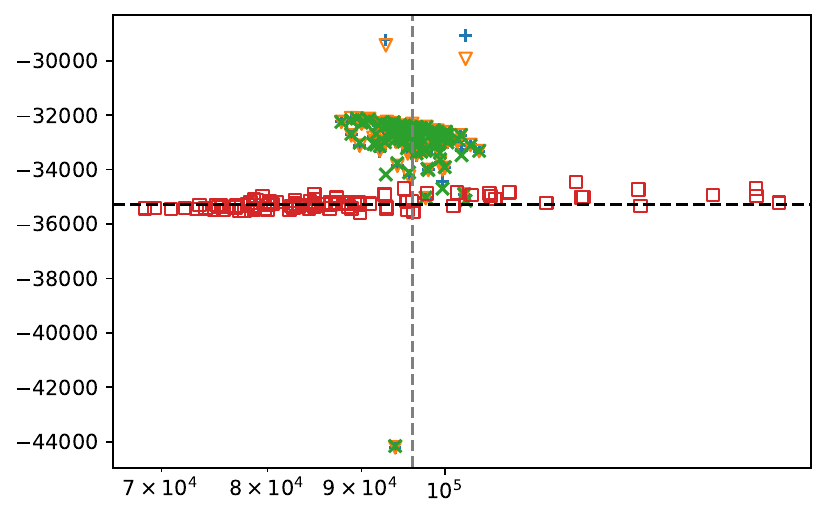}};
\node[below=of img1, node distance = 0, yshift = 1.2cm] {Cost};
\node[left=of img1, node distance = 0, rotate = 90, anchor = center, yshift = -0.8cm] {$\sum_{i=1}^n \log p(y_i|\theta)$};
\node[below=of img1, node distance = 0, yshift = 0.5cm, xshift = 0cm] (img9) {\includegraphics[width = 0.8\textwidth]{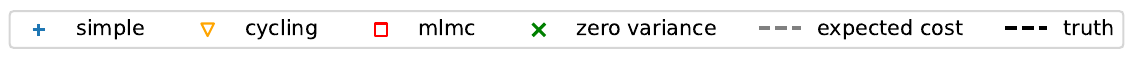}};
\end{tikzpicture}
\caption{Estimates of the marginal log-likelihood for a full data set of size $n=1000$ plotted against computational cost (measured in number of samples $Z_i$ drawn) for an expected cost $E[R] = 96$ samples per data point and dimension $d=20$. The vertical dashed line denotes the expected cost while the horizontal one denotes the true value of $\sum_{i=1}^n \log p(y_i|\theta)$.}
\label{fig:lvm_hd}
\end{figure}

Figure~\ref{fig:lvm_hd} shows the results obtained with simple and cycling estimators, MLMC in the case $d=20$ with $p=1/97$ which gives $\Exp[R]=96$. The cost of the simple and cycling estimators is stable but the estimate itself is not.

To further probe the cause of the large variance of the Taylor-based sum estimators, we separate the contribution to the variance due to the random truncation of the Taylor series, given by~\eqref{eq:varE}, from that due to the variance of the inputs $X_1, \dots, X_r$ given by~\eqref{eq:Evar}.
To find this contribution to the variance, we replace the inputs $X_j = p(y|Z_j, \theta)$ with inputs with zero variance, i.e $X_j \equiv p(y|\theta)$. In this case, simple and cycling estimators coincide and~\eqref{eq:Evar} is zero.

As can be seen from Figure~\ref{fig:lvm_hd}, the estimates obtained when $X_j = p(y|\theta)$ overlap with the estimates with $X_j=p(y|Z_j, \theta)$, which implies that the poor performance of the Taylor-based sum estimators is attributable as much to the random truncation of the Taylor series as it is to the large variance of the $X$'s.
In particular, as $d$ increases, the expansion point $x_0$ is no-longer close to $m$, and so $\beta^2$ approaches 1.
With our choice of $p$, ~\eqref{eq:varE} is not finite.
In order to guarantee finite variance for all $i=1,\dots, n$ we would need to set $\Exp[R]$ of the order $10^8$.

Similar considerations apply to MLMC, in fact, for this example we have $\textrm{KL}(p(\cdot|y, \theta), q(\cdot; y)) = \frac{d}{2}\left(\log \frac{4}{3}-\frac{1}{4}\right) \to \infty$ as $d\to \infty$, hence the importance weights get more and more degenerate as $d$ increases.

\subsection{Independent component analysis}
\label{app:ica}
Recall that $\theta=(A, \sigma^2)$,
% The model is a censored logistic model under which $z_j = b_j\zeta_j$ with
% $b_j\sim \textrm{Bernoulli}(\alpha)$ and $\zeta_j \sim \textrm{Logistic}(1/2)$,
% we assume that $\alpha$ is known and fixed. In particular,
\begin{align*}
p(z|\theta) & =p(z) = \prod_{j=1}^{d_z} \left\{\alpha\textrm{Logistic}(z_j; 1/2) + (1-\alpha) \delta_0(z_j)\right\},\\
p(y|z,\theta)&= \N(y; Az, \sigma^2\textsf{Id}_{d_y}).
\end{align*}

Given that the prior does not depend on $\theta$, we can write
\begin{align*}
	\nabla_{\theta} p(y|\theta) & = \int \nabla_{\theta} \log p(y|z,\theta)p(y|z,\theta) p(z) \rmd z,
\end{align*}
where $\nabla_{\theta}\log p(y|z, \theta)$ may be decomposed into
\begin{equation*}
\nabla_\sigma\log p(y|z, \theta) = \frac{1}{\sigma^3}\Vert y -Az\Vert^2-\frac{d_y}{\sigma}, \qquad
\nabla_A\log p(y|z, \theta) = \frac{1}{\sigma^2}( y-Az)z^T.
\end{equation*}

We select the following proposal distribution, as it is proportional to the
likelihood $p(y|z,\theta)$:
\begin{align}
	\label{eq:proposal_ica}
	q(z; y, \theta) = \N(z; (A^TA)^{-1}A^Ty, \sigma^2(A^TA)^{-1}),
\end{align}
where again $\theta=(A, \sigma^2)$. The corresponding importance weights equal
\begin{align*}
	w(z; y, \theta):=\frac{(2\pi\sigma^2)^{(d_z-d_y)/2}}{\textrm{det}(A^TA)^{1/2}}\exp\left(-\frac{\Vert y\Vert^2 - y^TA(A^TA)^{-1}A^Ty}{2\sigma^2}\right)p(z).
\end{align*}

For the Taylor-based estimators, the $X_i$'s are obtained by sampling $Z_i$ from $q$ in~\eqref{eq:proposal_ica} and setting $X_i = w(Z_i; y, \theta)$. To obtain the unbiased estimators of $\nabla_{\theta}p(y|\theta)$ we set $G_i = w(Z_i; y, \theta)\nabla_{\theta}\log p(y|Z_i,\theta)$, with $Z_i$ sampled from $q$.

An alternative and simpler proposal is given by $q(z) = p(z)$.
However, this proposal leads to $(X_i, G_i)$'s with high variance which require large $p$ to ensure finite variance of the Taylor-based gradient estimators. On the contrary, with~\eqref{eq:proposal_ica} it is sufficient to take $p\in(0,0.9)$ to guarantee finite variance, and thus our choice of $p=1/n_0=0.1$ is valid.
Similar considerations apply to MLMC with proposal $q(z) = p(z)$, the estimates of the gradient are unstable due to the high variability of the weights.

For SAEM, we draw the initial value of $\sigma$ from a Uniform distribution over the interval $[0, 1]$ and the initial distribution of $A=(a_1, a_2)$ from a Uniform distribution over $[0, 1]^{512}$.
We select the learning rate for SAEM to be $\Delta_t = 1/t$, with $t$ denoting the iteration number, which satisfies the conditions in \cite{delyon1999convergence} to guarantee convergence.
The MCMC algorithm targeting the posterior $p(z|y, \theta)$ is the Metropolis-within-Gibbs algorithm described in \citet[Section 3.3]{allassonniere2012stochastic}.

Given the high dimension of the parameter space, gradient-based methods struggle if the starting point of the iteration is far from the optimum. Hence, we initialise all SGD algorithms at the value of $\theta$ obtained after one iteration of SAEM \citep{polyak1992acceleration}.
We use two different learning rates for $\sigma$ and $A$, since the gradients of the two parameters are of different order of magnitude. We set $\Delta_\sigma = 10^{-7}$ and $\Delta_A = 10^{-4}$. This is in line with the suggestion of \citet[Section 5.1]{beatson2019efficient}.
We stop all algorithms after 5000 iterations to guarantee convergence of all algorithms.

At each iteration of SGD, we set $x_0$ using $n_0=10$ pilot runs and obtain the corresponding $\beta^2$ as described in Section~\ref{subsec:budget}.
If $1/n_0 \in(0, 1-\beta^2)$, we set $p=1/n_0$, i.e. $\E[R]=n_0$; otherwise we set $p=(1-\beta^2)/2$.
For MLMC we adjust $j, p$ so that $\Exp\left[2^{j+1+\widetilde{R}}\right]\approx n_0$.

For SGD we use mini-batches of size $2$, while for SAEM we scan the whole data
set and, as suggested in \citet[Section 3.3]{allassonniere2012stochastic},
randomly update one component of the latent vector $z$ for each data point $y$.
This can be seen as a form of mini-batch sampling as noticed in~\cite{kuhn2020properties}.

The $\mse$ for the parameters $(\sigma, A)$ is defined as $\mse(\sigma) := \Exp\left[(\hat{\sigma}-\sigma)^2\right]$,
\begin{align*}
\mse(A) := \Exp\left[\frac{1}{d_y}\sum_{i, j=1}^{d_y^{1/2}}((\hat{a}_1-a_1)_{ij})^2 + \frac{1}{d_y}\sum_{i, j=1}^{d_y^{1/2}}((\hat{a}_2-a_2)_{ij})^2\right],
\end{align*}
where $\hat{\sigma}, \hat{a}_1, \hat{a}_2$ are the estimates obtained with each method.

To assess convergence of SAEM and SGD we check the value of $\sigma$ and of the trace of $a_1$ and $a_2$ over iterations (Figure~\ref{fig:ica_convergence}).

\begin{figure}
\centering
\begin{tikzpicture}[every node/.append style={font=\normalsize}]
\node (img1) {\includegraphics[width = 0.3\textwidth]{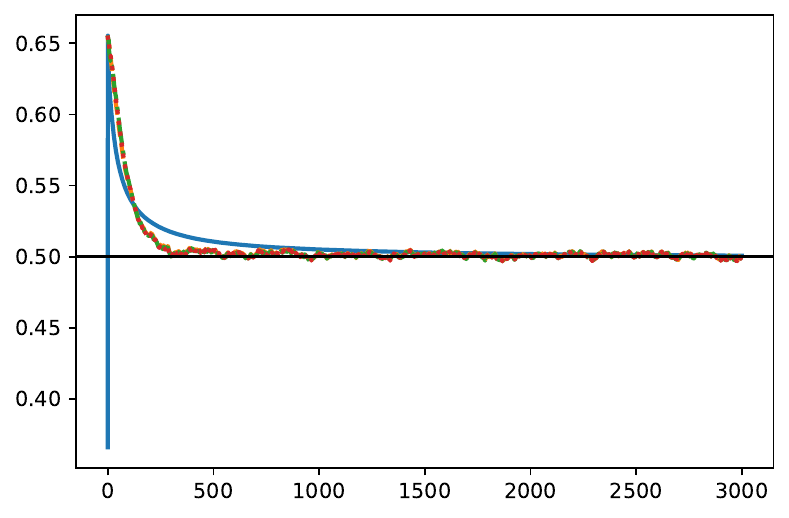}};
\node[right=of img1, node distance = 0, xshift = -0.9cm] (img2) {\includegraphics[width = 0.3\textwidth]{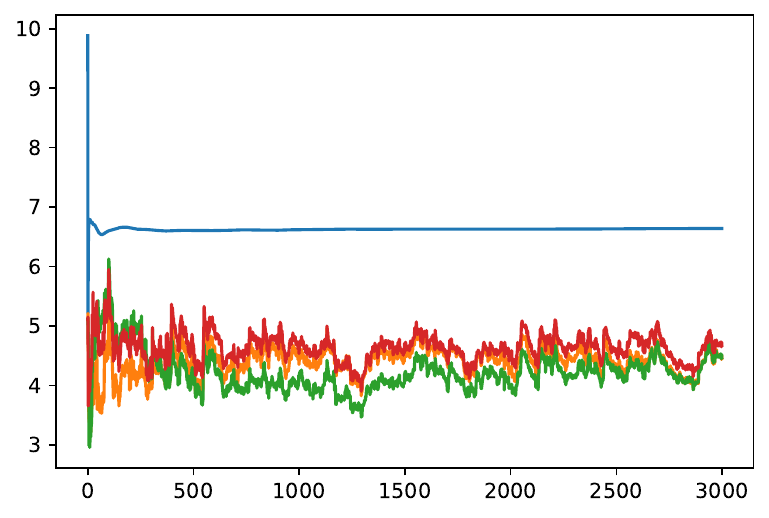}};
\node[right=of img2, node distance = 0, xshift = -0.9cm] (img3) {\includegraphics[width = 0.3\textwidth]{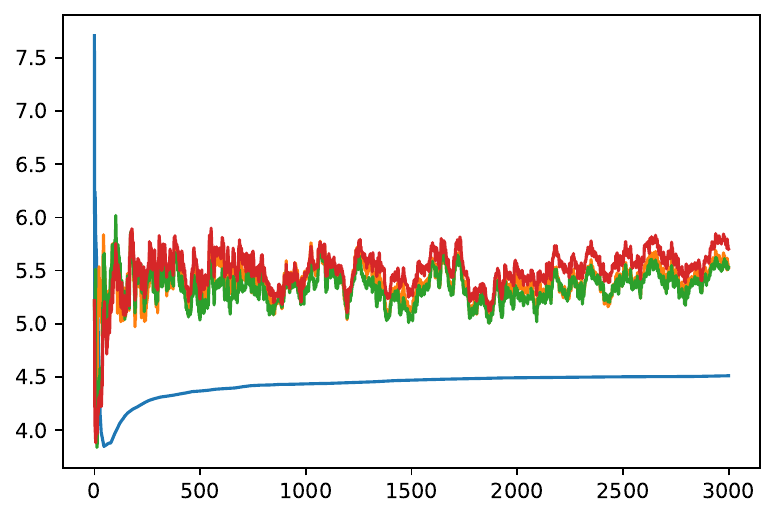}};
\node[above=of img1, node distance = 0, yshift = -1.2cm] {$\sigma$};
\node[above=of img3, node distance = 0, yshift = -1.2cm] {$\textrm{trace}(a_2)$};
\node[above=of img2, node distance = 0, yshift = -1.2cm] {$\textrm{trace}(a_1)$};
\node[below=of img1, node distance = 0, yshift = 1.2cm] {Iteration};
\node[below=of img2, node distance = 0, yshift = 1.2cm] {Iteration};
\node[below=of img3, node distance = 0, yshift = 1.2cm] {Iteration};
\node[below=of img1, node distance = 0, yshift = 0.5cm, xshift = 4cm] (img9) {\includegraphics[width = 0.6\textwidth]{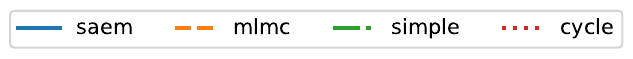}};
\end{tikzpicture}
\caption{Converge of SAEM and SGD for the ICA model.}
\label{fig:ica_convergence}
\end{figure}
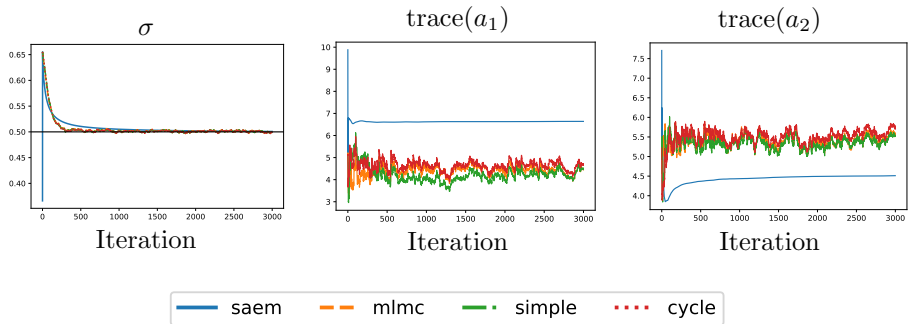

\subsection{Exponential random graphs}
\label{app:erg}

For a given $\theta$, we obtain an unbiased estimate of $\mathcal{Z}(\theta)$ by
running a waste-free SMC sampler \citep{dau2022waste} based on adaptive
tempering, i.e. the intermediate distributions are $\pi_t(y) \propto
\exp\left\{\lambda_t s(y)\right\}$, with $0=\lambda_0<\ldots<\lambda_T=\theta$,
and at beginning of iteration $t$, $\lambda_t$ is set so that ESS (effective
sample size) equals $N/2$, where $N$ is the number of particles. Since we use a
waste-free sampler, $N=M\times P$, where $M$ is the number of resampled
particles, here $M=100$, and $P$ is the length of the simulated chains, here
$P=500$. For the low variance setting, we use the same approach except
$M=10^3$. Note that the initial distribution $\pi_0$ is easy to sample from, as
it is the uniform distribution over the $\binom{n}{2}$ possible networks. The
Markov kernel used at time $t$ to move the particles is a Metropolis kernel
based on the following proposal: one edge $(i, j)$ is chosen at random, and
$y_{ij}$ is flipped. We use the Python package \texttt{particles}
(\url{https://github.com/nchopin/particles}) to implement this SMC sampler.

The proposal $q$ is Gaussian with mean $\mu = (-3.27, 0.6)$ and covariance matrix
\begin{align*}
\Sigma = \begin{pmatrix}
 0.2016 & -0.0432\\
-0.0432 &  0.0108
\end{pmatrix}.
\end{align*}
This proposal was obtained through a small number of pilot runs.

The $\theta_i$'s are generated using a scrambled Sobol' sequence, 
see~\cite{OwenScrambledSobol}.

\subsection{Exponential random graphs: Comparison with the exchange algorithm}

We compare our approach with one of the state of the art for Bayesian inference on ERG models, i.e. the exchange algorithm of~\cite{caimo2011bayesian} implemented in the \texttt{bergm} package in \texttt{R}. We use the same setup described in Section 4.5 therein. The shape of the posterior distribution seems to be in good agreement with the obtained in our experiments (see Figure~\ref{fig:ergm_posterior_appendix}).

\begin{figure}
	\centering
	\begin{tikzpicture}[every node/.append style={font=\normalsize}]
		\node (img1) {\includegraphics[width = 0.3\textwidth]{ergm_posterior_simple.pdf}};
		\node[right=of img1, node distance = 0, xshift = -0.7cm] (img2) {\includegraphics[width = 0.3\textwidth]{ergm_posterior_cycle.pdf}};
		\node[right=of img2, node distance = 0, xshift = -0.7cm] (img3) {\includegraphics[width = 0.3\textwidth]{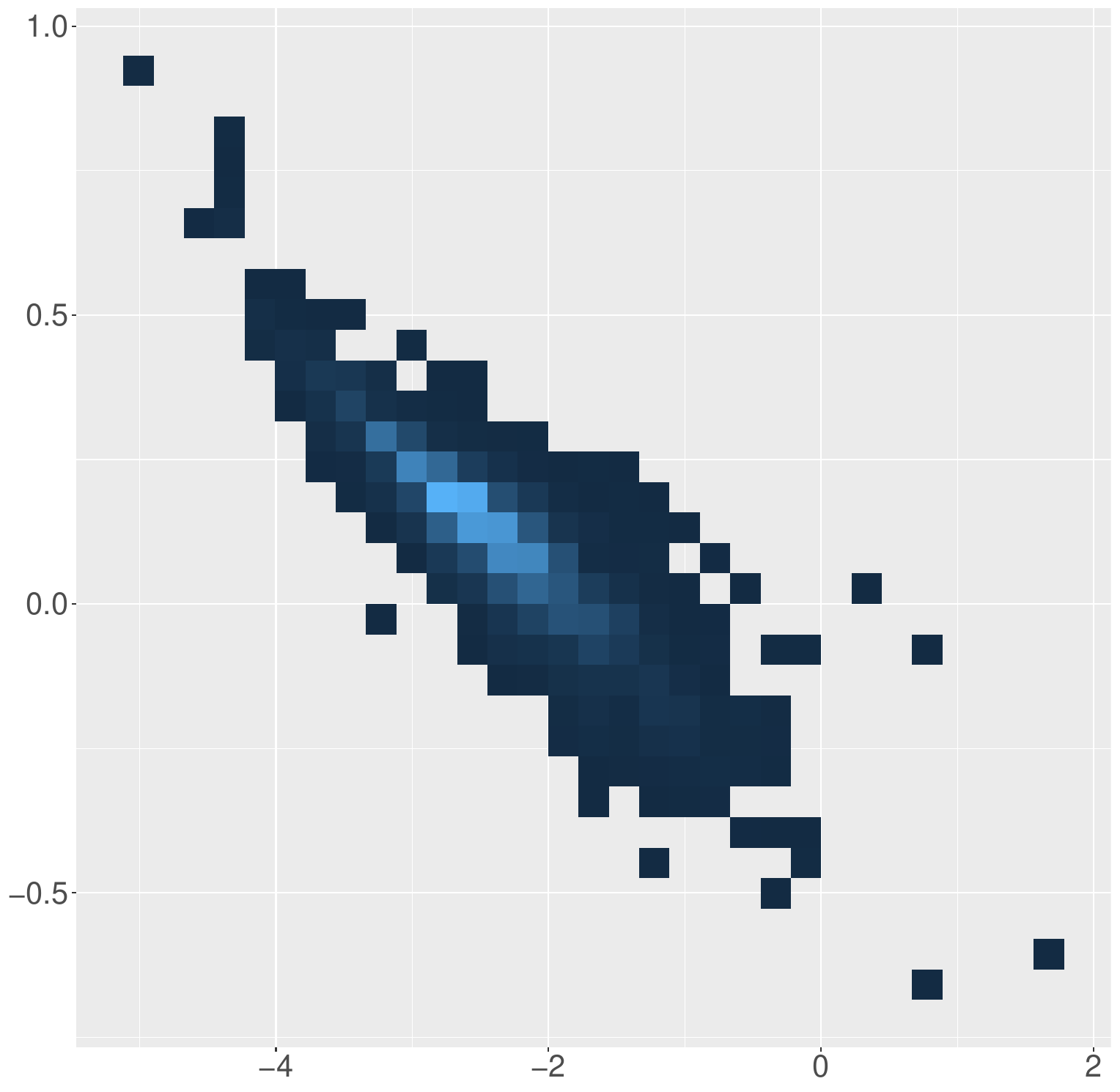}};
		\node[above=of img1, node distance = 0, yshift = -1.2cm] {Simple};
		\node[above=of img2, node distance = 0, yshift = -1.2cm] {Cycling};
		\node[above=of img3, node distance = 0, yshift = -1.2cm] {Exchange algorithm};
	\end{tikzpicture}
	\caption{Bivariate weighted histograms approximating the posterior distributions obtained with the simple estimator, the cycling estimator and the \texttt{bergm} package implementing the exchange algorithm of \cite{caimo2011bayesian}.
	}
	\label{fig:ergm_posterior_appendix}
\end{figure}

As we now discuss, there are good reasons in favour of our method. First, the exchange algorithm of \cite{caimo2011bayesian} requires running nested MCMC chains where the inner one samples from the likelihood $p(y|\theta)$; in addition, in the version of the algorithm implemented in \texttt{bergm}, population Monte Carlo is employed to improve mixing. As a consequence, the resulting method is rather expensive and, due to the serial nature of MCMC, harder to parallelise than our proposed approach, in which each $\mathcal{Z}(\theta)$ is obtained from different independent runs of SMC which can run in parallel.

Secondly, a fundamental task in the context of network model is the ability to select the model that best fits the data. This type of inference requires estimation of the model evidence $p(y)$.
To the best of our knowledge, given the output of the exchange algorithm one can compute the model evidence \[p(y) = p(\theta)\frac{\exp\{\theta^Ts(y)\}}{\mathcal{Z}(\theta)}\frac{1}{p(\theta|y)}\] in two ways:
\begin{enumerate}
\item An approach similar to ours in which $\log \mathcal{Z}(\theta)$ is estimated via path sampling, and $p(\theta|y)$ via the exchange algorithm and a kernel density estimate (KDE) \cite[Section 4]{caimo2013bayesian}.
\item Pseudo-likelihood approaches in which the intractable likelihood is replaced by a tractable approximation \citep{bouranis2018bayesian}.
\end{enumerate}

The first approach requires running the exchange algorithm to sample from $p(\theta|y)$, from which a KDE is formed. Path sampling with auxiliary MCMC chains is used to sample from $p(y|\theta)$. 
This approach is clearly only suitable for a small number of parameters due to the instability of the KDE in high dimensions. Also, as mentioned by the authors, this strategy incurs two sources of error when estimating $\log \mathcal{Z}(\theta)$: the Monte Carlo error to approximate the expectation w.r.t. $p(y|\theta)$, and the error due to discretising the path integral. The resulting approximation of the model evidence is naturally biased.
Our approach allows computation of the model evidence directly from the samples used to estimate the posterior, without requiring additional samples from $p(y|\theta)$ and without the use of KDE. It also provides unbiased estimators which can be used within pseudo-marginal methods for model selection.

The second approach also cannot provide unbiased estimators of the model evidence and the quality of the approximation largely depends on the pseudo-likelihood used. To obtain a likelihood approximation close to the true one, \cite{bouranis2018bayesian} propose several improvements. We tested the evidence estimation in the \texttt{bergm} package on our toy example and found that the Monte Carlo MLE estimator needed to implement their proposed adjustments to the pseudo-likelihood could not converge since the underlying MCMC was not mixing.

\subsection{Comparison with biased estimators}
To further motivate the need for unbiased estimators, we compare the MLMC estimator of \cite{shi2021multilevel} and our simple and cycling estimators with their biased counterparts.

\subsubsection{Toy LVM}
We first consider the toy LVM and two biased estimators:
	\begin{itemize}
		\item the naive estimator $f(\bar{X})$ where $\bar{X}$ is the mean of the $X_i$'s;
		\item the IWAE estimator $I_j$ given below~\eqref{eq:mlmc_mll};
	\end{itemize}
and their unbiased counterparts: simple/cycling and MLMC, respectively (see Figure~\ref{fig:convergence_iwae} and Figure~\ref{fig:bias_boxplot}).

Figure~\ref{fig:convergence_iwae} shows the behaviour of the relative mean absolute deviation $\E [| f(m)-\hat{f}|] /|f(m)|$ and bias as the expected computational cost increases for the toy LVM model with $d=2$. 

\begin{figure}
	\centering
	\begin{tikzpicture}[every node/.append style={font=\normalsize}]
		\node (img1) {\includegraphics[width = 0.45\textwidth]{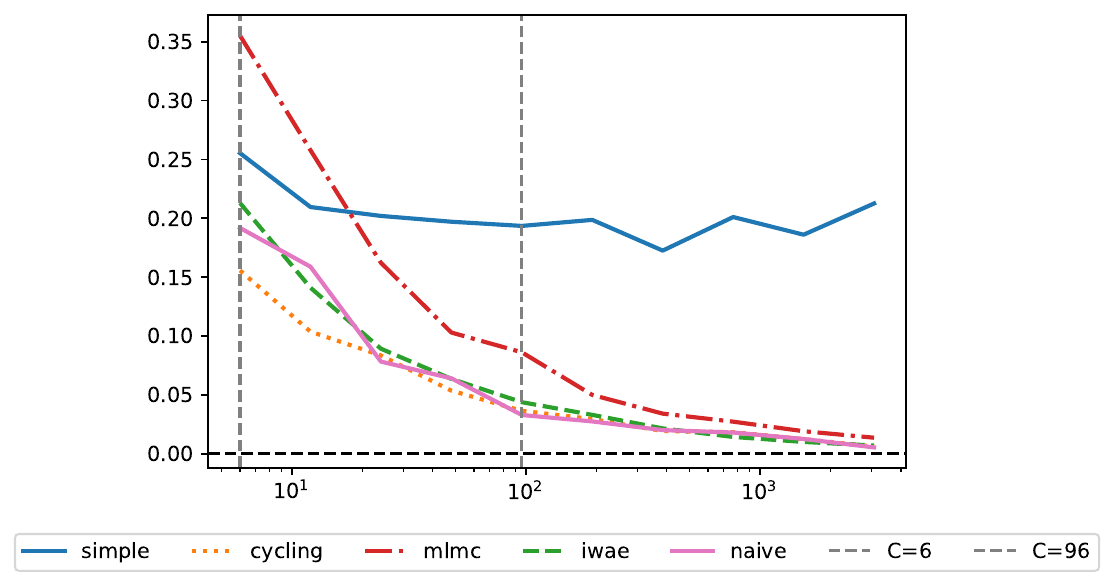}};
		\node[below=of img1, node distance = 0, yshift = 1.2cm] {cost};
		\node[left=of img1, node distance = 0, rotate = 90, anchor = center, yshift = -0.8cm] {$\E[|f(m)-\hat{f}|]/|f(m)|$};
		\node[right=of img1, node distance = 0, xshift = -0.5cm] (img2) {\includegraphics[width = 0.45\textwidth]{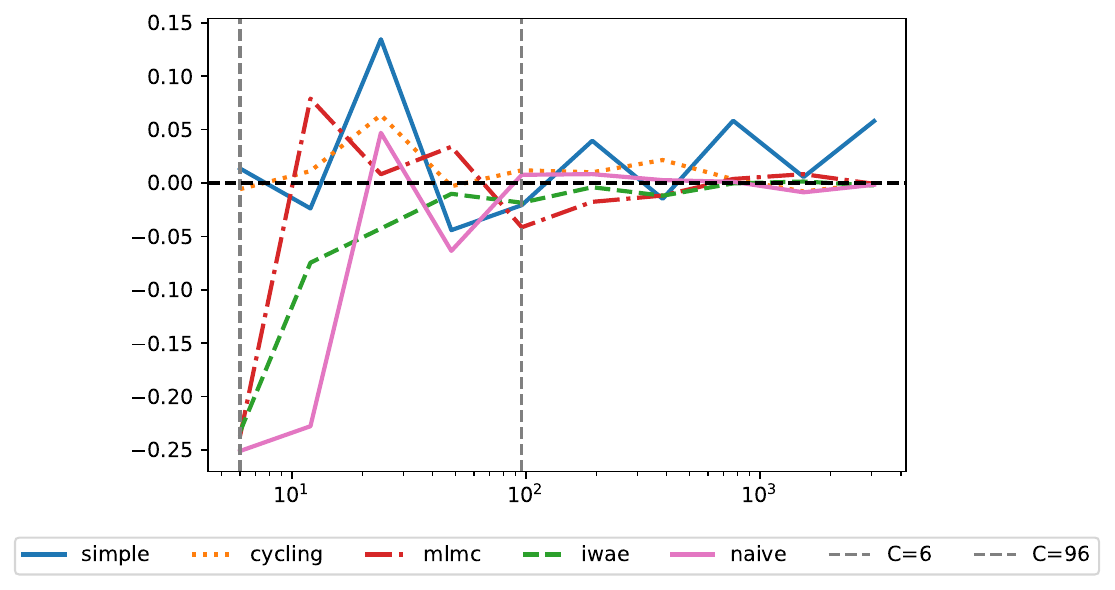}};
		\node[below=of img2, node distance = 0, yshift = 1.2cm] {cost};
		\node[left=of img2, node distance = 0, rotate = 90, anchor = center, yshift = -0.8cm] {$f(m)-\E[\hat{f}]$};
	\end{tikzpicture}
	\caption{Mean absolute deviation and bias over 100 repetitions against cost for 2 biased estimators (IWAE, naive) and 3 unbiased estimators (MLMC, simple, cycling). The vertical dashed lines are the expected costs used in Figure 2 of the revised/original manuscript, namely $E[R]=6$ and $E[R]=96$.}
	\label{fig:convergence_iwae}
\end{figure}

For large enough cost all estimators except the simple one have very small relative mean absolute deviation. This is because the variance of the simple estimator does not converge to zero as the cost increases, as shown in Example~\ref{ex:simple_var}. However, for small cost the average (out of 100 replicates) bias of both IWAE and the naive estimator is considerably larger than that of the cycling estimator (see Figure~\ref{fig:convergence_iwae}).

\begin{figure}
	\centering
	\begin{tikzpicture}[every node/.append style={font=\normalsize}]
		\node (img1) {\includegraphics[width = 0.45\textwidth]{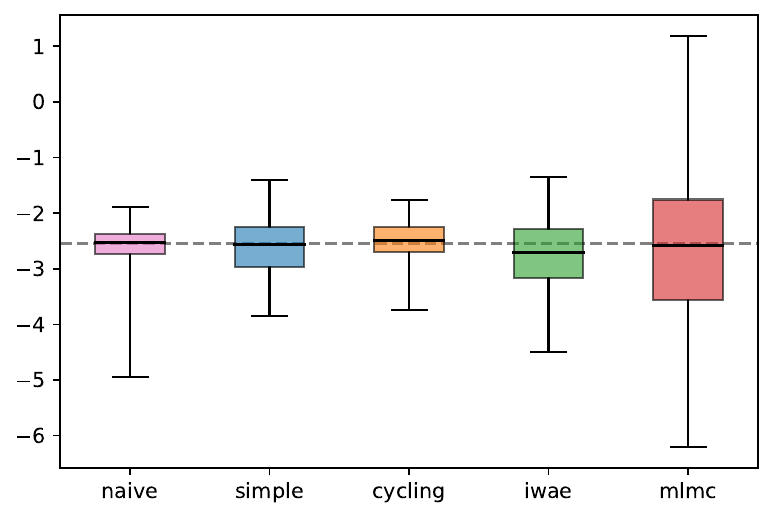}};
		\node[left=of img1, node distance = 0, rotate = 90, anchor = center, yshift = -0.8cm] {$\hat{f}$};
		\node[right=of img1, node distance = 0, xshift = -0.5cm] (img2) {\includegraphics[width = 0.45\textwidth]{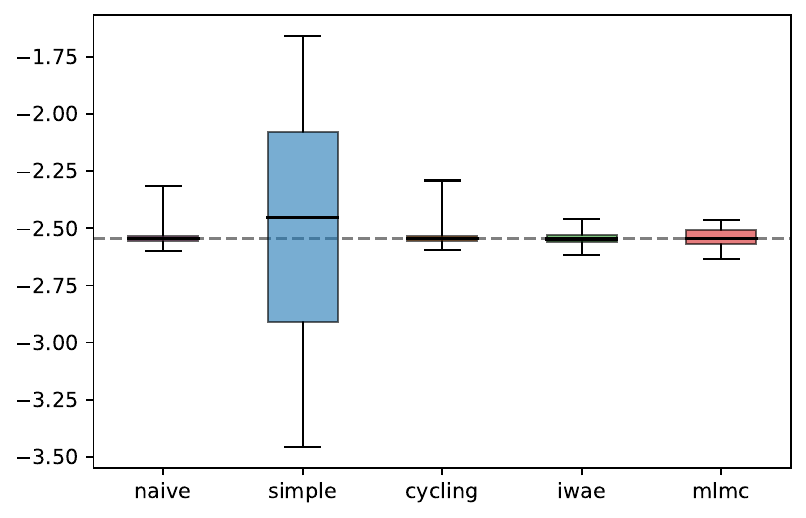}};
		\node[left=of img2, node distance = 0, rotate = 90, anchor = center, yshift = -0.8cm] {$\hat{f}$};
	\end{tikzpicture}
	\caption{Distribution of the estimates of $f(m)$ for small ($E[R]=6$) and high ($E[R]=96$) cost. The dashed horizontal line is the true value.}
	\label{fig:bias_boxplot}
\end{figure}

Figure~\ref{fig:bias_boxplot} shows the distribution of the 100 replicates for the smallest cost (left-most vertical dashed line) and the highest one. We can see that cycling provides a provably unbiased estimator without increasing the variability in the original biased naive estimator. When the cost is large it is also clear that cycling outperforms simple since the variability of the latter does not decrease to zero as cost increases.

\subsubsection{Independent component analysis}

We also compare the plug in estimator $f(\bar{X})$ and the IWAE one for the ICA model using the same setup of Table~\ref{tab:ica_appendix} where we set the number of samples for IWAE and the plug in estimator to $10$.

\begin{table}
\centering
\begin{tabular}{l|ccc}
Method & $\mse(\sigma)$ & $\mse(A)$ & runtime (s) \\
\hline\noalign{\smallskip}
SAEM & $7.00\cdot 10^{-7}$ & 0.22 & 33 \\
MLMC & $2.95\cdot 10^{-6}$ & 0.29 & 11\\
simple & $2.97\cdot 10^{-6}$ & 0.27 & 10 \\
cycling & $6.80\cdot 10^{-5}$ & 0.24 & 14 \\
IWAE & $3.55\cdot 10^{-6}$ & 0.30 & 7\\
Plug in & $3.62\cdot 10^{-6}$ & 0.27 & 7\\
\end{tabular}
\caption{Comparison of reconstruction accuracy and cost for the censored logistic ICA model. The results are averaged over 100 repetitions.}
\label{tab:ica_appendix}
\end{table} 

\subsubsection{Exponential random graphs}

To further showcase the strengths of the cycling estimator we also consider the ERG model. In particular, the setup of Table 2 for $\E[R]=10$. We compare the work-normalised-variance and the relative mean squared error 
($\mse$) for simple, cycling and the plug in estimator with same number of replicates $R$.
The relative $\mse$ is obtained by considering as true value of $1/\mathcal{Z}(\theta)$ that obtained by taking the reciprocal of 30,000 unbiased values of $\mathcal{Z}(\theta)$ obtained in the low variance regime (see Table~\ref{tab:ergm_appendix}). It is evident that in both a low variance regime and in a medium one cycling is preferable to the simple plug in estimator. In the low variance regime even the simple estimator is preferable to the plug in one.

\begin{table}
	\centering
	\begin{tabular}{l|cc|cc|cc}
&\multicolumn{2}{c}{Simple}   & \multicolumn{2}{c}{Cycling}                                                                                                    & \multicolumn{2}{c}{Plug in}                                                                                                                \\
		\hline\noalign{\smallskip}
&		              $\wnv$                       & $\mse$             & $\wnv$                       & $\mse$         & $\wnv$ & $\mse$ \\
		\hline\noalign{\smallskip}
		Moderate & 61 & 0.32 & 24 & 0.15 & 26 & 0.29\\
		Low & 83 &0.04  & 26 & 0.01& 105 &0.09  \\
	\end{tabular}

	\caption{Work-normalised variance and mean squared error
		for $10^3$ replicates of $\widehat{\mathcal{Z}(\theta)^{-1}}$ for
		the simple, the cycling estimator and the plug in estimator.}
	\label{tab:ergm_appendix}
\end{table}
%%% Local Variables:
%%% mode: latex
%%% TeX-master: "paper.tex"
%%% End:

\end{document}